\documentclass[3p]{elsarticle}
\usepackage{amsmath}
\usepackage{amssymb}
\usepackage{amsthm}
\usepackage{bm}
\usepackage{color}
\usepackage{enumerate}
\usepackage{geometry}
\usepackage{hyperref}
\usepackage{subfigure}
\usepackage[capitalize,nameinlink]{cleveref}

\renewcommand{\Re}{\operatorname{Re}}
\renewcommand{\Im}{\operatorname{Im}}
\renewcommand{\i}{\mathrm{i}}
\DeclareMathOperator{\sech}{sech}
\theoremstyle{definition}
	\newtheorem{definition}{Definition}[section]
	\newtheorem{remark}[definition]{Remark}

	\newtheorem{lemma}[definition]{Lemma}
	
	\newtheorem{theorem}[definition]{Theorem}
\theoremstyle{remark}

\begin{document}

\title{General soliton solutions to the coupled Hirota equation via the Kadomtsev-Petviashvili reduction}
\author[1]{Changyan Shi}
\author[1]{Bingyuan Liu}
\author[1]{Bao-Feng Feng\corref{cor1}}
\ead{baofeng.feng@utrgv.edu}
\cortext[cor1]{Corresponding author}
\affiliation[1]{organization={School of Mathematical and Statistical Sciences, The University of Texas Rio Grande Valley},
postcode={78539},
city={Edinburg, TX},
country={USA}}

\begin{abstract}
In this paper, we are concerned with various soliton solutions to the coupled Hirota  equation, as well as to the Sasa-Satsuma equation which can be viewed as one reduction case of the coupled Hirota  equation. First, we derive bright-bright, dark-dark, and bright-dark soliton solutions of the coupled Hirota equation by using the Kadomtsev-Petviashvili reduction method. Then, we present the bright and dark soliton solutions to the Sasa-Satsuma equation which are expressed by determinants of $N \times N$ instead of $2N \times 2N$ in the literature.
The dynamics of first-, second-order solutions are investigated in detail. It is intriguing that, for the SS equation, the bright soliton for \(N=1\) is also the soliton to the complex mKdV equation while the amplitude and velocity of dark soliton for \(N=1\) are determined by the background plane wave. For \(N=2\), the bright soliton can be classified into three types:  oscillating, single-hump, and two-hump ones while the dark soliton can be classified into five types:  dark (single-hole), anti-dark, Mexican hat, anti-Mexican hat and double-hole. Moreover, the types of bright solitons for the Sasa-Satsuma equation can be changed due to collision.   
\end{abstract}


\maketitle

\section{Introduction}
The nonlinear Schr\"odinger (NLS) equation and the Korteweg-de Vries (KdV) equation play a crucial role in the study of the soliton theory. From a physics point of view, the NLS equation and its soliton solution arises in the research to the propagation of light in nonlinear optical fibers and planar waveguides \cite{hasegawa1973transmission,fibich2015nonlinear}, the Bose-Einstein condensates \cite{dalfovo1999theory,pitaevskii2003bose}, the Langmuir waves in plasmas \cite{zakharov1972collapse,kato2005nonlinear} and water waves \cite{benney1967propagation}. The KdV equation models the shallow water waves \cite{korteweg1895xli} and finds its application in the Fermi-Pasta-Ulam-Tsingou problem \cite{benettin2013fermi}.  In addition to the NLS and the KdV equations, other integrable nonlinear evolution equations including the modified KdV (mKdV) equation, the Toda lattice \cite{toda1970waves}, and the Kadomtsev-Petviashvili (KP) equation \cite{kadomtsev1970stability} have been thoroughly studied in the literature.
Meanwhile, various methods such as the Darboux transformation \cite{matveev1991darboux}, Hirota's bilinear method \cite{hirota2004direct}, and the KP reduction method \cite{jimbo1983solitons, sato1989kp,ohta1993casorati,feng2014general} have been developed to find exact solutions to the soliton equations.

As one generalization of the NLS equation with higher orders \cite{agrawal2000nonlinear}, we have
\begin{align}
    \i \frac{\partial \psi}{\partial T} + 3 \i \alpha |\psi|^2 \frac{\partial \psi}{\partial X} + \beta \frac{\partial^2 \psi}{\partial X^2} + \i \gamma \frac{\partial^3 \psi}{\partial X^3} + \delta |\psi|^2 \psi = 0 \label{hirota_eq}
\end{align}
where \(\alpha,\ \beta, \ \gamma, \ \delta\) are real parameters account for the self-steepening effect, the group velocity dispersion, third-order dispersion and the self-phase modulation, respectively. It is shown by Hirota  \cite{hirota1973exact} to be integrable if  \(\alpha \beta = \gamma \delta\), thus, it is named the Hirota equation under this condition. 
By the following transformation on \eqref{hirota_eq}
\begin{align*}
    \psi(X, T) = \exp\bigl(\i (c_1 (x - c_3) t + c_2 t)\bigr) u(x, t), \quad x = T + c_3 X, \quad t = T,
\end{align*}
where \(c_1 = \delta/3\alpha = \beta /3 \gamma\),  \(c_2 = -2c_1^3 \gamma\), \(c_3 = 3c_1^2 \gamma\), \cref{hirota_eq} can be transformed to the complex mKdV equation \cite{mahalingam2001propagation}
\begin{align}\label{hirota_cmkdv_form}
    u_t + 3 \alpha |u|^2 u_x + \gamma u_{xxx} = 0.
\end{align}

In the present paper, we study the following coupled Hirota (cHirota) equation \cite{tasgal1992soliton, porsezian1994coupled, gilson2003sasa}
\begin{align}
    u_{1,t}&=u_{1,xxx} - 3 \left(\varepsilon_1|u_1|^2 + \varepsilon_2|u_2|^2\right) u_{1,x} - 3 u_1 \left(\varepsilon_1 u_1^*u_{1,x} + \varepsilon_2 u_2^*u_{2,x}\right),\label{chirota_1}\\
    u_{2,t}&=u_{2,xxx} - 3 \left(\varepsilon_1|u_1|^2 + \varepsilon_2|u_2|^2\right) u_{2,x} - 3 u_2 \left(\varepsilon_1 u_1^*u_{1,x} + \varepsilon_2 u_2^*u_{2,x}\right),\label{chirota_2}
\end{align}
where \(\varepsilon_1\) and \(\varepsilon_2\) are real numbers and asterisk \(*\) denotes the complex conjugate. 
The cHirota equation can be modeled for optical pulses with polarizations in birefringent fiber \cite{porsezian1997optical}. 
Depending on the boundary conditions applied to each component of \eqref{chirota_1}-\eqref{chirota_2}, three distinct cases arise in the above equation.
\begin{enumerate}[i)]  
    \item When enforcing the nonzero boundary condition, i.e., requiring \(u_1, u_2 \to \rho > 0\) as \(x \to \pm \infty\), solutions such as the dark soliton \cite{bindu2001dark, jiang2023asymptotic}, breather \cite{chai2019localized, pan2024super}, and rogue wave \cite{chen2013rogue, chen2014dark, huang2016rational, chan2017rogue} can be derived.  
    \item When imposing the zero boundary condition, i.e., setting \(u_1, u_2 \to 0\) as \(x \to \pm \infty\), the resulting solutions include the bright soliton \cite{kang2019construction, xie2020elastic, wang2021analytical, wang2021dressing} and the degenerate soliton \cite{monisha2023degenerate}.  
    \item When applying a mixed boundary condition, where one component satisfies a zero boundary condition while the other follows a nonzero boundary condition. To be specific, we take \(u_1 \to 0, u_2 \to \rho\) (or vice versa) as \(x \to \pm \infty\). Solutions such as the bright-dark soliton \cite{wang2014generalized, liu2016bright, liu2016mixed} and the bright-dark rogue wave \cite{wang2014rogue} can be derived.  
\end{enumerate}

As pointed in \cite{gilson2003sasa} that the Sasa-Satsuma (SS) equation \cite{sasa1991new}
\begin{equation}\label{sse}
    u_t = u_{xxx} - 6\varepsilon |u|^2 u_x - 3\varepsilon u \left(|u|^2\right)_x
\end{equation}
can be reduced from the cHirota equation: by employing reduction conditions \(u_1^* = u_2\) and \(\varepsilon = \varepsilon_1 = \varepsilon_2\). 
The SS equation was studied as a higher order extension to the NLS equation with third-order dispersion, self-frequency shift, and self-steepening included \cite{kodama1987nonlinear,trippenbach1998effects}. 
Its bright soliton \cite{mihalache1993inverse, gilson2003sasa}, dark soliton \cite{ohta2010dark}, breather \cite{wu2022multi}, and rouge wave \cite{bandelow2012sasa, chen2013twisted, akhmediev2015rogue, feng2022higher} solutions have been studied.
{Prior research  on the SS equation \eqref{sse} starts with a C-type reduction, resulting a \(2N\times 2N\) or \((2N+1) \times (2N+1)\) determinant form for the bright and dark soliton solutions \cite{gilson2003sasa,ohta2010dark}. The reduction of the soliton solutions of the cHirota equation \eqref{chirota_1}-\eqref{chirota_2} to the ones of the SS equation \eqref{sse} enables us to give a \(N\times N\) or \((N+1)\times (N+1)\) determinant form of soliton solutions and completes the soliton solution of \eqref{sse} when \(N\) is an odd number.}

The cHirota equation \eqref{chirota_1}-\eqref{chirota_2} is integrable  \cite{park2000higher, nakkeeran2000exact, bindu2001dark, sakovich2000symmetrically} with the Lax pair
$$\bm{\Phi}_x=U\bm{\Phi},\ \bm{\Phi}_t=V\bm{\Phi},
$$ 
with \(\bm{\Phi} = (\phi_1, \phi_2)^T\) and
\begin{align}
    \begin{split}\label{vcmKdV_Lax_pair}
        &U = \i \lambda \sigma + Q, \quad V = - \i \lambda^3 \sigma  - \lambda^2 Q - \i \left(Q_x + Q^2\right) \sigma_2 + (Q_{xx} - 2 Q^3 - \left[Q_x, Q\right]), \\
        & Q = \begin{pmatrix}
            0 & \sqrt{\varepsilon_1} u_1 & \sqrt{\varepsilon_2} u_2 \\
            \sqrt{\varepsilon_1} u_1^* & 0 & 0 \\
            \sqrt{\varepsilon_2} u_2^* & 0 & 0 \\
        \end{pmatrix}, \quad
        \sigma = \begin{pmatrix}
            -1 & 0 & 0 \\
            0 & 0 & 0 \\
            0 & 0 & 0 \\
        \end{pmatrix}, \quad
        \sigma_2 = \begin{pmatrix}
            -1 & 0 & 0 \\
            0 & 1 & 0 \\
            0 & 0 & 1 \\
        \end{pmatrix}.
    \end{split}
\end{align}
By using above Lax pair, bright and dark soliton solutions to the cHirota equation were derived by the inverse scatting transformation \cite{tasgal1992soliton}, Darboux transformation \cite{bindu2001dark, wang2014generalized,jiang2023asymptotic, xie2020elastic}, Riemann-Hilbert method \cite{wu2019riemann, kang2019construction}, and \(\bar{\partial}\)-dressing method \cite{wang2021dressing}. { In particular, the general \(N\)-dark-dark and \(N\)-bright-bright soliton solutions to \eqref{chirota_1}-\eqref{chirota_2} were given in \cite{jiang2023asymptotic,wang2021dressing}. However, it remains an open problem to derive the general \(N\)-bright-dark soliton solution. Furthermore, as the KP reduction method have not been applied to the cHirota equation, we aim at deriving all general \(N\)-soliton solution to \eqref{chirota_1}-\eqref{chirota_2} using this approach. The resulting solutions can further be utilized to provide bright and dark soliton solutions to the SS equation, as mentioned earlier.}

The remainder of the present paper is structured as follows. In \cref{seccion:bb_chirota,section:dd_chirota,section:bd_chirota}, bright-bright, dark-dark, and bright-dark soliton solutions of the cHirota equation \eqref{chirota_1}-\eqref{chirota_2} are derived using the KP hierarchy reduction method. The dynamical behavior of the first- and second-order solutions is illustrated and discussed.  In  \cref{section:b_sasa,section:d_sasa}, the bright and dark soliton solutions are obtained for the SS equation \eqref{sse} from the counterpart solutions of the system \eqref{chirota_1}-\eqref{chirota_2}. Moreover,  the conditions for the oscillated, single-hump, and two-hump soliton solutions are obtained for the SS equation \eqref{sse}. Whereas, a classification of single dark soliton: single-hole, anti-dark, Mexican hat, anti-Mexican, and double-hole of \eqref{sse} is provided.

\section{Bright-bright soliton solution to the coupled Hirota equation}\label{seccion:bb_chirota}
In this section, we consider the bright-bright soliton solution to the cHirota equation \eqref{chirota_1}-\eqref{chirota_2},  which is given by the following theorem.

\begin{theorem}\label{theorem:bb_solution}
    Equation \eqref{chirota_1}-\eqref{chirota_2} admits the bright soliton solutions given by \(u_1=g_1/f, \ u_2=g_2/f\) with \(f,\ g_1, \ g_2\) defined as
    \begin{align}\label{def_bright_f_g}
        f=|M|,\quad
        g_1=\begin{vmatrix}
            M & \Phi \\ -\left(\bar{\Psi}\right)^T & 0
            \end{vmatrix},\quad
        g_2=\begin{vmatrix}
            M & \Phi \\ -\left(\bar{\Upsilon}\right)^T & 0
            \end{vmatrix},
    \end{align}
    where \(M\) is \(N\times N\) matrix, \(\Phi \), \(\bar{\Psi} \), are \(N\)-component vectors whose elements are defined respectively as
    \begin{align}
        &m_{ij}=\frac{1}{p_i+p_j^*}\left(e^{\xi_i+\xi_j^*} - \varepsilon_1 C_i^* C_j - \varepsilon_2 D_i^* D_j \right),\quad \xi_i=p_i x + p_i^3 t+\xi_{i0},\\
        &\Phi = \left( e^{\xi_{1}},e^{\xi_2},\ldots, e^{\xi_N}\right)^T,\quad
        \bar{\Psi}=\left(C_1, C_2,\ldots, C_N\right)^T, \quad \bar{\Upsilon}=\left(D_1, D_2,\ldots, D_N\right)^T.
    \end{align}
    Here, \(p_i\), \(\xi_{i0}\), \(C_i\), \(D_i\) are complex parameters.
\end{theorem}
In the following subsection, we will give the proof of the above theorem.
    
\subsection{Derivation of bright-bright soliton solution to the coupled Hirota equation}
A substitution of \(u_1=g_1/f,\ u_2=g_2/f\) into the system \eqref{chirota_1}-\eqref{chirota_2}, a set of four
bilinear equations to the cHirota equation is obtained 
\begin{align}
    &(D_x^3 - D_t)g_1 \cdot f= - 3 \varepsilon_1 s_{12} g_2^*,\label{2cmKdVSF_BL_1}\\
    &(D_x^3 - D_t)g_2 \cdot f= - 3 \varepsilon_2 s_{21} g_1^*,\label{2cmKdVSF_BL_2}\\
    &D_x^2 f\cdot f + 2 \varepsilon_1 |g_1|^2 + 2 \varepsilon_2 |g_2|^2=0,\label{2cmKdVSF_BL_3}\\
    &D_x g_1 \cdot g_2 = s_{12} f.\label{2cmKdVSF_BL_4}
\end{align}
by introducing auxiliary functions \(s_{12}, s_{21}\) where  
we assume \(s_{12} = -s_{21}\).


To derive the bright-bright soliton solutions, we start with the \(\tau\)-functions in three-dimensional KP-Toda hierarchy which is given by  
\begin{align}
    &f=|M|,\quad
    g_1=\begin{vmatrix}
        M & \Phi \\ -\bar{\Psi}^T & 0
    \end{vmatrix},
    \quad 
    \bar{g}_1=\begin{vmatrix}
        M & \Psi \\ -\bar{\Phi}^T & 0
    \end{vmatrix},\quad
    g_2=\begin{vmatrix}
        M & \Phi \\ -\bar{\Upsilon}^T & 0
    \end{vmatrix},
    \quad 
    \bar{g}_2=\begin{vmatrix}
        M & \Upsilon \\ -\bar{\Phi}^T & 0
    \end{vmatrix},\label{tau_bright_begin}
    \\
    & s_{12}=\begin{vmatrix}
        M & \Phi & \partial_{x_1} \Phi \\
        -\bar{\Upsilon}^T & 0 & 0 \\
        -\bar{\Psi}^T & 0 & 0 
    \end{vmatrix},\label{tau_bright_end}
\end{align}
where \(M\) is \(N\times N\) matrix, \(\Phi \), \(\bar{\Phi}\), \(\Psi\), \(\bar{\Psi}\), \(\Upsilon\) and \(\bar{\Upsilon}\) are \(N\)-component vectors whose elements are defined respectively as
\begin{align}
    &m_{ij}=\frac{1}{p_i+\bar{p}_j}e^{\xi_i+\bar{\xi}_j}+\frac{\tilde{C}_i \bar{C}_j}{q_i+\bar{q}_j}e^{\eta_i+\bar{\eta}_j}+\frac{\tilde{D}_i \bar{D}_j}{r_i+\bar{r}_j}e^{\chi_i+\bar{\chi}_j},\\
    &\Phi =\left( e^{\xi _{1}},e^{\xi _2},\ldots, e^{\xi _{N}}\right)^T, \quad \bar{\Phi} =\left( e^{\bar{\xi}_{1}},e^{\bar{\xi}_2},\ldots ,e^{\bar{\xi}_{N}}\right)^T, \quad\\
    &\Psi =\left(\tilde{C}_1 e^{\eta _{1}}, \tilde{C}_2 e^{\eta _2},\ldots, \tilde{C}_N e^{\eta _{N}}\right)^T,\quad
    \bar{\Psi}=\left(\bar{C}_1 e^{\bar{\eta}_1}, \bar{C}_2 e^{\bar{\eta}_2},\ldots, \bar{C}_N e^{\bar{\eta}_{N}}\right)^T,\\
    &\Upsilon =\left(\tilde{D}_1 e^{\chi_1}, \tilde{D}_2 e^{\chi_2},\ldots, \tilde{D}_N e^{\chi_N}\right)^T,\quad
    \bar{\Upsilon}=\left(\bar{D}_1 e^{\bar{\chi}_1}, \bar{D}_2 e^{\bar{\chi}_2},\ldots, \bar{D}_N e^{\bar{\chi}_{N}}\right)^T,\\
    &\xi_i=p_i x_1 + p_i^2 x_2 + p_i^3 x_3+\xi_{i0},\quad \bar{\xi}_i=\bar{p}_i x_1 - \bar{p}_i^2 x_2 + \bar{p}_i^3 x_3+\bar{\xi}_{i0},\\
    &\eta_i=q_i y_1,\quad \bar{\eta}_i=\bar{q}_i y_1, \quad \chi_i=r_i z_1,\quad \bar{\chi}_i=\bar{r}_i z_1. 
\end{align}
It is shown that the above \( \tau\)-functions satisfy 
the following bilinear equations \cite{gilson2003sasa}
\begin{align}
    &\left(D_{x_1}^3+3D_{x_1}D_{x_2}-4D_{x_3}\right)g_1 \cdot f=0,\label{KP_bright_1}\\
    &\left(D_{x_1}^3+3D_{x_1}D_{x_2}-4D_{x_3}\right)g_2 \cdot f=0,\label{KP_bright_2}\\
    &D_{y_1}D_{x_1}f\cdot f=-2g_1 \bar{g}_1,\label{KP_bright_3}\\
    &D_{z_1}D_{x_1}f\cdot f=-2g_2 \bar{g}_2,\label{KP_bright_4}\\
    &D_{x_1}g_1\cdot g_2=s_{12}f,\label{KP_bright_5}\\
    &D_{y_1}\left(D_{x_1}^2-D_{x_2}\right)g_1 \cdot f=0,\label{KP_bright_6}\\
    &D_{z_1}\left(D_{x_1}^2-D_{x_2}\right)g_1 \cdot f=-4s_{12}\bar{g}_2,\label{KP_bright_7}\\
    &D_{y_1}\left(D_{x_1}^2-D_{x_2}\right)g_2 \cdot f=-4s_{21}\bar{g}_1.\label{KP_bright_8}\\
    &D_{z_1}\left(D_{x_1}^2-D_{x_2}\right)g_2 \cdot f=0.\label{KP_bright_9}
\end{align}
\begin{remark}\label{remark:KP_bright_relation}
    With above defined \(\tau\)-function,  \(s_{21}\)  is defined from \(s_{12}\) by exchanging the last two rows in \(s_{12}\) such that \(s_{12}=-s_{21}\).
\end{remark}

Next, we perform reductions from the KP-Toda hierarchy \eqref{KP_bright_1}-\eqref{KP_bright_7} to the bilinear form of the cHirota equation \eqref{2cmKdVSF_BL_1}-\eqref{2cmKdVSF_BL_4}. 
Firstly we consider the dimension reduction to eliminate the variables \(x_2,\ y_1, \ z_1\).
To achieve this, let us introduce the following lemma
\begin{lemma}
Under the restrictions
\begin{align}
    p_i = - \varepsilon_1 q_i = - \varepsilon_2 r_i, \quad  \bar{p}_i = - \varepsilon_1 \bar{q}_i = - \varepsilon_2 \bar{r}_i,
\end{align}
where \(i=1,2,\ldots, N\). \(\tau\)-functions \eqref{tau_bright_begin}-\eqref{tau_bright_end} satisfy the following relation
\begin{align}\label{bright_dim_reduction}
    \varepsilon_1\partial_{y_1}+\varepsilon_2\partial_{z_1}=\partial_{x_1}.
\end{align}
\end{lemma}
\begin{proof}
Consider \(\tau\)-function \(f\), note that 
\begin{align*}
    f&=\begin{vmatrix}
        \dfrac{1}{p_i+\bar{p}_j}e^{\xi_i+\bar{\xi}_j}+\dfrac{\tilde{C}_i \bar{C}_j}{q_i+\bar{q}_j}e^{\eta_i+\bar{\eta}_j}+\dfrac{\tilde{D}_i \bar{D}_j}{r_i+\bar{r}_j}e^{\chi_i+\bar{\chi}_j}
    \end{vmatrix}\\
    &=\left(\prod_{i=1}^{N}e^{\xi_i}\right)\left(\prod_{i=1}^{N}e^{\bar{\xi}_i}\right) \begin{vmatrix}
        \dfrac{1}{p_i+\bar{p}_j}+\dfrac{\tilde{C}_i \bar{C}_j}{q_i+\bar{q}_j}e^{\eta_i+\bar{\eta}_j-\xi_i-\bar{\xi}_j}+\dfrac{\tilde{D}_i \bar{D}_j}{r_i+\bar{r}_j}e^{\chi_i+\bar{\chi}_j-\xi_i-\bar{\xi}_j}
    \end{vmatrix},
\end{align*}
the term \((\prod_{i=1}^{N}e^{\xi_i})(\prod_{i=1}^{N}e^{\bar{\xi}_i})\) can be dropped due to property of the \(D\)-operator, thus
\begin{align*}
    \left(\varepsilon_1\partial_{y_1}+\varepsilon_2\partial_{z_1}\right)f&=\sum_{i,j=0}^{N} \Delta_{ij} \left(\varepsilon_1\partial_{y_1}+\varepsilon_2\partial_{z_1}\right)\left[\dfrac{1}{p_i+\bar{p}_j}+\dfrac{\tilde{C}_i \bar{C}_j}{q_i+\bar{q}_j}e^{\eta_i+\bar{\eta}_j-\xi_i-\bar{\xi}_j}+\dfrac{\tilde{D}_i \bar{D}_j}{r_i+\bar{r}_j}e^{\chi_i+\bar{\chi}_j-\xi_i-\bar{\xi}_j}\right]\\
    &=\sum_{i,j=0}^{N} \Delta_{ij} \left[\varepsilon_1(q_i+\bar{q}_j)\dfrac{\tilde{C}_i \bar{C}_j}{q_i+\bar{q}_j}e^{\eta_i+\bar{\eta}_j-\xi_i-\bar{\xi}_j}+\varepsilon_2(r_i+\bar{r}_j)\dfrac{\tilde{D}_i \bar{D}_j}{r_i+\bar{r}_j}e^{\chi_i+\bar{\chi}_j-\xi_i-\bar{\xi}_j}\right]
\end{align*}
where \(\Delta_{ij}\) is the cofactor of the \((i,j)\) entry of \(f\). Also, note 
\begin{align*}
    \partial_{x_1}f&=\sum_{i,j=0}^{N} \Delta_{ij} \partial_{x_1}\left[\dfrac{1}{p_i+\bar{p}_j}+\dfrac{\tilde{C}_i \bar{C}_j}{q_i+\bar{q}_j}e^{\eta_i+\bar{\eta}_j-\xi_i-\bar{\xi}_j}+\dfrac{\tilde{D}_i \bar{D}_j}{r_i+\bar{r}_j}e^{\chi_i+\bar{\chi}_j-\xi_i-\bar{\xi}_j}\right]\\
    &=\sum_{i,j=0}^{N} \Delta_{ij} \left[(-p_i-\bar{p}_j)\dfrac{\tilde{C}_i \bar{C}_j}{q_i+\bar{q}_j}e^{\eta_i+\bar{\eta}_j-\xi_i-\bar{\xi}_j}+(-p_i-\bar{p}_j)\dfrac{\tilde{D}_i \bar{D}_j}{r_i+\bar{r}_j}e^{\chi_i+\bar{\chi}_j-\xi_i-\bar{\xi}_j}\right],
\end{align*}
while \(p_i = - \varepsilon_1 p_i = - \varepsilon_2 r_i, \ \bar{p}_i = - \varepsilon_1 \bar{q}_i = - \varepsilon_2 \bar{r}_i\), the relation \eqref{bright_dim_reduction} holds. 
One can prove this for \(g_1,\ \bar{g}_1,\ g_2,\ \bar{g}_2\) and \(s_{12}\) analogously.
\end{proof}

With the above dimension reduction, one obtains
\begin{align}
    &D_{x_1}^2 f\cdot f = -2\varepsilon_1 g_1 \bar{g_{1}} -2\varepsilon_2 g_2 \bar{g_{2}}, \label{KP_bright_3.1}\\
    &\left(D_{x_1}^3-D_{x_1}D_{x_2}\right)g_1 \cdot f = -4\varepsilon_2 s_{12}\bar{g_2},\label{KP_bright_6.1}\\
    &\left(D_{x_1}^3-D_{x_1}D_{x_2}\right)g_2 \cdot f = -4\varepsilon_1 s_{21}\bar{g_1}\label{KP_bright_8.1}
\end{align}
from \eqref{KP_bright_3}-\eqref{KP_bright_4}, \eqref{KP_bright_6}-\eqref{KP_bright_7} and \eqref{KP_bright_8}-\eqref{KP_bright_9}, respectively. 
The elimination of \(D_{x_2}\) by using (\eqref{KP_bright_1}) and \ (\eqref{KP_bright_2}) leads to
\begin{align}
    \left(D_{x_1}^3-D_{x_3}\right)g_1 \cdot f = -3\varepsilon_2 s_{12}\bar{g_2},\label{KP_bright_1.2}\\
    \left(D_{x_1}^3-D_{x_3}\right)g_2 \cdot f = -3\varepsilon_1 s_{21}\bar{g_1}.\label{KP_bright_2.2}
\end{align}
Hence, by performing the dimension reductions mentioned above, the variables \(x_2,\ y_1,\ z_1\) are reduced into dummy variables, and we can set \(x_2=y_1=z_1=0,\ x_1 = x,\  x_3 = t\). 

Furthermore, to achieve complex conjugate condition, we have the following lemma  
\begin{lemma}
Under following parameter restrictions
\begin{equation}
    \bar{p}_i^*=p_i,\quad \bar{\xi}_{i,0}^*=\xi_{i0}, \quad C_i = \left(\tilde{C}_i\right)^* = \bar{C}_i, \quad D_i = \left(\tilde{D}_i\right)^* = \bar{D}_i
\end{equation}
where \(i=1,2,\ldots,N\), we have
\begin{align}
    \bar{g}_1=g_1^*,\quad \bar{g}_2=g_2^*, \quad f\in \mathbb{R}.
\end{align}
\end{lemma}

\begin{proof}
With these restrictions, we have
\begin{align*}
    &\bar{\xi}_i=\bar{p}_i x+\bar{p}_i^3 t+\bar{\xi}_{i,0}=p_i^* x+(p_i^*)^3 t+\xi_{i,0}^*=\xi_i^*
\end{align*}
thus
\begin{align*}
    m_{i,j}^*&=\frac{1}{p_i^*+\bar{p}_j^*} \left(e^{\xi_i^*+\bar{\xi}_j^*} - \varepsilon_1\tilde{C}_i^* \bar{C}_j^* - \varepsilon_2 \tilde{D}_i^* \bar{D}_j^*\right) = \frac{1}{\bar{p}_i + p_j} \left(e^{\bar{\xi}_i + \xi_j} - \varepsilon_1 C_i C_j^* + - \varepsilon_2 D_i D_j^* \right) \\
    &=\frac{1}{p_j+\bar{p}_i} \left(e^{\xi_j+\bar{\xi}_i} - \varepsilon_1\tilde{C}_j \bar{C}_i - \varepsilon_2 \tilde{D}_j \bar{D}_i \right) =m_{j,i}.
\end{align*}
Based on this, we can prove \(\bar{g}_1=g_1^*\) as 
\begin{align*}
    \bar{g}_1=\begin{vmatrix}
        m_{i,j} & \tilde{C}_i \\ -e^{\bar{\xi}_j} & 0
    \end{vmatrix} = \begin{vmatrix}
        m_{i,j} & -\tilde{C}_i \\ e^{\bar{\xi}_j} & 0
    \end{vmatrix} = \begin{vmatrix}
        m_{j,i} & e^{\bar{\xi}_i} \\ -\tilde{C}_j & 0
    \end{vmatrix} = \begin{vmatrix}
        m_{i,j}^* &  e^{\xi_i^*} \\ -\left(C_j\right)^* & 0
    \end{vmatrix} = \begin{vmatrix}
        m_{i,j}^* &  e^{\xi_i^*} \\ -\left(\bar{C}_j\right)^* & 0
    \end{vmatrix} = g_1^*.
\end{align*}
Similarly,  \(\bar{g}_2=g_2^*\) can be approved. Moreover, note that
\begin{align*}
    f = \det\left(m_{i,j}\right) = \det\left(m_{j,i}\right) = \det\left(\left(m_{i,j}\right)^*\right) = f^*
\end{align*}
which implies \(f\in \mathbb{R}\).
\end{proof}

Under above complex conjugate reduction, \eqref{KP_bright_1.2}, \eqref{KP_bright_2.2}, \eqref{KP_bright_3.1}, and \eqref{KP_bright_4} become nothing but the bilinear form of the cHirota equation \eqref{2cmKdVSF_BL_1}-\eqref{2cmKdVSF_BL_4}. This complete the proof of \cref{theorem:bb_solution}.

\subsection{Dynamics of the bright-bright soliton solution to the coupled Hirota equation}
In this subsection, we are going to explore the dynamics of the afore-obtained bright-bright soliton solutions.

By taking \(N=1\) in \cref{theorem:bb_solution}, for solution \(u_1 = g_1 /f,\ u_2 = g_2 /f\), we have 
\begin{align*}
    &f = \frac{1}{p_1+p_1^*} \left( \exp \left(\xi_1 + \xi_1^*\right) - \varepsilon_1 |C_1|^2 - \varepsilon_2 |D_1|^2\right),\quad g_1 = C_1 \exp \left(\xi_1\right), \quad g_2 = D_1 \exp \left(\xi_1\right),\\
    &\xi_1 = p_1 x +  p_1^3 t + \xi_{1,0}.
\end{align*}
The regularity condition for this solution is \(f \neq 0\), one may deduce that it requires \(- \varepsilon_1 |C_1|^2 - \varepsilon_2 |D_1|^2 > 0\), which implies there is no bright soliton solution for \( \varepsilon_1=\varepsilon_2=1\). Above solution can be simplified as
\begin{align*}
    u_1 &= C_1 \Re(p_1) \frac{\exp\left(\i \Im (\xi_{1})\right)}{\sqrt{-\varepsilon_1|C_1|^2 -\varepsilon_2|D_1|^2}} \sech\left(\Re (\xi_1) - \log\left(\sqrt{-\varepsilon_1|C_1|^2 -\varepsilon_2|D_1|^2}\right)\right), \\
    u_2 &= D_1 \Re(p_1) \frac{\exp\left(\i \Im (\xi_{1})\right)}{\sqrt{-\varepsilon_1|C_1|^2 -\varepsilon_2|D_1|^2}} \sech\left(\Re (\xi_1) - \log\left(\sqrt{-\varepsilon_1|C_1|^2 -\varepsilon_2|D_1|^2}\right)\right).
\end{align*}
The above expressions indicate that the solitons from \(u_1\) and \(u_2\) are both located along the line \(\Re (\xi_1) - \log\left(\sqrt{-\varepsilon_1|C_1|^2 -\varepsilon_2|D_1|^2}\right) = 0\). Their maximum amplitudes are \[C_1 \Re(p_1)/\sqrt{-\varepsilon_1|C_1|^2 -\varepsilon_2|D_1|^2}, \quad D_1 \Re(p_1)/\sqrt{-\varepsilon_1|C_1|^2 -\varepsilon_2|D_1|^2}\] respectively.  Since \(\Re (\xi_1) = 2 \Re(p_1) \left(x - \left(3 \left[\Im(p_1)\right]^2 - \left[\Re(p_1)\right]^2\right) t \right) + \Re (\xi_{1,0})\), the speed of solitons from \(u_1\) and \(u_2\) are both \(3 \left[\Im(p_1)\right]^2 - \left[\Re(p_1)\right]^2\).

We can also consider the intensity of above solution, note that
\begin{align*}
    \int_{-\infty}^{+\infty} \sech^2 \left(\Re (\xi_1) - \log\left(\sqrt{-\varepsilon_1|C_1|^2 -\varepsilon_2|D_1|^2}\right)\right) \, dx = \frac{2}{|\Re(p_1)|},
\end{align*}
therefore, the intensity for \(u_1\), \(u_2\) are
\begin{align*}
    N(u_1) = \int_{-\infty}^{+\infty} |u_1|^2 \, dx = \frac{2|\Re (p_1)||C_1|^2}{-\varepsilon_1|C_1|^2 -\varepsilon_2|D_1|^2},\quad  N(u_2) = \int_{-\infty}^{+\infty} |u_2|^2 \, dx = \frac{2|\Re (p_1)||D_1|^2}{-\varepsilon_1|C_1|^2 -\varepsilon_2|D_1|^2}.
\end{align*}
Note that for \(-\varepsilon_1N(u_1) -\varepsilon_2N(u_2) = 2|\Re(p_1)|\), therefore when \(\varepsilon_1=\varepsilon_2=-1\), the total intensity for all components is \(N = N(u_1) + N(u_2) = 2|\Re(p_1)|\). Since the amplitude of the soliton is proportional to \(C_1\) and \(D_1\), one may observe soliton with different height on each component, see \cref{figure:bb_N=1}. 
\begin{figure}[!ht]
    \centering
    \subfigure[]{
        \includegraphics[width=60mm]{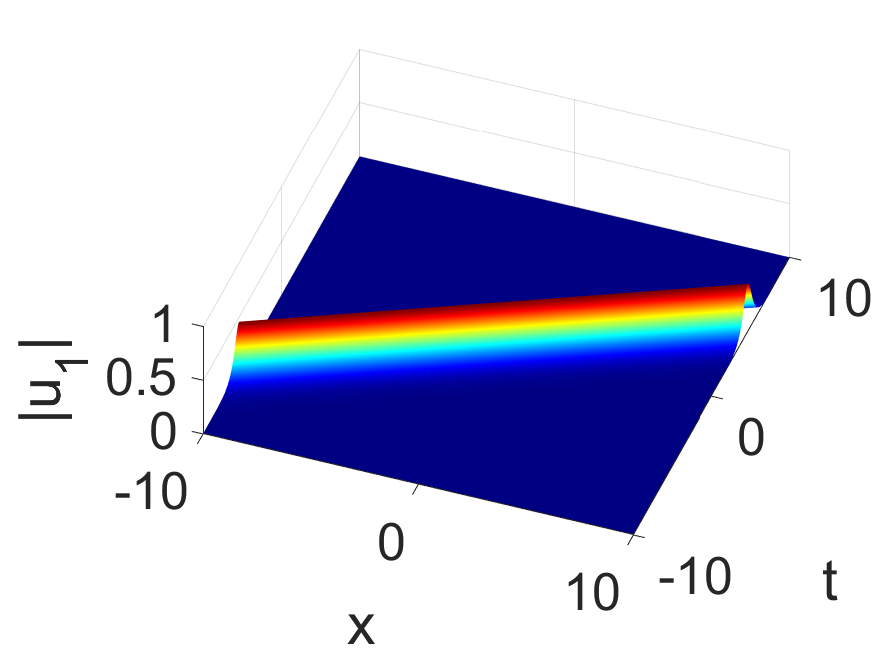}
    }
    \subfigure[]{
        \includegraphics[width=60mm]{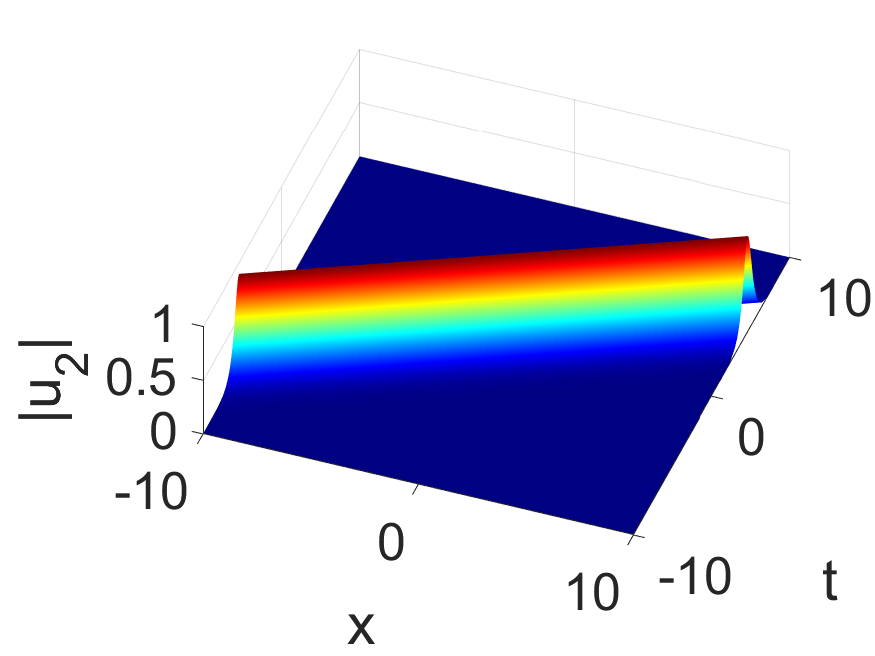}
    }
    \caption{One-bright soliton solution to the cHirota equation with parameters \(p_1=1+\i, C_1=1, D_1=2, \xi_{1,0} = 0, \varepsilon_1 = \varepsilon_2 = -1\).}\label{figure:bb_N=1}
\end{figure}

The two-bright soliton is given by \cref{theorem:bb_solution} for \(N=2\) 
\begin{align*}
    & f = \det(M), \quad
    g_1 = \begin{vmatrix}
        M & \begin{matrix} \exp(\xi_1) \\ \exp(\xi_2) \end{matrix} \\
        \begin{matrix} -C_1 & -C_2 \end{matrix} &  0
    \end{vmatrix}, \quad
    g_2 = \begin{vmatrix}
        M & \begin{matrix} \exp(\xi_1) \\ \exp(\xi_2) \end{matrix} \\
        \begin{matrix} -D_1 & -D_2 \end{matrix} &  0
    \end{vmatrix}, \\
    &  M = \begin{pmatrix}
        \dfrac{1}{p_1 + p_1^*} \left(e^{\xi_1 + \xi_1^*} + c_{1,1}\right) &
        \dfrac{1}{p_1 + p_2^*} \left(e^{\xi_1 + \xi_2^*} + c_{1,2}\right) \\
        \dfrac{1}{p_2 + p_1^*} \left(e^{\xi_2 + \xi_1^*} + c_{1,2}^*\right) &
        \dfrac{1}{p_2 + p_2^*} \left(e^{\xi_2 + \xi_2^*} + c_{2,2}\right) 
    \end{pmatrix}\,.
\end{align*}
where 
\begin{align*}
    c_{1,1} = -\varepsilon_1|C_1|^2 - \varepsilon_2|D_1|^2,\quad c_{1,2} = -\varepsilon_1C_1^* C_2  - \varepsilon_2D_1^* D_2,\quad c_{2,2} = -\varepsilon_1|C_2|^2 - \varepsilon_2|D_2|^2, 
\end{align*}
\(\xi_1 = p_1 x + p_1^3 t + \xi_{1,0}, \ \xi_2 = p_2 x + p_2^3 t + \xi_{2,0}\). 
Let us perform asymptotic analysis to the two-solution solution.
The dynamical behavior of second order bright soliton can be illustrated by the asymptotic analysis. To start, we name soliton determined by \(\xi_1 + \xi_1^*\) as soliton 1, and soliton 2 is from \(\xi_2 + \xi_2^*\). Assume soliton 1 is on the right of soliton 2 before the collision. Next, consider the solution \(u_1, u_2\) under \(t \rightarrow \pm \infty\):
\begin{enumerate}[(1)]
    \item Before collision, i.e., \(t \rightarrow - \infty\), this gives
    
    Soliton 1 (\(\xi_1 + \xi_1^* \approx 0,\ \xi_2 + \xi_2^* \rightarrow - \infty\))
    \begin{align*}
        u_1 &\simeq \frac{(p_1+p_1^*)(p_2+p_2^*)}{c_{2,2}} \left(\frac{C_1 c_{2,2}}{p_2 + p_2^*} - \frac{C_2 c_{1,2}^*}{p_2 + p_1^*}\right) \frac{\exp(\i \Im(\xi_1))}{2 \theta_1} \sech \left(\Re (\xi_1) - \log\left(\theta_1\right)\right), \\
        u_2 &\simeq \frac{(p_1+p_1^*)(p_2+p_2^*)}{c_{2,2}} \left(\frac{D_1 c_{2,2}}{p_2 + p_2^*} - \frac{D_2 c_{1,2}^*}{p_2 + p_1^*}\right) \frac{\exp(\i \Im(\xi_1))}{2 \theta_1} \sech \left(\Re (\xi_1) - \log\left(\theta_1\right)\right),
    \end{align*}
    where the phase term \(\theta_1\) is given by
    \begin{align*}
        \theta_1 = \sqrt{c_{1,1} - \frac{(p_1+p_1^*)(p_2+p_2^*)|c_{1,2}|^2}{|p_2+p_1^*|^2 c_{2,2}}}.
    \end{align*}

    Soliton 2 (\(\xi_2 + \xi_2^* \approx 0,\ \xi_1 + \xi_1^* \rightarrow + \infty\))
    this can be simplified as
    \begin{align*}
        u_1 &\simeq C_2 \frac{(p_2+p_2^*)(p_2 - p_1)}{(p_2+p_1^*)c_{2,2}} \frac{\exp(\i \Im(\xi_2))}{2\varphi_1} \sech \left(\Re (\xi_2) + \log\left(\varphi_1\right)\right), \\
        u_2 &\simeq D_2 \frac{(p_2+p_2^*)(p_2 - p_1)}{(p_2+p_1^*)c_{2,2}} \frac{\exp(\i \Im(\xi_2))}{2\varphi_1} \sech \left(\Re (\xi_2) + \log\left(\varphi_1\right)\right),
    \end{align*}
    where the phase term \(\varphi_1\) is given by
    \begin{align*}
        \varphi_1 = \sqrt{\frac{1}{c_{2,2}} - \frac{(p_1+p_1^*)(p_2+p_2^*)}{|p_1+p_2^*|^2 c_{2,2}}}
    \end{align*}

    \item After collision, i.e., \(t \rightarrow + \infty\), this gives

    Soliton 1 (\(\xi_1 + \xi_1^* \approx 0,\ \xi_2 + \xi_2^* \rightarrow + \infty\))
    \begin{align*}
        u_1 &\simeq C_1 \frac{(p_1+p_1^*)(p_1 - p_2)}{(p_1+p_2^*) c_{1,1}} \frac{\exp(\i \Im(\xi_1))}{2\theta_2} \sech \left(\Re (\xi_1) + \log\left(\theta_2\right)\right), \\
        u_2 &\simeq D_1 \frac{(p_1+p_1^*)(p_1 - p_2)}{(p_1+p_2^*) c_{1,1}} \frac{\exp(\i \Im(\xi_1))}{2\theta_2} \sech \left(\Re (\xi_1) + \log\left(\theta_2\right)\right),
    \end{align*}
    where the phase term \(\theta_2\) is given by
    \begin{align*}
        \theta_2 = \sqrt{\frac{1}{c_{1,1}} - \frac{(p_1+p_1^*)(p_2+p_2^*)}{|p_2+p_1^*|^2 c_{1,1}}.}
    \end{align*}

    Soliton 2 (\(\xi_2 + \xi_2^* \approx 0,\ \xi_1 + \xi_1^* \rightarrow - \infty\))

    \begin{align*}
        u_1 &\simeq \frac{(p_1+p_1^*)(p_2+p_2^*)}{c_{1,1}} \left(\frac{C_2 c_{1,1}}{p_1 + p_1^*} - \frac{C_1 c_{1,2}}{p_1 + p_2^*}\right) \frac{\exp(\i \Im(\xi_2))}{2 \varphi_2} \sech \left(\Re (\xi_2) - \log\left(\varphi_2\right)\right), \\
        u_2 &\simeq \frac{(p_1+p_1^*)(p_2+p_2^*)}{c_{1,1}} \left(\frac{D_2 c_{1,1}}{p_1 + p_1^*} - \frac{D_1 c_{1,2}}{p_1 + p_2^*}\right) \frac{\exp(\i \Im(\xi_2))}{2 \varphi_2} \sech \left(\Re (\xi_2) - \log\left(\varphi_2\right)\right),
    \end{align*}
    where the phase term \(\varphi_2\) is given by
    \begin{align*}
        \varphi_2 = \sqrt{c_{2,2} - \frac{(p_1+p_1^*)(p_2+p_2^*)|c_{1,2}|^2}{|p_1+p_2^*|^2 c_{1,1}}}.
    \end{align*}

\end{enumerate}
It is observed that, as depicted in \cref{figure:bb_N=2}, a phase shift occurs after the collision of the two solitons, given that the phases term \(-\log(\theta_1) \neq \log(\theta_2)\) and \(-\log(\varphi_1) \neq \log(\varphi_2)\). The aforementioned asymptotic analysis also suggests that setting one of \(C_1, C_2, D_1, D_2\) to zero results in a Y-shaped behavior in the soliton solution. This is because one of the solitons before or after the collision becomes zero, as illustrated in \cref{figure:bb_N=2_Y-shape}.
\begin{figure}[!ht]
    \centering
    \subfigure[]{
        \includegraphics[width=60mm]{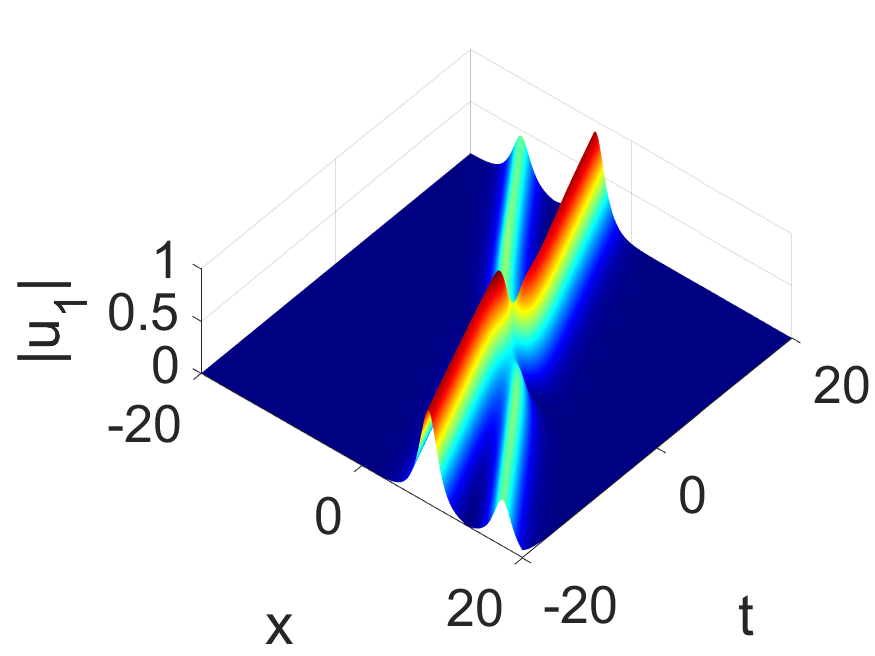}
    }
    \subfigure[]{
        \includegraphics[width=60mm]{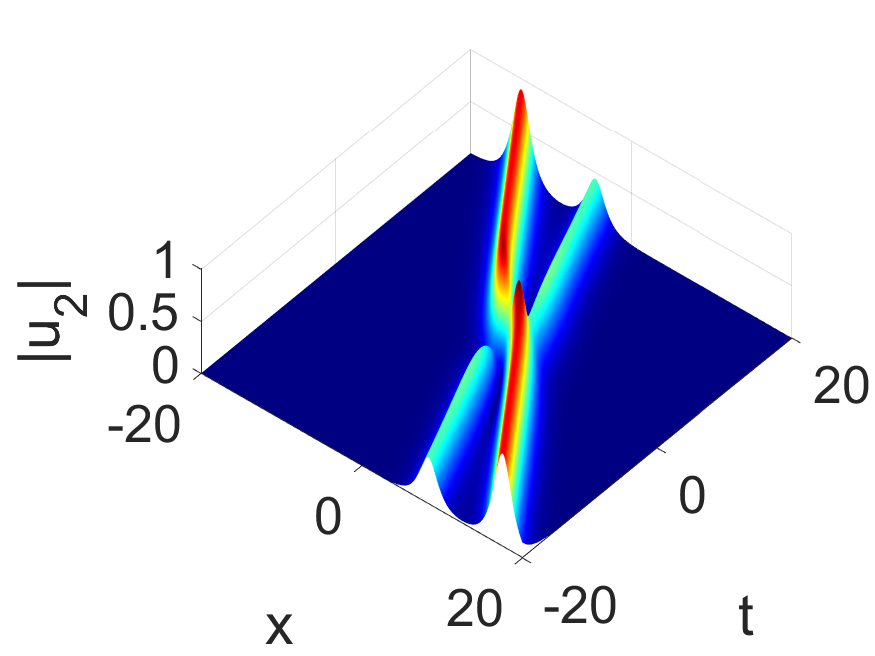}
    }
    \caption{Two-bright soliton solution to the cHirota equation with parameters \(p_1=1+\i /\sqrt{5}, p_2 = 1+\i /\sqrt{10}, C_1 = 2, D_1 = 1, C_2 = 1, D_2 = 2, \xi_{1,0} = \xi_{2,0} = 0, \varepsilon_1 = \varepsilon_2 = -1\).}\label{figure:bb_N=2}
\end{figure}

\begin{figure}[!ht]
    \centering
    \subfigure[]{
        \includegraphics[width=60mm]{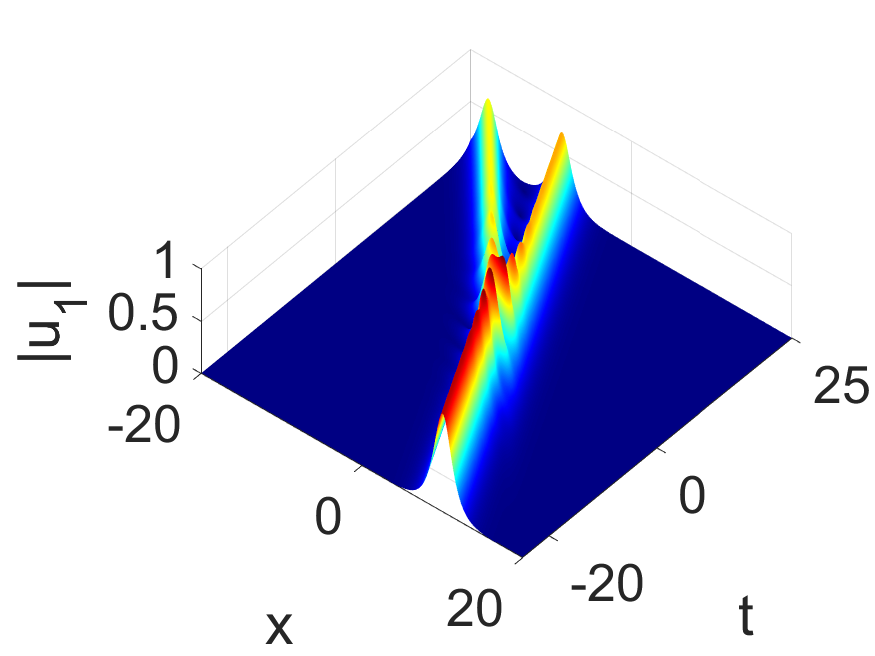}
    }
    \subfigure[]{
        \includegraphics[width=60mm]{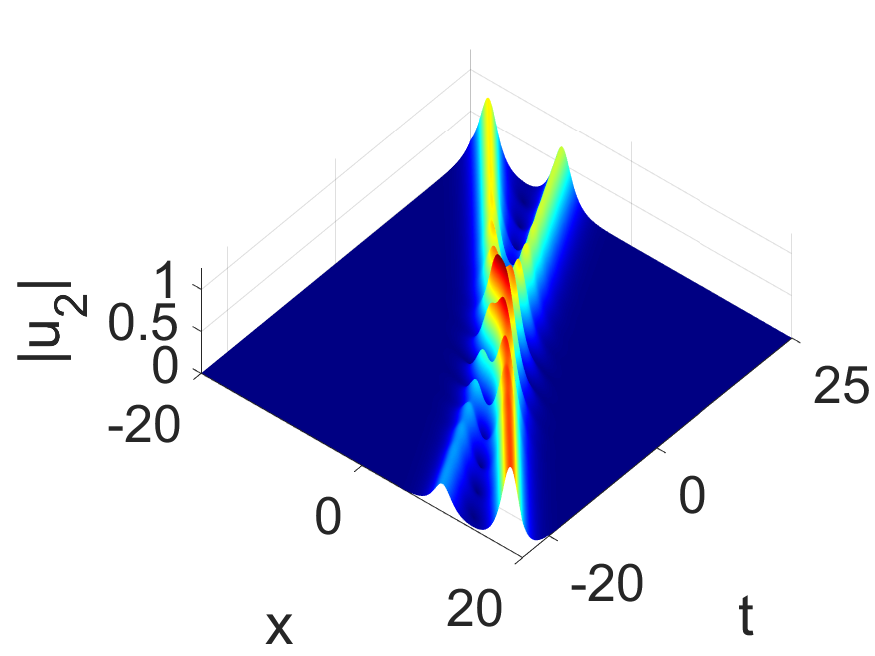}
    }
    \caption{Y-shaped bright soliton solution to the cHirota equation with parameters \(p_1=1-\i /\sqrt{5}, p_2 = 1+\i /\sqrt{10}, C_1 = 1, D_1 = 0, C_2 = 1, D_2 = 1, \xi_{1,0} = \xi_{2,0} = 0, \varepsilon_1 = \varepsilon_2 = -1\).}\label{figure:bb_N=2_Y-shape}
\end{figure}

\section{Dark-Dark soliton solution to the coupled Hirota equation}\label{section:dd_chirota}
The general dark soliton to \eqref{chirota_1}-\eqref{chirota_2} is given by the following theorem.

\begin{theorem}\label{theorem:dd_solution}
Equation \eqref{chirota_1}-\eqref{chirota_2} admits the dark soliton solutions given by
\begin{equation}\label{def_dark_u}
    \begin{split}
        u_1&=\rho_1 \frac{g_1}{f} \exp\left(\i \left(\alpha_1 x - \left(\alpha_1^3+3\alpha_1\left(\varepsilon_1\rho_1^2+\varepsilon_2\rho_2^2\right)+3 (\varepsilon_1\rho_1^2 \alpha_1+\varepsilon_2\rho_2^2 \alpha_2)\right) t\right)\right),\\
        u_2&=\rho_2 \frac{g_2}{f} \exp\left(\i \left(\alpha_2 x - \left(\alpha_2^3+3\alpha_2\left(\varepsilon_1\rho_1^2+\varepsilon_2\rho_2^2\right)+3 (\varepsilon_1\rho_1^2 \alpha_1+\varepsilon_2\rho_2^2 \alpha_2)\right) t\right)\right),
    \end{split}
\end{equation}
and \(f,\ g_1,\ g_2\) are defined as
\begin{equation}
    f=\tau_{0,0}, \quad g_1=\tau_{1,0}, \quad g_2=\tau_{0,1},
\end{equation}
where \(\tau_{k,l}\) is a \(N\times N\) determinant defined as
\begin{equation}
    \tau_{k,l}=\det\left(\delta_{ij}d_i e^{-\xi_i-\eta_j} + \frac{1}{p_i+q_j} \left(-\frac{p_i - \i \alpha_1}{q_j + \i \alpha_1}\right)^k\left(-\frac{p_i - \i \alpha_2}{q_j + \i \alpha_2}\right)^l\right),
\end{equation}
with \(\xi_i=p_i (x - 3(\varepsilon_1\rho_1^2+\varepsilon_2\rho_2^2) t) + p_i^3 t + \xi_{i 0}\), \(\eta_i=q_i (x - 3 (\varepsilon_1\rho_1^2+\varepsilon_2\rho_2^2) t) + q_i^3 t + \xi_{i 0}\). Here \(\alpha_1,\ \alpha_2,\ \rho_1,\ \rho_2\) are real parameters, the parameters \(d_i,\ \xi_{i,0},\ p_i,\ q_i\) satisfy the following complex conjugate relation for each \(h = 0, 1, \ldots, \lfloor N/2 \rfloor\)
\begin{align}\label{dark_complex_relation}
    \begin{split}
        p_i = q_i^*, p_{N+1-i} = q_{N+1-i}^*,\quad d_i,d_{N+1-i} \in \mathbb{R},\quad \xi_{i,0},\xi_{N+1-i} \in \mathbb{R},& \quad \text{for } i \in \{\mathbb{Z} | 1 \leq i \leq h\},\\
        p_i = q_{N+1-i}^*, p_{N+1-i} = q_i^*,\quad d_i = d_{N+1-i} \in \mathbb{R},\quad \xi_{i,0} = \xi_{N+1-i,0} \in \mathbb{R},& \quad \text{for } i \in  \{\mathbb{Z} | h+1 \leq i \leq \lceil N/2 \rceil\}\,.
    \end{split}
\end{align}
Moreover, these parameters need to satisfy the constraint \(G(p_i,q_i) = 0\), for \(i = 1, 2, \ldots, N\), where \(G(p,q)\) defined as
\begin{align}\label{Gpq}
    G(p,q)&=\frac{\varepsilon_1\rho_1^2}{(p_i-\i \alpha_1)(q_i+\i \alpha_1)}+\frac{\varepsilon_2\rho_2^2}{(p_i-\i \alpha_2)(q_i+\i \alpha_2)}-1.
\end{align}
\end{theorem}
\subsection{Derivation of dark-dark soliton solution to the coupled Hirota equation}
Under the dependent variable transformations 
\begin{equation}
u_1=\rho_1(g_1/f)\exp(\i (\alpha_1 x - \omega_1 t)), \quad u_2=\rho_2(g_2/f)\exp(\i (\alpha_2 x - \omega_2 t))    
\end{equation}
with \begin{equation}\label{w_j}
\omega_1=\alpha_1^3+3\alpha_1\left(\varepsilon_1\rho_1^2+\varepsilon_2\rho_2^2\right)+3(\varepsilon_1\rho_1^2\alpha_1+\varepsilon_2\rho_2^2\alpha_2),\quad
\omega_2=\alpha_2^3+3\alpha_2\left(\varepsilon_1\rho_1^2+\varepsilon_2\rho_2^2\right)+3(\varepsilon_1\rho_1^2\alpha_1+\varepsilon_2\rho_2^2\alpha_2).
\end{equation}
it can be shown that the coupled Hirtota equation \eqref{chirota_1}--\eqref{chirota_2} can converted into a set of bilinear equations
\begin{align}
    \begin{split}
        &\left[D_x^3 - D_t +3\i \alpha_1 D_x^2-3 \left(\alpha_1^2 +2\varepsilon_1 \rho_1^2+2\varepsilon_2 \rho_2^2\right)D_x - 3\i \alpha_1 (\varepsilon_1\rho_1^2+\varepsilon_2\rho_2^2) + 3\i (\varepsilon_1\rho_1^2 \alpha_1+\varepsilon_2\rho_2^2 \alpha_2)\right]g_1\cdot f\\
        &\qquad = - 3\i \varepsilon_2 (\alpha_1-\alpha_2) \rho_2^2 s_{12} g_2^* ,\\
    \end{split}\label{vcmKdVDE_BL_1}\\
    \begin{split}
        &\left[D_x^3 - D_t +3\i \alpha_2 D_x^2-3 \left(\alpha_1^2 +2\varepsilon_1 \rho_1^2+2\varepsilon_2 \rho_2^2\right)D_x - 3\i \alpha_2 (\varepsilon_1\rho_1^2+\varepsilon_2\rho_2^2) + 3\i  (\varepsilon_1\rho_1^2 \alpha_1+\varepsilon_2\rho_2^2 \alpha_2)\right]g_2\cdot f\\
        &\qquad = - 3\i \varepsilon_1 (\alpha_2-\alpha_1) \rho_1^2 s_{12} g_1^* ,\\
    \end{split}\label{vcmKdVDE_BL_2}\\
    &\left[D_x  + \i (\alpha_1-\alpha_2)\right]g_1 \cdot g_2=\i (\alpha_1-\alpha_2)s_{12} f, \label{vcmKdVDE_BL_3}\\
    &\left(D_x^2-2\varepsilon_1 \rho_1^2-2\varepsilon_2 \rho_2^2\right) f \cdot f + 2\varepsilon_1 \rho_1^2 |g_1|^2 + 2\varepsilon_2 \rho_2^2 |g_2|^2=0.\label{vcmKdVDE_BL_4}
\end{align}

To deduce the dark soliton solution, we start with the \(\tau\)-function defined as
\begin{equation}\label{tau_dark}
    \tau_{k,l}=\det\left(m_{i j}^{k,l}\right)_{1 \leq i, j \leq N},
\end{equation}
where \(k,l \in \mathbb{Z}\). \(N\) is a positive integer. And the matrix element is defined as
\begin{align*}
    & m_{i j}^{k,l}=d_{ij} + \frac{e^{\xi_i+\eta_j}}{p_i+q_j} \left(-\frac{p_i-a}{q_j+a}\right)^k \left(-\frac{p_i-b}{q_j+b}\right)^l, \\
    & \xi_i=p_i x+p_i^2 y+p_i^3 t+\frac{1}{p_i-a} r+\frac{1}{p_i-b} s+\xi_{i 0}, \\
    & \eta_i=q_i x-q_i^2 y+q_i^3 t+\frac{1}{q_i+a} r+\frac{1}{q_i+b} s+\eta_{i 0}.
\end{align*}
Here \(d_{ij},\ p_i,\ q_j,\ \xi_{i 0},\ \eta_{j 0}\), and \(a_n\) are constants.

It is shown by Ohta \cite{ohta2010dark} that \(\tau_{k,l}\) satisfies a set of bilinear equations below 
\begin{align}
    & \left(D_r D_x-2\right) \tau_{k,l} \cdot \tau_{k,l}=-2 \tau_{k+1,l} \tau_{k-1,l}, \label{KP_dark_1} \\
    & \left(D_s D_x-2\right) \tau_{k,l} \cdot \tau_{k,l}=-2 \tau_{k,l+1} \tau_{k,l-1}, \label{KP_dark_2} \\
    & \left(D_x^2-D_y+2 a D_x\right) \tau_{k+1,l} \cdot \tau_{k,l}=0, \label{KP_dark_3}\\
    & \left(D_x^2-D_y+2 b D_x\right) \tau_{k,l+1} \cdot \tau_{k,l}=0, \label{KP_dark_4}\\
    & \left(D_x^3+3 D_x D_y-4 D_t+3 a\left(D_x^2+D_y\right)+6 a^2 D_x\right) \tau_{k+1,l} \cdot \tau_{k,l}=0, \label{KP_dark_5}\\
    & \left(D_x^3+3 D_x D_y-4 D_t+3 b\left(D_x^2+D_y\right)+6 b^2 D_x\right) \tau_{k,l+1} \cdot \tau_{k,l}=0, \label{KP_dark_6}\\
    & \left(D_s\left(D_x^2-D_y+2 a D_x\right)-4\left(D_x+a-b\right)\right) \tau_{k+1,l} \cdot \tau_{k,l} + 4(a-b) \tau_{k+1,l+1} \cdot \tau_{k,l-1}=0, \label{KP_dark_7}\\
    & \left(D_r\left(D_x^2-D_y+2 a D_x\right)-4D_x\right) \tau_{k+1,l} \cdot \tau_{k,l} =0, \label{KP_dark_8}\\
    & \left(D_r\left(D_x^2-D_y+2 b D_x\right)-4\left(D_x+b-a\right)\right) \tau_{k,l+1} \cdot \tau_{k,l} + 4(b-a) \tau_{k+1,l+1} \cdot \tau_{k-1,l}=0, \label{KP_dark_9}\\
    & \left(D_s\left(D_x^2-D_y+2 b D_x\right)-4D_x\right) \tau_{k,l+1} \cdot \tau_{k,l}=0, \label{KP_dark_10}\\
    & \left(D_x+a-b\right) \tau_{k+1,l} \cdot \tau_{k,l+1}=(a-b) \tau_{k+1,l+1}  \tau_{k,l}\,.\label{KP_dark_11}
\end{align}


The following lemma gives the an algebraic condition to achieve the dimension reduction.
\begin{lemma}\label{lemma:dim_rd_dark}
Take \(d_{ij} = \delta_{ij} d_i\), where \(d_i\) is a constant, and \(\delta_{ij}\) is the Kronecker delta, then under the condition
\begin{equation}
    \frac{\varepsilon_1\rho_1^2}{(p_i-a)(q_i+a)}+\frac{\varepsilon_2\rho_2^2}{(p_i-b)(q_i+b)}=1
\end{equation}
we have
\begin{equation}\label{dmr_dark}
    \left(\varepsilon_1\rho_1^2 \partial_r+\varepsilon_2\rho_2^2 \partial_s\right)\tau_{k,l} = \partial_x\tau_{k,l}.
\end{equation}
\end{lemma}
\begin{proof}
To begin, by setting \(d_{ij} = \delta_{ij} d_i\), our \(\tau\)-functions are now expressed as
\begin{align*}
    &\tau_{k,l}=\prod_{n=1}^{N} e^{\xi_n+\eta_n} \det\left(\mathfrak{m}_{i j}^{k,l} \right),
    \\ 
    &\mathfrak{m}_{i j}^{k,l}= \delta_{ij} d_i e^{-\xi_i-\eta_i} + \frac{1}{p_i+q_j} \left(-\frac{p_i-a}{q_j+a}\right)^k \left(-\frac{p_i-b}{q_j+b}\right)^l.
\end{align*}
The term \(\prod_{n=1}^{N} e^{\xi_i+\eta_i}\) can be dropped in \eqref{KP_dark_1}-\eqref{KP_dark_5} due to the property of the \(D\)-operator in the bilinear equations. 
Note that 
\begin{align}
    \left(\varepsilon_1\rho_1^2 \partial_r+\varepsilon_2\rho_2^2 \partial_s - \partial_x\right)\mathfrak{m}_{i j}^{k,l}&=(p_i+q_i) G(p_i,q_i)\delta_{ij} d_i e^{-\xi_i-\eta_i},\label{diff_rel_dark}\\
    G(p,q)&=\frac{\varepsilon_1\rho_1^2}{(p_i-a)(q_i+a)}+\frac{\varepsilon_2\rho_2^2}{(p_i-b)(q_i+b)}-1.
\end{align}
Hence, if the right-hand side of \eqref{diff_rel_dark} is zero, the dimension reduction condition \eqref{dmr_dark} is satisfied. To archive this, let us consider the case of \(i \neq j\), which results \(\delta_{ij} = 0\), and this leads the right-hand side of \eqref{diff_rel_dark} to be zero. On the other hand, for the case of \(i = j\), since \(\delta_{ii} = 1\), we need to require \(G(p_i,q_i) = 0\). This implies
\begin{equation*}
    \frac{\varepsilon_1\rho_1^2}{(p_i-a)(q_i+a)}+\frac{\varepsilon_2\rho_2^2}{(p_i-b)(q_i+b)}=1.
\end{equation*}
Note that we are not taking \(p_i+q_i=0\) since this term is the denominator of diagonal elements of the determinant of \(\tau\)-function. Overall, we have
\begin{equation*}
    \left(\varepsilon_1 \rho_1^2 \partial_r + \varepsilon_2 \rho_2^2 \partial_s\right)\tau_{\mathrm{n}}=\sum_{i,j=1}^{N} \Delta_{ij}\left(\varepsilon_1 \rho_1^2 \partial_r + \varepsilon_2 \rho_2^2 \partial_s\right)\mathfrak{m}_{ij}^{\mathrm{n}}=\sum_{i,j=1}^{N} \partial_x\mathfrak{m}_{ij}^{\mathrm{n}}=\partial_x\tau_{\mathrm{n}}.
\end{equation*}
Which is nothing but \eqref{dmr_dark}.
\end{proof}

With the help of \cref{lemma:dim_rd_dark}, \eqref{KP_dark_1}-\eqref{KP_dark_2}, \eqref{KP_dark_7}-\eqref{KP_dark_8} and \eqref{KP_dark_9}-\eqref{KP_dark_10} lead to
\begin{align}
    & \left(D^2_x - 2\varepsilon_1 \rho_1^2 - 2\varepsilon_2 \rho_2^2\right)\tau_{k,l} \cdot \tau_{k,l} =-2\varepsilon_1 \rho_1^2\tau_{k+1,l} \tau_{k-1,l}-2\varepsilon_2 \rho_2^2\tau_{k,l+1} \tau_{k,l-1}, \label{KP_dark_1.1} \\
    \begin{split}
        & \left(D_x^3-D_x D_y+2 a D_x^2 -4(\varepsilon_1\rho_1^2 + \varepsilon_2\rho_2^2)D_x -4\varepsilon_2\rho_2^2\left(a-b\right)\right) \tau_{k+1,l} \cdot \tau_{k,l} \\
        &\qquad + 4 \varepsilon_2 \rho_2^2 (a-b) \tau_{k+1,l+1} \cdot \tau_{k,l-1}=0 \label{KP_dark_7.1}
    \end{split}\\
    \begin{split}
        & \left(D_x^3-D_x D_y+2 b D_x^2 -4(\varepsilon_1\rho_1^2 + \varepsilon_2\rho_2^2)D_x -4\varepsilon_1\rho_1^2\left(b-a\right)\right) \tau_{k,l+1} \cdot \tau_{k,l} \\
        &\qquad + 4 \varepsilon_1 \rho_1^2 (b-a) \tau_{k+1,l+1} \cdot \tau_{k-1,l}=0\,. \label{KP_dark_9.1}
    \end{split}
\end{align}
Next, with the use of \eqref{KP_dark_3}--\eqref{KP_dark_6} and the gauge transformation \(x\to x - 3(\varepsilon_1\rho_1^2 + \varepsilon_2\rho_2^2) t,\ t\to t\), \cref{KP_dark_7.1,KP_dark_9.1} become  
\begin{align}
    \begin{split}
        & \left(D_x^3 - D_t + 3 a D_x^2 + 3 \left(a^2 - 2\varepsilon_1 \rho_1^2 - 2\varepsilon_2 \rho_2^2\right) D_x - 3 a (\varepsilon_1\rho_1^2+\varepsilon_2\rho_2^2) + 3 (\varepsilon_1\rho_1^2 a+ \varepsilon_2 \rho_2^2 b) \right) \tau_{k+1,l} \cdot \tau_{k,l} \\
        &\qquad+ 3 \varepsilon_2 \rho_2^2 (a-b) \tau_{k+1,l+1} \cdot \tau_{k,l-1}=0,\label{KP_dark_5.1}
    \end{split}\\
    \begin{split}
        & \left(D_x^3 - D_t + 3 b D_x^2 + 3 \left(b^2 - 2\varepsilon_1 \rho_1^2 - 2\varepsilon_2 \rho_2^2\right) D_x - 3 b (\varepsilon_1\rho_1^2+\varepsilon_2\rho_2^2) + 3 (\varepsilon_1 \rho_1^2 a+ \varepsilon_2\rho_2^2 b) \right) \tau_{k,l+1} \cdot \tau_{k,l} \\
        &\qquad+ 3 \varepsilon_1 \rho_1^2 (b-a) \tau_{k+1,l+1} \cdot \tau_{k-1,l}=0\,,\label{KP_dark_6.1}
    \end{split}
\end{align}
respectively. With the completion of the dimension reduction, we can then remove the dummy variables by setting them to zero, i.e., take \(y=r=s=0\).

The following lemma gives 
the complex conjugate condition. 

\begin{lemma}\label{lemma:cmp_rd_dark}
For \(h = 0, 1, \ldots, \lfloor N/2 \rfloor\), define the following index sets, 
\begin{align}\label{index_sets}
    \begin{split}
        I_1(h) &= \{i \in \mathbb{Z} | 1 \leq i \leq h\}, \\ 
        I_2(h) &= \{i \in \mathbb{Z} | h+1 \leq i \leq \lceil N/2 \rceil\}, \\
        I_3(h) &= \{N+1-i \in \mathbb{Z} | h+1 \leq i \leq \lceil N/2 \rceil\}, \\
        I_4(h) &= \{N+1-i \in \mathbb{Z} | 1 \leq i \leq h\},
    \end{split}
\end{align}
therefore \(I_1(h) \cup I_2(h) \cup I_3(h) \cup I_4(h) = \{i \in \mathbb{Z} | 1 \leq i \leq N\}\). Under the dimension reduction \cref{lemma:dim_rd_dark}, by taking the following restrictions,
\begin{align}\label{aby_dark}
\begin{split}
    p_i = q_i^*,\ d_i \in \mathbb{R},\ \eta_{i,0} = \xi_{i,0} \in \mathbb{R},& \quad \text{for} \ i \in I_1(h) \cup I_4(h),\\
    p_i = q_{N+1-i}^*,\ p_{N+1-i} = q_i^*,\ d_i = d_{N+1-i} \in \mathbb{R},\ \eta_{i,0} = \xi_{N+1-i,0} \in \mathbb{R},& \quad \text{for} \ i \in I_2(h),\\
    a = \i \alpha_1 \in \i \mathbb{R}, \ b = \i \alpha_2 \in \i \mathbb{R},&
\end{split}
\end{align}
we have the \(\tau\)-function \eqref{tau_dark} that satisfies
\begin{equation}\label{cmp_relation_dark}
    \tau_{k,l}^* = \tau_{-k,-l}.
\end{equation}
\end{lemma}
\begin{proof}
By the restriction \eqref{aby_dark}, we have 
\begin{align*}
    \eta_i^* &= q_i^* \left(x - 3 (\varepsilon_1\rho_1^2 + \varepsilon_2\rho_2^2) t\right) + \left(q_i^*\right)^3 t + \eta_{i 0}^* \\
    &= p_i \left(x - 3 (\varepsilon_1\rho_1^2 + \varepsilon_2\rho_2^2) t\right) + p_i^3 t + \xi_{i 0} = \xi_i,& \quad \text{for} \quad i \in I_1(h) \cup I_4(h),\\
    \eta_{N+i}^* &= q_{N+i}^* \left(x - 3 (\varepsilon_1\rho_1^2 + \varepsilon_2\rho_2^2) t\right) + \left(q_{N+i}^*\right)^3 t + \eta_{N+1-i, 0}^* \\
    &= p_i \left(x - 3 (\varepsilon_1\rho_1^2 + \varepsilon_2\rho_2^2) t\right) + p_i^3 t + \xi_{i 0} = \xi_i,& \quad \text{for} \quad i \in I_2(h),\\
    \eta_i^* &= q_i^* \left(x - 3 (\varepsilon_1\rho_1^2 + \varepsilon_2\rho_2^2) t\right) + \left(q_i^*\right)^3 t + \eta_{i0}^* \\
    &= p_{N+i} \left(x - 3 (\varepsilon_1\rho_1^2 + \varepsilon_2\rho_2^2) t\right) + \left(p_{N+i}\right)^3 t + \xi_{N+1-i, 0} = \xi_{N+i},& \quad \text{for} \quad i \in I_2(h).
\end{align*}
Note that
\begin{align*}
    \tau_{k,l}^* = \det \left(\left(\mathfrak{m}_{ij}^{k,l}\right)^*\right)&= \det\left(\delta_{ij} d_i e^{-\xi_i^*-\eta_j^*} + \dfrac{1}{p_i^* + q_j^*}\left(-\frac{p_i^* + \i \alpha_1}{q_j^* - \i \alpha_1}\right)^k \left(-\frac{p_i^* + \i \alpha_2}{q_j^* - \i \alpha_2}\right)^l\right).
\end{align*}
For \(i,j \in I_1(h) \cup I_4(h)\), we have
\begin{align*}
    \left(\mathfrak{m}_{i,j}^{k,l}\right)^* &= \delta_{ji} d_j e^{-\xi_j-\eta_i} + \dfrac{1}{q_i+ p_j}\left(-\frac{q_i + \i \alpha_1}{p_j - \i \alpha_1}\right)^k \left(-\frac{q_i + \i \alpha_2}{p_j - \i \alpha_2}\right)^l = \mathfrak{m}_{j,i}^{-k,-l}.
\end{align*}
For \(i,j \in I_2(h)\), we have
\begin{align*}
    \left(\mathfrak{m}_{i,j}^{k,l}\right)^* &= \delta_{ij} c_{N+i-1} e^{-\eta_{N+1-i}-\xi_{N+1-j}} + \dfrac{1}{q_{N+1-i} + p_{N+1-j}}\left(-\frac{q_{N+1-i} + \i \alpha_1}{p_{N+1-j} - \i \alpha_1}\right)^k \left(-\frac{q_{N+1-i} + \i \alpha_2}{p_{N+1-j} - \i \alpha_2}\right)^l\\
    &= \mathfrak{m}_{N+1-j,N+1-i}^{-k,-l}.
\end{align*}
For \(i \in I_1(h) \cup I_4(h), j \in I_2(h)\), we have
\begin{align*}
    \left(\mathfrak{m}_{ij}^{k,l}\right)^* &=  \dfrac{1}{q_i + p_{N+1-j}} \left(-\frac{q_i + \i \alpha_1}{p_{N+1-j} - \i \alpha_1}\right)^k\left(-\frac{q_i + \i \alpha_2}{p_{N+1-j} - \i \alpha_2}\right)^l = \mathfrak{m}_{N+1-j,i}^{-k,-l},\\
    \left(\mathfrak{m}_{i,N+1-j}^{k,l}\right)^* &=  \dfrac{1}{q_i + p_j} \left(-\frac{q_i + \i \alpha_1}{p_j - \i \alpha_1}\right)^k\left(-\frac{q_i + \i \alpha_2}{p_j - \i \alpha_2}\right)^l = \mathfrak{m}_{j,i}^{-k,-l}.
\end{align*}
For \(i \in I_2(h), j \in I_1(h) \cup I_4(h)\), we have
\begin{align*}
    \left(\mathfrak{m}_{ij}^{k,l}\right)^* &=  \dfrac{1}{q_{N+1-i}+ p_i} \left(-\frac{q_{N+1-i} + \i \alpha_1}{p_i - \i \alpha_1}\right)^k\left(-\frac{q_{N+1-i} + \i \alpha_2}{p_i - \i \alpha_2}\right)^l = \mathfrak{m}_{j,N+1-i}^{-k,-l}, \\
    \left(\mathfrak{m}_{N+1-i,j}^{k,l}\right)^* &=  \dfrac{1}{q_i+ p_i} \left(-\frac{q_i + \i \alpha_1}{p_i - \i \alpha_1}\right)^k\left(-\frac{q_{N+1-i} + \i \alpha_2}{p_i - \i \alpha_2}\right)^l = \mathfrak{m}_{j,i}^{-k,-l}.
\end{align*}
We can then consider the following row operation on determinant \(\tau_{k,l}^*\): for \(i \in I_2(h)\), switching \(i\)-th row with \((N+1-i)\)-th row, \(i\)-th column with \((N+1-i)\)-th column. This involves an even number of row operations and does not change the sign of the determinant. 
\begin{align*}
    \tau_{k,l}^* &= \begin{vmatrix}
        \left(\mathfrak{m}_{i\in I_1, j\in I_1}^{k,l}\right)^* & \left(\mathfrak{m}_{i\in I_1, j\in I_2}^{k,l}\right)^* & \left(\mathfrak{m}_{i\in I_1, j\in I_3}^{k,l}\right)^* & \left(\mathfrak{m}_{i\in I_1, j\in I_4}^{k,l}\right)^*\\
        \left(\mathfrak{m}_{i\in I_2, j\in I_1}^{k,l}\right)^* & \left(\mathfrak{m}_{i\in I_2, j\in I_2}^{k,l}\right)^* & \left(\mathfrak{m}_{i\in I_2, j\in I_3}^{k,l}\right)^* & \left(\mathfrak{m}_{i\in I_2, j\in I_4}^{k,l}\right)^* \\
        \left(\mathfrak{m}_{i\in I_3, j\in I_1}^{k,l}\right)^* & \left(\mathfrak{m}_{i\in I_3, j\in I_2}^{k,l}\right)^* & \left(\mathfrak{m}_{i\in I_3, j\in I_3}^{k,l}\right)^* & \left(\mathfrak{m}_{i\in I_3, j\in I_4}^{k,l}\right)^* \\
        \left(\mathfrak{m}_{i\in I_4, j\in I_1}^{k,l}\right)^* & \left(\mathfrak{m}_{i\in I_4, j\in I_2}^{k,l}\right)^* & \left(\mathfrak{m}_{i\in I_4, j\in I_3}^{k,l}\right)^* & \left(\mathfrak{m}_{i\in I_4, j\in I_4}^{k,l}\right)^* \\
    \end{vmatrix} \\
    &= \begin{vmatrix}
        \mathfrak{m}_{j\in I_1, i\in I_1}^{-k,-l} &\mathfrak{m}_{j\in I_1, i\in I_3}^{-k,-l} & \mathfrak{m}_{j\in I_1, i\in I_2}^{-k,-l} & \mathfrak{m}_{j\in I_1, i\in I_4}^{-k,-l} \\
        \mathfrak{m}_{j\in I_3, i\in I_1}^{-k,-l} & \mathfrak{m}_{j\in I_3, i\in I_3}^{-k,-l} & \mathfrak{m}_{j\in I_3, i\in I_2}^{-k,-l} & \mathfrak{m}_{j\in I_3, i\in I_4}^{-k,-l} \\
        \mathfrak{m}_{j\in I_2, i\in I_1}^{-k,-l} & \mathfrak{m}_{j\in I_2, i\in I_3}^{-k,-l} & \mathfrak{m}_{j\in I_2, i\in I_2}^{-k,-l} & \mathfrak{m}_{j\in I_2, i\in I_4}^{-k,-l}\\
        \mathfrak{m}_{j\in I_4, i\in I_1}^{-k,-l} & \mathfrak{m}_{j\in I_4, i\in I_3}^{-k,-l} & \mathfrak{m}_{j\in I_4, i\in I_2}^{k,l} & \mathfrak{m}_{j\in I_4, i\in I_4}^{-k,-l} \\
    \end{vmatrix} \\
    &= \begin{vmatrix}
        \mathfrak{m}_{i\in I_1, j\in I_1}^{-k,-l} &\mathfrak{m}_{i\in I_1, j\in I_2}^{-k,-l} & \mathfrak{m}_{i\in I_1, j\in I_3}^{-k,-l} & \mathfrak{m}_{i\in I_1, j\in I_4}^{-k,-l} \\
        \mathfrak{m}_{i\in I_2, j\in I_1}^{-k,-l} &\mathfrak{m}_{i\in I_2, j\in I_2}^{-k,-l} & \mathfrak{m}_{i\in I_2, j\in I_3}^{-k,-l} & \mathfrak{m}_{i\in I_2, j\in I_4}^{-k,-l} \\
        \mathfrak{m}_{i\in I_3, j\in I_1}^{-k,-l} &\mathfrak{m}_{i\in I_3, j\in I_2}^{-k,-l} & \mathfrak{m}_{i\in I_3, j\in I_3}^{-k,-l} & \mathfrak{m}_{i\in I_3, j\in I_4}^{-k,-l} \\
        \mathfrak{m}_{i\in I_4, j\in I_1}^{-k,-l} &\mathfrak{m}_{i\in I_4, j\in I_2}^{-k,-l} & \mathfrak{m}_{i\in I_4, j\in I_3}^{-k,-l} & \mathfrak{m}_{i\in I_4, j\in I_4}^{-k,-l} \\
    \end{vmatrix} = \tau_{-k,-l}.
\end{align*}
\end{proof}

By \cref{lemma:cmp_rd_dark}, we have \(\tau_{1,0}=\tau_{-1,0}^*,\ \tau_{0,1}=\tau_{0,-1}^*,\ \tau_{0,0}=\tau_{0,0} \in \mathbb{R}\), and these relations on \(\tau\)-functions allow us to set
\begin{align}\label{tau_settings_dark}
    f=\tau_{0,0}, \quad g_1=\tau_{1,0}=\tau_{-1,0}^*=g_1^*,\quad g_2=\tau_{0,1}=\tau_{0,-1}^*=g_2^*,\quad s_{12}=\tau_{1,1}.
\end{align}
Setting \(k=l=0\) in \cref{KP_dark_5.1,KP_dark_6.1,KP_dark_11,KP_dark_1.1}, one obtains the bilinear form of cHirota equation \eqref{vcmKdVDE_BL_1}-\eqref{vcmKdVDE_BL_4}. 
This 
completes the proof of \cref{theorem:dd_solution}.

\subsection{Dynamics of dark-dark soliton solution}
By taking \(N=1\) in \cref{theorem:dd_solution}, we can only choice \(h=0\) for \eqref{dark_complex_relation}, this gives the solution as
\begin{align*}\label{M=2_plane_wave}
    u_1 &= \rho_1 \exp(\i \theta_1)\frac{d_1 \exp(-\xi_1-\xi_1^*) - \dfrac{1}{p_1 + p_1^*}\left(\dfrac{p_1 - \i \alpha_1}{p_1^* + \i \alpha_1}\right)}{d_1 \exp(-\xi_1-\xi_1^*) + 1/(p_1 + p_1^*)} , \\
    u_2 &= \rho_2 \exp(\i \theta_2)\frac{d_1 \exp(-\xi_1-\xi_1^*) - \dfrac{1}{p_1 + p_1^*}\left(\dfrac{p_1 - \i \alpha_2}{p_1^* + \i \alpha_2}\right)}{d_1 \exp(-\xi_1-\xi_1^*) + 1/(p_1 + p_1^*)} , 
\end{align*}
where \(\xi_1 = p_1 (x - 3 (\varepsilon_1\rho_1^2 + \varepsilon_2\rho_2^2) t) + p_1^3 t + \xi_{1 0}\), and \(\exp(\i \theta_1)\), \(\exp(\i \theta_2)\) are plane wave solution on each component as
\begin{align}
    \theta_i = \alpha_i x - \left(\alpha_i^3 + 3\alpha_i \left(\varepsilon_1\rho_1^2 + \varepsilon_2\rho_2^2\right) + 3 \left(\varepsilon_1\rho_1^2 \alpha_1 + \varepsilon_2\rho_2^2 \alpha_2\right)\right)t, \quad i = 1,2.
\end{align} 
The parameters \(p_1, \rho_1, \rho_2, \alpha_1, \alpha_2, \varepsilon_1, \varepsilon_2\) satisfy \eqref{Gpq}
\begin{align}\label{p_1_equation_n=1}
    \frac{\varepsilon_1 \rho_1^2}{|p_1-\i \alpha_1|^2} + \frac{\varepsilon_2 \rho_2^2}{|p_1-\i \alpha_2|^2}=1.
\end{align}
It is evident that for above solutions \(u_1\) and \(u_2\) to be regular, we require \(d_1 \Re(p_1) > 0\), hence the denominators of \(u_1\) and \(u_2\) will be nonzero.
The first order solution can be further simplified as
\begin{align*}
    u_1 &= \frac{\rho_1 \exp(\i \theta_1)}{p_1^* + \i \alpha_1} \biggl(\i \Bigl(\alpha_1 - \Im(p_1)\Bigr) - \Re(p_1)\tanh\Bigl(\Re(\xi_1) - \log\left(4 d_1 \Re(p_1)\right)\Bigr) \biggr) , \\
    u_2 &= \frac{\rho_2 \exp(\i \theta_2)}{p_1^* + \i \alpha_2} \biggl(\i \Bigl(\alpha_2 - \Im(p_1) \Bigr) - \Re(p_1)\tanh\Bigl(\Re(\xi_1) - \log\left(4 d_1 \Re(p_1)\right)\Bigr) \biggr).
\end{align*}
Note that \(\Re(\xi_1) = \Re(p_1) \Bigl(x - \left(3 (\varepsilon_1\rho_1^2 + \varepsilon_2\rho_2^2) + 3\left[\Im(p_1)\right]^2  - \left[\Re(p_1)\right]^2\right)t\Bigr) + \Re(\xi_{1 0})\), the speed of soliton is therefore \(3 (\varepsilon_1\rho_1^2 + \varepsilon_2\rho_2^2) + 3\left[\Im(p_1)\right]^2  - \left[\Re(p_1)\right]^2\).
Moreover, by assuming \(d_1 > 0\), we have
\begin{align*}
    |u_1|^2 &= \rho_1^2 \left(1 - \frac{\left[\Re(p_1)\right]^2}{|p_1^* + \i \alpha_1|^2}\sech^2\Bigl(\Re(\xi_1) - \log\left(4 d_1 \Re(p_1)\right)\Bigr) \right), \\
    |u_2|^2 &= \rho_2^2 \left(1 - \frac{\left[\Re(p_1)\right]^2}{|p_1^* + \i \alpha_2|^2}\sech^2\Bigl(\Re(\xi_1) - \log\left(4 d_1 \Re(p_1)\right)\Bigr) \right).
\end{align*}
The minimum of \(|u_1|^2\) and \(|u_2|^2\) are given by \(\rho_1^2 \left(1 -\left[\Re(p_1)\right]^2/|p_1^* + \i \alpha_1|^2\right)\) and \(\rho_2^2 \left(1 - \left[\Re(p_1)\right]^2/|p_1^* + \i \alpha_2|^2\right)\) respectively. 

An example of one-dark soliton is illustrated in \cref{figure:dd_N=1}.
\begin{figure}[!ht]
    \centering
    \subfigure[]{\label{figure:dd_N=1(a)}
        \includegraphics[width=60mm]{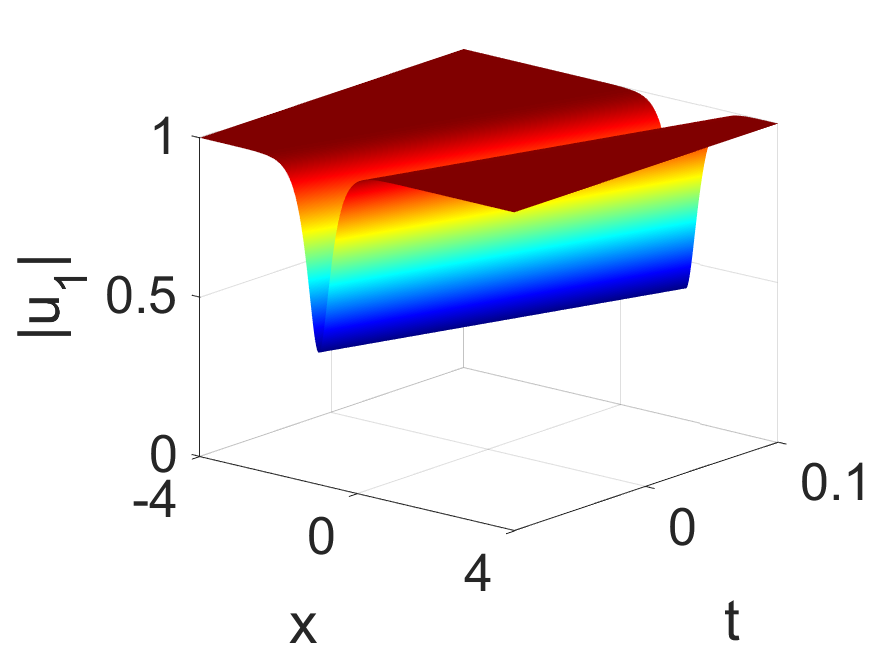}
    }
    \subfigure[]{\label{figure:dd_N=1(b)}
        \includegraphics[width=60mm]{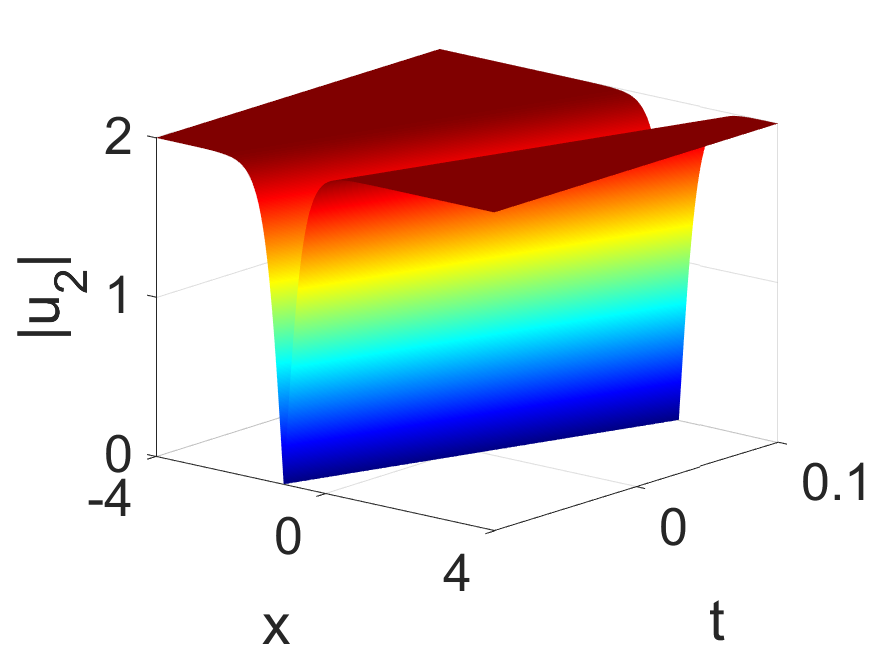}
    }
    \caption{One-dark soliton solution to the cHirota equation with parameters \(p_1= \sqrt{\left(\sqrt{113}+9\right)/2} + \i, d_1 = 1, \alpha_1 = 2, \alpha_2 =1, \rho_1 = 1, \rho_2 = 2, \xi_{1,0} = 0, \varepsilon_1 = \varepsilon_2 = 1\).}\label{figure:dd_N=1}
\end{figure}


{\bf {Two-dark solution}:} In this case, the \(\tau\) functions take the form
\begin{align*}
    g_1 &= \begin{vmatrix}
        d_1 \exp(- \xi_1 - \eta_1) + \dfrac{1}{p_1 + q_1}\left(-\dfrac{p_1 - \i \alpha_1}{q_1 + \i \alpha_1}\right) & \dfrac{1}{p_1 + q_2}\left(-\dfrac{p_1 - \i \alpha_1}{q_2 + \i \alpha_1}\right) \\
        \dfrac{1}{p_2 + q_1}\left(-\dfrac{p_2 - \i \alpha_1}{q_1 + \i \alpha_1}\right) & d_2 \exp(- \xi_2 - \eta_2) + \dfrac{1}{p_2 + q_2}\left(-\dfrac{p_2 - \i \alpha_1}{q_2 + \i \alpha_1}\right) \\
    \end{vmatrix}\\
    g_2 &= \begin{vmatrix}
        d_1 \exp(- \xi_1 - \eta_1) + \dfrac{1}{p_1 + q_1}\left(-\dfrac{p_1 - \i \alpha_2}{q_1 + \i \alpha_2}\right) & \dfrac{1}{p_1 + q_2}\left(-\dfrac{p_1 - \i \alpha_2}{q_2 + \i \alpha_2}\right) \\
        \dfrac{1}{p_2 + q_1}\left(-\dfrac{p_2 - \i \alpha_2}{q_1 + \i \alpha_2}\right) & d_2 \exp(- \xi_2 - \eta_2) + \dfrac{1}{p_2 + q_2}\left(-\dfrac{p_2 - \i \alpha_2}{q_2 + \i \alpha_2}\right) \\
    \end{vmatrix}\\
    f &= \begin{vmatrix}
        d_1 \exp(- \xi_1 - \eta_1) + \dfrac{1}{p_1 + q_1} & \dfrac{1}{p_1 + q_2} \\
        \dfrac{1}{p_2 + q_1} & d_2 \exp(- \xi_2 - \eta_2) + \dfrac{1}{p_2 + q_2} \\
    \end{vmatrix}
\end{align*}
where \(\xi_i = p_i (x - 3 (\varepsilon_1\rho_1^2 + \varepsilon_2\rho_2^2) t) + p_i^3 t + \xi_{i 0}, \eta_i = q_i (x - 3 (\varepsilon_1\rho_1^2 + \varepsilon_2\rho_2^2) t) + q_i^3 t + \xi_{i 0}, i = 1, 2\), and \(\theta_i\) is defined by \eqref{M=2_plane_wave}. By \eqref{dark_complex_relation}, we have two possible choices for \(h\), which means parameters \(p_1, p_2, q_1, q_2, d_1, d_2\) need to satisfy one of the following conditions
\begin{align}
    &(h=0 \ \text{case}) &q_2 = p_1^*, \quad q_1 = p_2^*, \quad d_2 = d_1 \in \mathbb{R}, \label{dark_N=2_complex_1}\\
    &(h=1 \ \text{case}) &q_1 = p_1^*, \quad q_2 = p_2^*, \quad d_1,d_2 \in \mathbb{R}. \label{dark_N=2_complex_2}
\end{align}

Moreover, for \eqref{dark_N=2_complex_1}, the parameters need to satisfy \(G(p_1, p_2^*) = 0\) and \(G(p_2, p_1^*) = 0\); for \eqref{dark_N=2_complex_2}, we need \(G(p_1, p_1^*) = 0\) and \(G(p_2, p_2^*) = 0\), where \(G(p,q)\) is defined by \eqref{Gpq}. 

For \(h = 0\), the following one-breather solution is obtained
\begin{align*}
    g_1 ={} & d_1^2 \exp(-2\Re(\xi_1 + \xi_2)) +\frac{p_1 - \i \alpha_1}{p_1^* + \i \alpha_1}\frac{p_2 - \i \alpha_1}{p_2^* + \i \alpha_1} \left(\frac{1}{|p_1 + p_2^*|^2} - \frac{1}{4 \Re(p_1) \Re(p_2)}\right) \\
    & - d_1 \exp(-\Re(\xi_1 + \xi_2))\left(\frac{\exp(\i \Im(\xi_2 - \xi_1))}{p_1 + p_2^*}\frac{p_1 - \i \alpha_1}{p_2^* + \i \alpha_1} + \frac{\exp(\i \Im(\xi_1 - \xi_2))}{p_1^* + p_2}\frac{p_2 - \i \alpha_1}{p_1^* + \i \alpha_1}\right), \\
    g_2 ={} & d_1^2 \exp(-2\Re(\xi_1 + \xi_2)) +\frac{p_1 - \i \alpha_2}{p_1^* + \i \alpha_2}\frac{p_2 - \i \alpha_2}{p_2^* + \i \alpha_2} \left(\frac{1}{|p_1 + p_2^*|^2} - \frac{1}{4 \Re(p_1) \Re(p_2)}\right)\\
    & - d_1 \exp(-\Re(\xi_1 + \xi_2))\left(\frac{\exp(\i \Im(\xi_2 - \xi_1))}{p_1 + p_2^*}\frac{p_1 - \i \alpha_2}{p_2^* + \i \alpha_2} + \frac{\exp(\i \Im(\xi_1 - \xi_2))}{p_1^* + p_2}\frac{p_2 - \i \alpha_2}{p_1^* + \i \alpha_2}\right),\\
    f ={} & d_1^2 \exp(-2\Re(\xi_1 + \xi_2)) + \frac{1}{|p_1 + p_2^*|^2} - \frac{1}{4 \Re(p_1) \Re(p_2)}\\
    & + d_1 \exp(-\Re(\xi_1 + \xi_2))\left(\frac{\exp(\i \Im(\xi_2 - \xi_1))}{p_1 + p_2^*} + \frac{\exp(\i \Im(\xi_1 - \xi_2))}{p_1^* + p_2}\right).
\end{align*}
The term \(\exp(\i \Im(\xi_1 - \xi_2))\) or \(\exp(\i \Im(\xi_2 - \xi_1))\) represents the periodicity  of the breather.  Such a breather solution is shown in \cref{figure:dd_N=2_h=0} with the period being \(\Im(\xi_1 - \xi_2)\).

\begin{figure}[!ht]
    \centering
    \subfigure[]{\label{figure:dd_N=2_h=0(a)}
        \includegraphics[width=60mm]{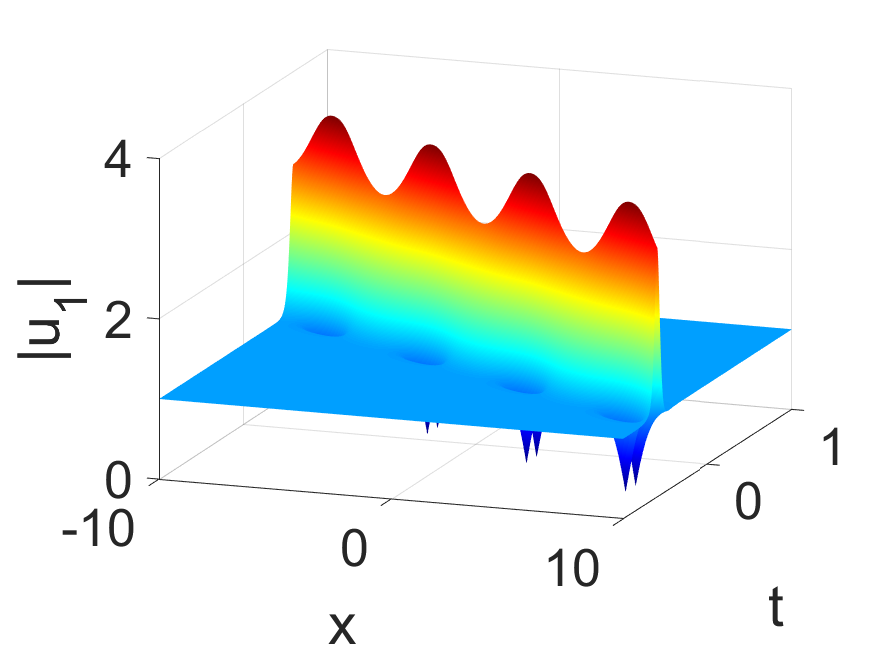}
    }
    \subfigure[]{\label{figure:dd_N=2_h=0(b)}
        \includegraphics[width=60mm]{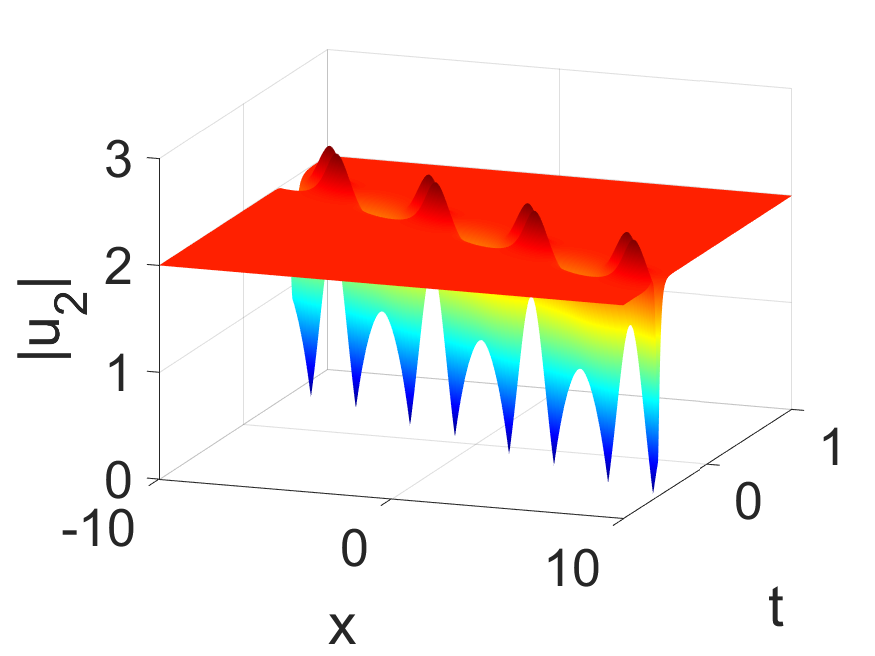}
    }
    \caption{One-breather solution to the cHirota equation with parameters \(p_1= 1 + \i, p_2 = \left(1/4-\i/4\right) \left(-8-\i+\sqrt{-5+36 \i}\right), d_1 = 1, \alpha_1 = 2, \alpha_2 =1, \rho_1 = 1, \rho_2 = 2, \xi_{1,0} = \xi_{2,0} = 0, \varepsilon_1 = \varepsilon_2 = -1\).}\label{figure:dd_N=2_h=0}
\end{figure}

For the case of \(h = 1\), two-dark soliton solution is given
\begin{align*}
    g_1 &= \begin{vmatrix}
        d_1 \exp(- \xi_1 - \xi_1^*) + \dfrac{1}{p_1 + p_1^*}\left(-\dfrac{p_1 - \i \alpha_1}{p_1^* + \i \alpha_1}\right) & \dfrac{1}{p_1 + p_2^*}\left(-\dfrac{p_1 - \i \alpha_1}{p_2^* + \i \alpha_1}\right) \\
        \dfrac{1}{p_2 + p_1^*}\left(-\dfrac{p_2 - \i \alpha_1}{p_1^* + \i \alpha_1}\right) & d_2 \exp(- \xi_2 - \xi_2^*) + \dfrac{1}{p_2 + p_2^*}\left(-\dfrac{p_2 - \i \alpha_1}{p_2^* + \i \alpha_1}\right) \\
    \end{vmatrix}\\
    g_2 &= \begin{vmatrix}
        d_1 \exp(- \xi_1 - \xi_1^*) + \dfrac{1}{p_1 + p_1^*}\left(-\dfrac{p_1 - \i \alpha_2}{p_1^* + \i \alpha_2}\right) & \dfrac{1}{p_1 + p_2^*}\left(-\dfrac{p_1 - \i \alpha_2}{p_2^* + \i \alpha_2}\right) \\
        \dfrac{1}{p_2 + p_1^*}\left(-\dfrac{p_2 - \i \alpha_2}{p_1^* + \i \alpha_2}\right) & d_2 \exp(- \xi_2 - \xi_2^*) - \dfrac{1}{p_2 + p_2^*}\left(-\dfrac{p_2 - \i \alpha_2}{p_2^* + \i \alpha_2}\right) \\
    \end{vmatrix}\\
    f &= \begin{vmatrix}
        d_1 \exp(- \xi_1 - \xi_1^*) + \dfrac{1}{p_1 + p_1^*} & \dfrac{1}{p_1 + p_2^*} \\
        \dfrac{1}{p_2 + p_1^*} & d_2 \exp(- \xi_2 - \xi_2^*) + \dfrac{1}{p_2 + p_2^*} \\
    \end{vmatrix}
\end{align*}
with \(G(p_1, p_1^*) = 0\) and \(G(p_2, p_2^*) = 0\).
Marking \(\xi_1 + \xi_1^*\) as soliton 1 and \(\xi_2 + \xi_2^*\) as soliton 2 and assuming soliton 1 is on the right of soliton 2 before the collision,
the asymptotic behavior of above solution is given below
\begin{enumerate}[(1)]
    \item Before collision, i.e., \(t \rightarrow - \infty\), this gives
    
    Soliton 1 (\(\xi_1 + \xi_1^* \approx 0,\ \xi_2 + \xi_2^* \rightarrow - \infty\))
    \begin{align*}
        u_1 &\simeq \frac{\rho_1 \exp(\i \theta_1)}{p_1^* + \i \alpha_1} \biggl(\i \Bigl(\alpha_1 - \Im(p_1)\Bigr) - \Re(p_1)\tanh\Bigl(\Re(\xi_1) - \log\left(4 d_1 \Re(p_1)\right)\Bigr) \biggr) , \\
        u_2 &\simeq \frac{\rho_2 \exp(\i \theta_2)}{p_1^* + \i \alpha_2} \biggl(\i \Bigl(\alpha_2 - \Im(p_1)\Bigr) - \Re(p_1)\tanh\Bigl(\Re(\xi_1) - \log\left(4 d_1 \Re(p_1)\right)\Bigr) \biggr).
    \end{align*}
    which is exactly same as the first order dark soliton.

    Soliton 2 (\(\xi_2 + \xi_2^* \approx 0,\ \xi_1 + \xi_1^* \rightarrow + \infty\))
    \begin{align*}
        u_1 &\simeq \frac{\rho_1 \exp(\i \theta_1)}{p_2^* + \i \alpha_1} \left(-\dfrac{p_1 - \i \alpha_1}{p_1^* + \i \alpha_1}\right) \biggl(\i \Bigl(\alpha_1 - \Im(p_2)\Bigr) - \Re(p_2)\tanh\Bigl(\Re(\xi_2) + \log\left(\varphi_1\right)\Bigr) \biggr) , \\
        u_2 &\simeq \frac{\rho_2 \exp(\i \theta_2)}{p_2^* + \i \alpha_2} \left(-\dfrac{p_1 - \i \alpha_2}{p_1^* + \i \alpha_2}\right) \biggl(\i \Bigl(\alpha_2 - \Im(p_2)\Bigr) - \Re(p_2)\tanh\Bigl(\Re(\xi_2) + \log\left(\varphi_1\right)\Bigr) \biggr).
    \end{align*}
    where the phase term \(\varphi_1\) is given by \(\varphi_1 = \frac{1}{2d_2} \left(\frac{1}{2\Re(p_2)} - \frac{2\Re(p_1)}{|p_1 + p_2^*|^2}\right)\).

    \item After collision, i.e., \(t \rightarrow + \infty\), this gives

    Soliton 1 (\(\xi_1 + \xi_1^* \approx 0,\ \xi_2 + \xi_2^* \rightarrow + \infty\))
    \begin{align*}
        u_1 &\simeq \frac{\rho_1 \exp(\i \theta_1)}{p_1^* + \i \alpha_1} \left(-\dfrac{p_2 - \i \alpha_1}{p_2^* + \i \alpha_1}\right) \biggl(\i \Bigl(\alpha_1 - \Im(p_1)\Bigr) - \Re(p_1)\tanh\Bigl(\Re(\xi_1) + \log\left(\theta_2\right)\Bigr) \biggr) , \\
        u_2 &\simeq \frac{\rho_2 \exp(\i \theta_2)}{p_1^* + \i \alpha_2} \left(-\dfrac{p_2 - \i \alpha_2}{p_2^* + \i \alpha_2}\right) \biggl(\i \Bigl(\alpha_2 - \Im(p_1)\Bigr) - \Re(p_1)\tanh\Bigl(\Re(\xi_1) + \log\left(\theta_2\right)\Bigr) \biggr).
    \end{align*}
    where the phase term \(\theta_2\) is given by \(\theta_2 = \frac{1}{2d_1} \left(\frac{1}{2\Re(p_1)} - \frac{2\Re(p_2)}{|p_1^* + p_2|^2}\right)\).

    Soliton 2 (\(\xi_2 + \xi_2^* \approx 0,\ \xi_1 + \xi_1^* \rightarrow - \infty\))
    \begin{align*}
        u_1 &\simeq \frac{\rho_1 \exp(\i \theta_1)}{p_2^* + \i \alpha_1} \biggl(\i \Bigl(\alpha_1 - \Im(p_2)\Bigr) - \Re(p_2)\tanh\Bigl(\Re(\xi_2) - \log\left(4 d_1 \Re(p_2)\right)\Bigr) \biggr) , \\
        u_2 &\simeq \frac{\rho_2 \exp(\i \theta_2)}{p_2^* + \i \alpha_2} \biggl(\i \Bigl(\alpha_2 - \Im(p_2)\Bigr) - \Re(p_2)\tanh\Bigl(\Re(\xi_2) - \log\left(4 d_1 \Re(p_2)\right)\Bigr) \biggr).
    \end{align*}

\end{enumerate}

One can find that there is a phase shift after the collision (see \cref{figure:dd_N=2_h=1}). 

\begin{figure}[!ht]
    \centering
    \subfigure[]{\label{figure:dd_N=2_h=1(a)}
        \includegraphics[width=60mm]{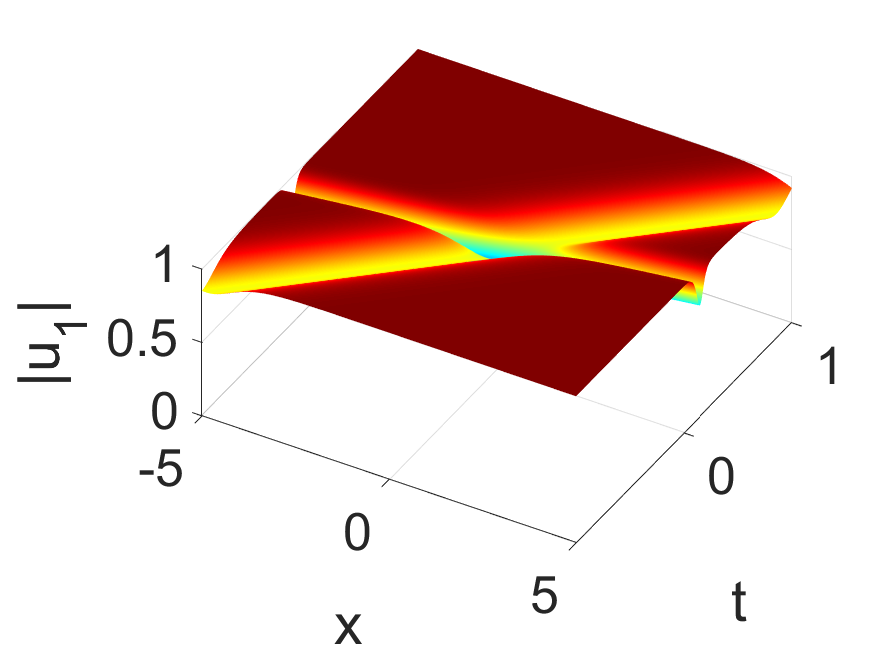}
    }
    \subfigure[]{\label{figure:dd_N=2_h=1(b)}
        \includegraphics[width=60mm]{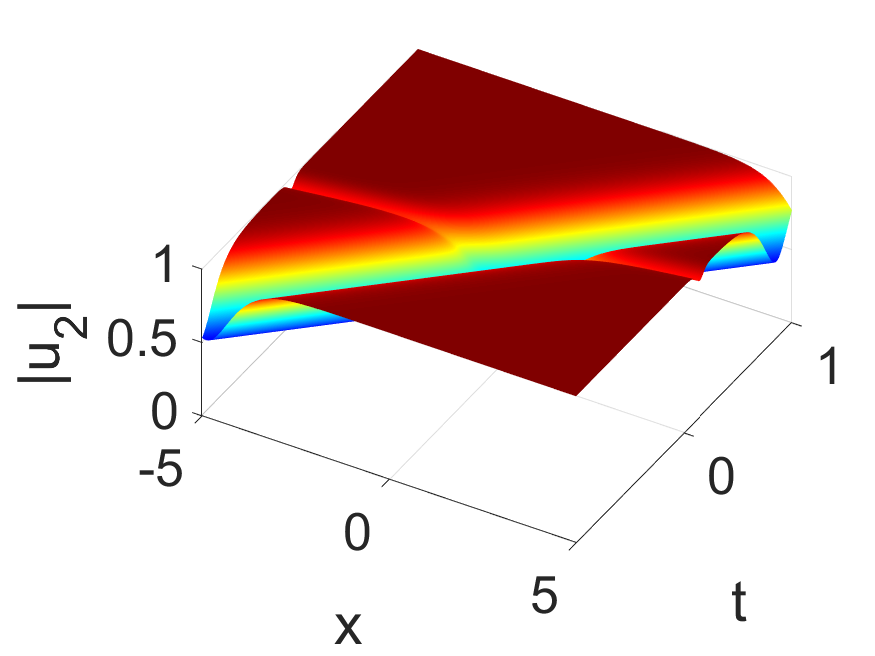}
    }
    \caption{Two-dark-soliton solution to the cHirota equation with parameters \(p_1= 1 + \left(3-\sqrt{5}\right) \i/2 , p_2 = 3/4 + \left(6+\sqrt{4 \sqrt{23}+11}\right) \i/4 , d_1 = d_2 = 1, \alpha_1 = 2, \alpha_2 =1, \rho_1 = \rho_2 = 1, \xi_{1,0} = \xi_{2,0} = 0, \varepsilon_1 = \varepsilon_2 = 1\).}\label{figure:dd_N=2_h=1}
\end{figure}

\section{Bright-dark soliton solution to the coupled Hirota equation}\label{section:bd_chirota}
Similarly, the bright-dark soliton solution is summarized by the following theorem.
\begin{theorem}\label{theorem:bd_solution}
Equation \eqref{chirota_1}-\eqref{chirota_2} admits the following bright-dark soliton solution, where the bright soliton solution \(u_1\) and dark soliton solution \(u_2\) are given by
\begin{equation}
    u_1 = \frac{g_1}{f}\exp\left( -3 \i\rho_2^2 \varepsilon_2 \alpha_2 t \right), \quad u_2 = \frac{h_2}{f}\rho_2 \exp\left(\i (\alpha_2 x - (\alpha_2^3 + 6\varepsilon_2\alpha_2 \rho_2^2) t)\right),
\end{equation}
and \(f,g_1, h_2\) are determinants defined as
\begin{equation}
    f = \left|M_0\right|,\quad g_1 = \begin{vmatrix}
        M_0 & \Phi \\
        -\left(\bar{\Psi}\right)^T & 0
    \end{vmatrix},\quad h_2 = \left|M_1\right|
\end{equation}
where \(M_k\) is \(N\times N\) matrix, \(\Phi\) and \(\bar{\Psi}\) are \(N\)-component vectors whose elements are defined as
\begin{align}
    &m_{ij}^k=\frac{1}{p_i+p_j^*}\left(-\frac{p_i - \i \alpha_2}{p_j^* + \i \alpha_2}\right)^k e^{\xi_i+\xi_j^*}+ \frac{\varepsilon_1 C_i^* C_j}{ (p_i+p_j^*) \left(\dfrac{\varepsilon_2 \rho_2^2}{(p_i - \i \alpha_2)(p_j^* + \i \alpha_2)} - 1\right)},\\
    &\xi_i=p_i \left(x - 3\varepsilon_2 \rho_2^2 t\right) + p_i^3 t +\xi_{i0}, \quad \Phi = \left(e^{\xi _{1}}, e^{\xi _2},\ldots, e^{\xi _{N}}\right)^T, \quad \Psi = \left(C_1, C_2,\ldots, C_N\right)^T.
\end{align}
Here, \(p_i, \xi_{i0}, C_i\) are complex parameters and \(\alpha_2\) is a real number.
\end{theorem}
\subsection{Derivation of dark-dark soliton solution to the coupled Hirota equation}
By the following transformations
\begin{align}\label{bdd_transformation}
u_1 = \frac{g_1}{f}\exp\left(\i (-\omega_1 t) \right), \quad u_2 = \frac{h_2}{f}\rho_2 \exp\left(\i (\alpha_2 x - \omega_2 t)\right),
\end{align}
where \(\omega_1 = 3 \rho_2^2 \varepsilon_2 \alpha_2, \ \omega_2 = \alpha_2^3 + 6\varepsilon_2\alpha_2 \rho_2^2\), the cHirota equation \eqref{chirota_1}--\eqref{chirota_2}
is converted into the following bilinear equations
\begin{align}
&\left(D_x^3 - D_t - 6\varepsilon_2 \rho_2^2 D_x+ 3\i \rho_2^2 \varepsilon_2 \alpha_2\right) g_1 \cdot f = 3 \i \rho_2^2 \varepsilon_2\alpha_2 s_{12} h_2^*,\label{BL_bright_dark_1}\\
& \left[D_x^3 - D_t + 3\i \alpha_2 D_x^2 - 3\left(\alpha_2^2 + 2 \varepsilon_2 \rho_2^2\right) D_x \right] h_2 \cdot f = -3\i \varepsilon_1 \alpha_2 s_{21} g_1^*,\label{BL_bright_dark_2}\\
& \left(D_x^2 - 2 \varepsilon_2 \rho_2^2 \right)f\cdot f  + 2 \varepsilon_1 |g_1|^2 + 2\rho_2^2 \varepsilon_2 |h_2|^2 = 0,\label{BL_bright_dark_3}\\
& D_x g_1 \cdot h_2 - \i g_1 h_2 \alpha_2 = -\i\alpha_2 s_{12} f.\label{BL_bright_dark_4}
\end{align}
Here \(s_{jk} = s_{kj}\).

To deduce the bright-dark soliton, one needs to start with the \(\tau\)-functions \(\tau_k^{0},\ \tau_k^{1},\ \bar{\tau}_k^{1}\),
\begin{align}\label{tau_bright_dark}
    \tau_k^{0} = \left|M_k\right|,
    \quad
    \tau_k^{1} = \begin{vmatrix}
        M_k & \Phi_k \\
        -\left(\bar{\Psi}\right)^T & 0
    \end{vmatrix},
    \quad 
    \bar{\tau}_k^{1} = \begin{vmatrix}
        M_k & \Psi \\
        -\left(\bar{\Phi}_k\right)^T & 0
    \end{vmatrix}
\end{align}
where \(M_k\) is a \(N \times N\) matrix, \(\Phi_k\), \(\bar{\Phi}_k\), \(\Psi\), and \(\bar{\Psi}\) are \(N\)-component vectors whose elements are defined as
\begin{align}
    &m_{ij}^k=\frac{1}{p_i+\bar{p}_j}\left(-\frac{p_i - a}{\bar{p}_j + a}\right)^k e^{\xi_i+\bar{\xi}_j}+\frac{\tilde{C}_i \bar{C}_j}{q_i+\bar{q}_j}e^{\eta_i+\bar{\eta}_j},\\
    &\Phi_k = \left( \left(\frac{p_1 - a}{-a}\right)^k e^{\xi _{1}}, \left(\frac{p_2 - a}{-a}\right)^k e^{\xi _2},\ldots, \left(\frac{p_N - a}{-a}\right)^k e^{\xi _{N}}\right)^T, \\
    &\bar{\Phi}_k = \left( \left(\frac{a}{\bar{p}_1 + a}\right)^k e^{\bar{\xi}_{1}},\left(\frac{a}{\bar{p}_2 + a}\right)^k e^{\bar{\xi}_2},\ldots ,\left(\frac{a}{\bar{p}_N + a}\right)^k e^{\bar{\xi}_{N}}\right)^T, \quad\\
    &\Psi =\left(\tilde{C}_1 e^{\eta _{1}}, \tilde{C}_2 e^{\eta _2},\ldots, \tilde{C}_N e^{\eta _{N}}\right)^T,\quad
    \bar{\Psi}=\left(\bar{C}_1 e^{\bar{\eta}_1}, \bar{C}_2 e^{\bar{\eta}_2},\ldots, \bar{C}_N e^{\bar{\eta}_{N}}\right)^T,\\
    &\xi_i=p_i x_1 + p_i^2 x_2 + p_i^3 x_3 +\frac{1}{p_i - a} x_{-1}+\xi_{i0},\quad \bar{\xi}_i=\bar{p}_i x_1 - \bar{p}_i^2 x_2 + \bar{p}_i^3 x_3 +\frac{1}{\bar{p}_i + a} x_{-1}+\bar{\xi}_{i0},\\
    &\eta_i=q_i y_1,\quad \bar{\eta}_i=\bar{q}_i y_1.
\end{align}
It cab be shown that the by  \(\tau\)-functions satisfy the following bilinear equations
\begin{align}
    & \left(D_{x_1}^3+3D_{x_1}D_{x_2}-4D_{x_3}\right)\tau_k^1 \cdot \tau_k^0=0,
    \label{KP_bright_dark_1}\\
    & \left(D_{x_{-1}}\left(D_{x_1}^2 - D_{x_2}\right) - 4 (D_{x_1} -a )\right)\tau_k^1\cdot \tau_k^0 - 4 a \tau_{k+1}^1 \tau_{k-1}^0 = 0,\label{KP_bright_dark_2}\\
    & D_{y_1}\left(D_{x_1}^2-D_{x_2}\right)\tau_k^1 \cdot \tau_k^0=0,\label{KP_bright_dark_3}\\
    & D_{y_1}D_{x_1}\tau_k^0 \cdot \tau_k^0=-2\tau_k^1 \bar{\tau}_k^1,\label{KP_bright_dark_4}\\
    & \left(D_{x_{-1}} D_{x_1}-2\right) \tau_k^0 \cdot \tau_k^0=-2 \tau_{k+1}^0 \tau_{k-1}^0, \label{KP_bright_dark_5} \\
    & \left(D_{x_1}^2-D_{x_2}+2 a D_{x_1}\right) \tau_{k+1}^0 \cdot \tau_k^0=0, \label{KP_bright_dark_6}\\
    & \left(D_{x_1}^3+3 D_{x_1} D_{x_2}-4 D_{x_3}+3 a\left(D_{x_1}^2+D_{x_2}\right)+6 a^2 D_{x_1}\right) \tau_{k+1}^0 \cdot \tau_k^0=0, \label{KP_bright_dark_7}\\
    & \left(D_{x_{-1}}\left(D_{x_1}^2-D_{x_2}+2 a D_{x_1}\right)-4D_{x_1}\right) \tau_{k+1}^0 \cdot \tau_k^0 =0, \label{KP_bright_dark_8}\\
    & \left(D_{y_1}\left(D_{x_1}^2 - D_{x_2} + 2a D_{x_1}\right)\right)\tau_{k+1}^0 \cdot \tau_k^0 + 4 a \tau_{k+1}^1 \bar{\tau}_k^1 = 0,\label{KP_bright_dark_9}\\
    & \left(D_{x_1} + a\right)\tau_{k+1}^0 \cdot \tau_k^1 = a \tau_{k+1}^1 \tau_k^0\,.\label{KP_bright_dark_10}
\end{align}

Let us start with the dimension reduction, which can be established by the following lemma.
\begin{lemma}
    By requiring the condition
    \begin{align}\label{dm_para_br_dk}
        q_i = \frac{1}{\varepsilon_1} \left(-p_i + \frac{\varepsilon_2 \rho_2^2}{p_i - a}\right), \quad \bar{q}_i = \frac{1}{\varepsilon_1} \left(-\bar{p}_i + \frac{\varepsilon_2 \rho_2^2}{\bar{p}_i + a}\right),
    \end{align}
    where \(i=1,2,\ldots,N\), the tau functions \eqref{tau_bright_dark} satisfy the differential relation
    \begin{equation}\label{dm_cond_br_dk}
        \varepsilon_1 \partial_y + \varepsilon_2 \rho_2^2 \partial_{x_{-1}} = \partial_{x_1}.
    \end{equation}
\end{lemma}
\begin{proof}
    For \(\tau_k^0\), we have
    \begin{align}
        &\tau_k^0=\prod_{n=1}^{N} e^{\xi_n+\bar{\xi}_n} \det\left(\mathfrak{m}_{i j}^{k} \right),
        \nonumber\\ 
        &\mathfrak{m}_{i j}^{k}= \frac{1}{p_i+\bar{p}_j}\left(-\frac{p_i - a}{\bar{p}_j + a}\right)^k+\frac{\tilde{C}_i \bar{C}_j}{q_i+\bar{q}_j}e^{\eta_i+\bar{\eta}_j-\xi_i-\bar{\xi}_j}. \label{frakm}
    \end{align}
    The term \(\prod_{n=1}^{N} e^{\xi_i+\eta_i}\) can be dropped in \eqref{KP_bright_dark_1}-\eqref{KP_bright_dark_10} due to the property of the \(D\)-operator in the bilinear equations. 
    Note that 
    \begin{align}
        \left(\varepsilon_1 \partial_y + \varepsilon_2 \rho_2^2 \partial_{x_{-1}} - \partial_{x_1}\right)\mathfrak{m}_{i j}^{k}=\left(\varepsilon_1 (q_i + \bar{q}_j) - \varepsilon_2 \rho_2^2 \left(\frac{1}{p_i-a} + \frac{1}{\bar{p}_j+a}\right) + (p_i + \bar{p}_j)\right)\frac{\tilde{C}_i \bar{C}_j}{q_i+\bar{q}_j}e^{\eta_i+\bar{\eta}_j-\xi_i-\bar{\xi}_j},\label{diff_rel_bright_dark}
    \end{align}
    Hence, by condition \eqref{dm_para_br_dk} we have \(\left(\varepsilon_1 \partial_y + \varepsilon_2 \rho_2^2 \partial_{x_{-1}}\right)\mathfrak{m}_{i j}^{k} = \partial_{x_1}\mathfrak{m}_{i j}^{k} \) and
    \begin{equation*}
        \left(\varepsilon_1 \partial_y + \varepsilon_2 \rho_2^2 \partial_{x_{-1}}\right)\tau_k^0=\sum_{i,j=1}^{N} \Delta_{ij}\left(\rho_1^2 \partial_r+\rho_2^2 \partial_s\right)\mathfrak{m}_{ij}^{\mathrm{n}}=\sum_{i,j=1}^{N} \frac{1}{c}\partial_x\mathfrak{m}_{ij}^{\mathrm{n}}=\frac{1}{c}\partial_x\tau_k^0.
    \end{equation*}
    
    For \(\tau_k^1\), we have
    \begin{align}
        \tau_k^1 = \prod_{n=1}^{N} e^{\xi_n+\bar{\xi}_n} \begin{vmatrix}
        \mathfrak{m}_{i j}^{k} & \left(\dfrac{p_i - a}{-a}\right)^k \\
        -\bar{C}_i e^{\bar{\eta}_j - \bar{\xi}_j} & 0
    \end{vmatrix},
    \end{align}
    where \(\mathfrak{m}_{i j}^{k}\) is defined as \eqref{frakm}.
    Similar to the previous case, by \eqref{dm_para_br_dk} we have \(\left(\varepsilon_1 \partial_y + \varepsilon_2 \rho_2^2 \partial_{x_{-1}}\right)e^{\bar{\eta}_j - \bar{\xi}_j} = \partial_{x_1}e^{\bar{\eta}_j - \bar{\xi}_j} \). Therefore, we have
    \begin{equation*}
        \left(\varepsilon_1 \partial_y + \varepsilon_2 \rho_2^2 \partial_{x_{-1}}\right)\tau_k^1 = \frac{1}{c}\partial_x\tau_k^1.
    \end{equation*}
    Similarly, we can show the proof of  \(\bar{\tau}_k^1\).
\end{proof}
With the help of dimension reduction, \eqref{KP_bright_dark_2} and \eqref{KP_bright_dark_3},  \eqref{KP_bright_dark_8} and \eqref{KP_bright_dark_9},
\eqref{KP_bright_dark_4} and \eqref{KP_bright_dark_5} lead to
\begin{align}
    &\left(D_{x_1}^3 - D_{x_1}D_{x_2} - 4\varepsilon_2 \rho_2^2 (D_{x_1} -a )\right)\tau_k^1\cdot \tau_k^0 - 4 a \varepsilon_2 \rho_2^2 \tau_{k+1}^1 \tau_{k-1}^0 = 0, \label{KP_bright_dark_2.1}\\
    &\left(D_{x_1}^3-D_{x_1}D_{x_2}+2 a D_{x_1}^2-4\varepsilon_2\rho_2^2 D_{x_1}\right) \tau_{k+1}^0 \cdot \tau_k^0 + 4 a \varepsilon_1 \tau_{k+1}^1 \bar{\tau}_k^1=0\,, \label{KP_bright_dark_8.1} \\
    &\left(D_{x_1}^2 - 2 \varepsilon_2 \rho_2^2 \right)\tau_k^0 \cdot \tau_k^0 = -2 \varepsilon_1 \tau_k^1 \bar{\tau}_k^1 -2\varepsilon_2 \rho_2^2 \tau_{k+1}^0 \tau_{k-1}^0, \label{KP_bright_dark_4.1}
\end{align}
respectively.
To eliminate the extra \(D_{x_1}D_{x_2}\) term in \eqref{KP_bright_dark_2.1} and \eqref{KP_bright_dark_8.1} with the use of \eqref{KP_bright_dark_1} and \eqref{KP_bright_dark_7}, one obtains
\begin{align}
    &\left(D_{x_1}^3 - D_{x_3}- 6\varepsilon_2 \rho_2^2 D_{x_1} + 3a \varepsilon_2 \rho_2^2\right)\tau_k^1\cdot \tau_k^0 - 3 a \varepsilon_2 \rho_2^2 \tau_{k+1}^1 \tau_{k-1}^0 = 0, \label{KP_bright_dark_1.2}\\
    &\left(D_{x_1}^3 - D_{x_3}+ 3aD_{x_1}^2 -3\left(-a^2+2\varepsilon_2\rho_2^2\right) D_x\right) \tau_{k+1}^0 \cdot \tau_k^0 + 3 a \varepsilon_1 \tau_{k+1}^1\bar{\tau}_k^1=0. \label{KP_bright_dark_7.2}
\end{align}
after a further gauge transformation \(x_1 \to x_1 - 3\varepsilon_2 \rho_2^2 x_3\).

Note that by performing the dimension reduction, \(x_{-1}, y_1, x_2\) become dummy variables, and we can set them to zero. 
Next, let us consider the complex conjugate reduction.
    Let \(x_1 = x, x_3 = t\) and \( a = \i\alpha_2 \in \i \mathbb{R} \), under the constraints in Lemma 4.2 and the following parameter constraints
    \begin{equation}
       \bar{p}_i^* = p_i, \quad \bar{\xi}_{i,0}^* = \xi_{i0}, \quad C_i = \left(\tilde{C}_i\right)^* = \bar{C}_i,
    \end{equation}
    where \( i = 1, 2, \ldots, N \),  it is shown that the \(\tau\) functions exhibit the following complex conjugate relationships:
    \begin{equation}\label{bd_complex_reduction_res}
        \tau_{-k}^0 = \left(\tau_{k}^0\right)^*, \quad \bar{\tau}_{-k}^1 = \left(\tau_k^1\right)^*.
    \end{equation}

With the help of complex conjugate relation, by taking \(k=0\) in \eqref{KP_bright_dark_1.2}, \eqref{KP_bright_dark_7.2}, \eqref{KP_bright_dark_4.1} and \eqref{KP_bright_dark_10} and setting \(\tau_0^0 = f, \tau_0^1 = g_1, \bar{\tau}_0^1 = g_1^*, \tau_1^0 = h_2, \tau_{-1}^0 = h_2^*, \tau_1^1 = s_{12} = s_{21}\), we get the bilinear form of the cHirota equation \eqref{BL_bright_dark_1}-\eqref{BL_bright_dark_4}. The theorem is approved.

\subsection{Dynamics of the bright-dark soliton solution}
Taking \(N = 1\) in \cref{theorem:bd_solution}, the solution \(u_1\) and \(u_2\) are expressed as
\begin{align}
    \begin{split}
        u_1 &= \frac{C_1 \exp(\xi_1) (p_1 + p_1^*)}{\exp(\xi_1 + \xi_1^*) + d_1} \exp\left(\i \theta_1\right) \\
        &= C_1 \exp\left(\i \theta_1\right) \Re(p_1) \frac{\exp(\i \Im (\xi_1))}{\sqrt{d_1}} \sech \left(\Re(\xi_1) -\log\left(\sqrt{d_1}\right)\right), \label{bd_n=1_b}\\
    \end{split}\\
    \begin{split}
        u_2 &= \frac{\exp(\xi_1 + \xi_1^*)\left(-\dfrac{p_1 - \i \alpha_2}{p_1^* + \i \alpha_2}\right) + d_1}{\exp(\xi_1 + \xi_1^*) + d_1} \rho_2\exp\left(\i \theta_2\right) \\
        &=\rho_2\frac{\exp(\i \theta_2)}{p_1^* + \i \alpha_2} \left[\i (\alpha_2 - \Im(p_1)) - \Re(p_1) \tanh \left(\Re(\xi_1) - \log\left(\sqrt{d_1}\right)\right)\right], \label{bd_n=1_d}
    \end{split}
\end{align}
where \(\xi_1 = p_1 \left(x - 3\varepsilon_2 \rho_2^2 t\right) + p_1^3 t + \xi_{10}\) and
\begin{align}\label{bd_theta_12}
    d_1 = \frac{\varepsilon_1 |C_1|^2 |p_1 - \i \alpha_2|^2}{\varepsilon_2 \rho_2^2 - |p_1 - \i \alpha_2|^2}, \quad
    \theta_1 = -3 \rho_2^2 \varepsilon_2 \alpha_2 t, \quad \theta_2 = \alpha_2 x - (\alpha_2^3 + 6\varepsilon_2\alpha_2 \rho_2^2) t.
\end{align}
One may recognize the bright and dark soliton from the expression of \(u_1\) and \(u_2\). Note that 
\[\Re(\xi_1) = \Re(p_1) \left(x - \left(3\varepsilon_2 \rho_2^2 - \left[\Re(p_1)\right]^2 + 3\left[\Im(p_1)\right]^2\right) t\right),\]
hence both the bright and dark solitons move with a speed \(3\varepsilon_2 \rho_2^2 - \left[\Re(p_1)\right]^2 + 3\left[\Im(p_1)\right]^2\). 
To avoid the singular solution, we require \(d_1 > 0\).
The first order bright-dark soliton solution is illustrate in \cref{figure:bd_N=1}.

\begin{figure}[!ht]
    \centering
    \subfigure[]{\label{figure:bd_N=1(a)}
        \includegraphics[width=60mm]{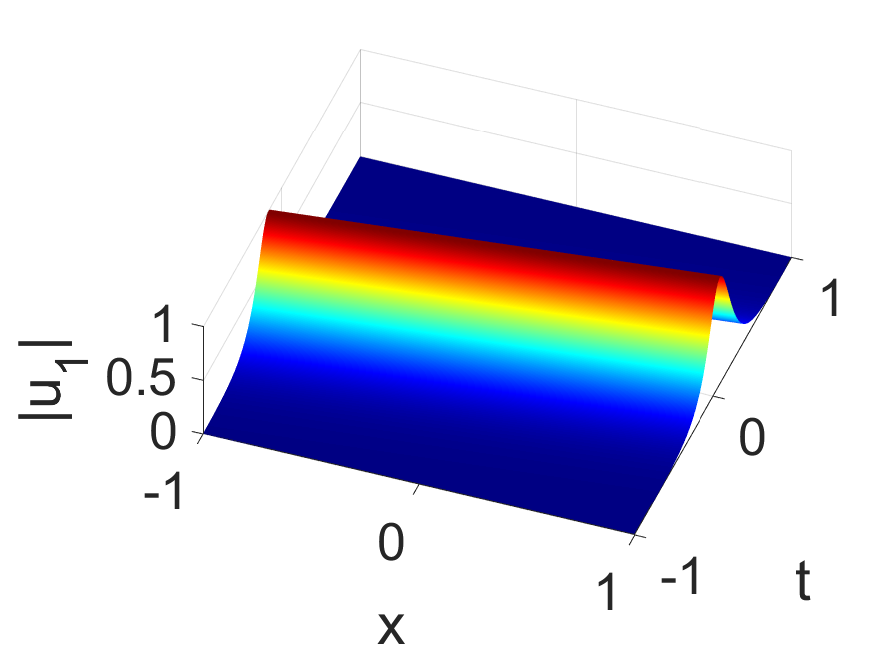}
    }
    \subfigure[]{\label{figure:bd_N=1(b)}
        \includegraphics[width=60mm]{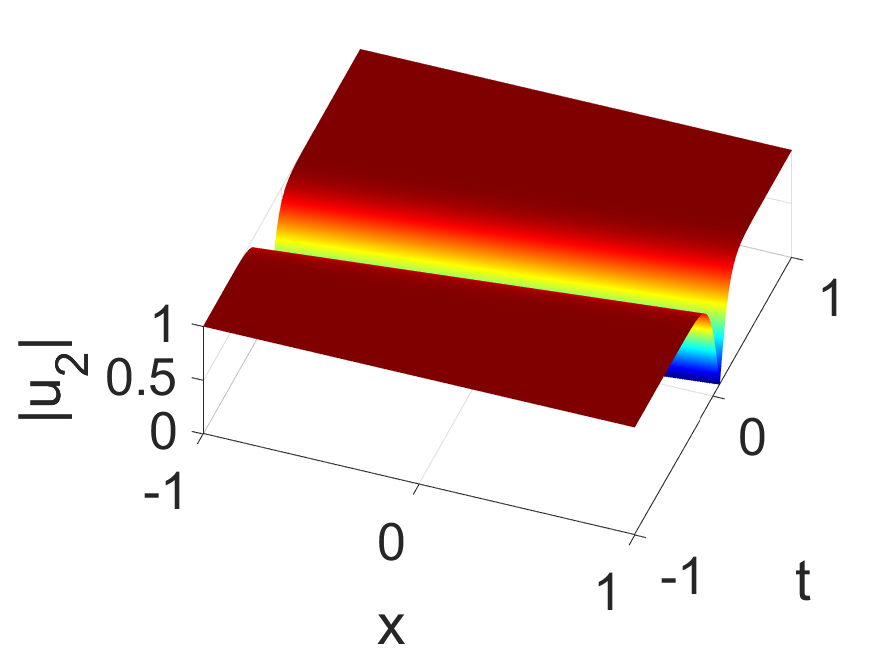}
    }
    \caption{One-bright-one-dark soliton solution to the cHirota equation with parameters \(p_1= 1 + \i, C_1 = 1 + \i, \rho_2 = 1, \alpha_2 =1, \varepsilon_1 = 1, \varepsilon_2 = 2\).}\label{figure:bd_N=1}
\end{figure}
By taking \(N = 2\) in \cref{theorem:bd_solution}, we have
\begin{align*}
    g_1 &= \begin{vmatrix}
        \dfrac{\exp(\xi_1 + \xi_1^*)}{p_1 + p_1^*} + \dfrac{|C_1|^2}{q_1 + q_1^*} & \dfrac{\exp(\xi_1 + \xi_2^*)}{p_1 + p_2^*} + \dfrac{C_1^* C_2}{q_1 + q_2^*} & \exp(\xi_1) \\
        \dfrac{\exp(\xi_2 + \xi_1^*)}{p_2 + p_1^*} + \dfrac{C_2^* C_1}{q_2 + q_1^*} & \dfrac{\exp(\xi_2 + \xi_2^*)}{p_2 + p_2^*} + \dfrac{|C_2|^2}{q_2 + q_2^*} & \exp(\xi_2) \\ 
         -C_1 & -C_2 &  0
    \end{vmatrix}, \\
    h_2 &= \begin{vmatrix}
        \dfrac{\exp(\xi_1 + \xi_1^*)}{p_1 + p_1^*}\left(-\dfrac{p_1 - \i \alpha_2}{p_1^* + \i \alpha_2}\right) + \dfrac{|C_1|^2}{q_1 + q_1^*} & \dfrac{\exp(\xi_1 + \xi_2^*)}{p_1 + p_2^*}\left(-\dfrac{p_1 - \i \alpha_2}{p_2^* + \i \alpha_2}\right) + \dfrac{C_1^* C_2}{q_1 + q_2^*} \\
        \dfrac{\exp(\xi_2 + \xi_1^*)}{p_2 + p_1^*}\left(-\dfrac{p_2 - \i \alpha_2}{p_1^* + \i \alpha_2}\right) + \dfrac{C_2^* C_1}{q_2 + q_1^*} & \dfrac{\exp(\xi_2 + \xi_2^*)}{p_2 + p_2^*}\left(-\dfrac{p_2 - \i \alpha_2}{p_2^* + \i \alpha_2}\right) + \dfrac{|C_2|^2}{q_2 + q_2^*} \\
    \end{vmatrix}, \\
    f &= \begin{vmatrix}
        \dfrac{\exp(\xi_1 + \xi_1^*)}{p_1 + p_1^*} + \dfrac{|C_1|^2}{q_1 + q_1^*} & \dfrac{\exp(\xi_1 + \xi_2^*)}{p_1 + p_2^*} + \dfrac{C_1^* C_2}{q_1 + q_2^*} \\
        \dfrac{\exp(\xi_2 + \xi_1^*)}{p_2 + p_1^*} + \dfrac{C_2^* C_1}{q_2 + q_1^*} & \dfrac{\exp(\xi_2 + \xi_2^*)}{p_2 + p_2^*} + \dfrac{|C_2|^2}{q_2 + q_2^*} \\
    \end{vmatrix},
\end{align*}
where \(\theta_1\), \(\theta_2\) are defined by \eqref{bd_theta_12} and 
\begin{align*}
    q_1 = \frac{1}{\varepsilon_1} \left(-p_1 + \frac{\varepsilon_2 \rho_2^2}{p_1 - \i \alpha_2}\right), \quad q_2 = \frac{1}{\varepsilon_1} \left(-p_2 + \frac{\varepsilon_2 \rho_2^2}{p_2 - \i \alpha_2}\right).
\end{align*}
Note that if we take either \(C_1 = 0\) or \(C_2 = 0\), we will get degenerated first order soliton solution. For example, when \(C_1 = 0\), we have
\begin{align*}
    u_1 &= \frac{\exp(\i \theta_1)}{C_2^*} \frac{(q_2 + q_2^*)(p_2 - p_1)}{p_2 + p_1^*} \frac{\exp(\i \Im(\Re(\xi_2)))}{2 \psi} \sech \Bigl(\Re(\xi_2) + \log(\psi)\Bigr),\\
    u_2 &= \rho_2\exp(\i \theta_2)\left(-\frac{p_1 - \i \alpha_2}{p_1^* + \i \alpha_2}\right)\biggl(\i \Bigl(\alpha_2 - \Im(p_2)\Bigr) - \Re(p_2)\tanh\Bigl(\Re(\xi_2) + \log(\psi)\Bigr) \biggr),
\end{align*}
where 
\begin{align*}
    \psi = \sqrt{\left(\frac{1}{(p_1+p_1^*)(p_2+p_2^*)} - \frac{1}{(p_1+p_2^*)(p_1^*+p_2)}\right)\frac{(p_1+p_1^*)(q_2+q_2^*)}{|C_2|^2}}.
\end{align*}
We call above solution is degenerated first order solution based on its structure in comparison to \eqref{bd_n=1_b} and \eqref{bd_n=1_d}.

Next, we will investigate the asymptotic behavior of the two-soliton solution. In this analysis, we assume that \(C_1 \neq 0\) and \(C_2 \neq 0\) to avoid solution degeneration. We label the soliton determined by \(\xi_1 + \xi_1^*\) as soliton 1, and the soliton determined by \(\xi_2 + \xi_2^*\) as soliton 2. Additionally, we assume that soliton 1 is located to the right of soliton 2 before their collision.

\begin{enumerate}[(1)]
    \item Before collision, i.e., \(t \rightarrow - \infty\), this gives
    
    Soliton 1 (\(\xi_1 + \xi_1^* \approx 0,\ \xi_2 + \xi_2^* \rightarrow - \infty\))
    \begin{align*}
        u_1 &\simeq C_1 \exp(\i \theta_1) \frac{(p_1+p_1^*)(q_1^* - q_2^*)}{q_2 + q_1^*} \frac{\exp(\i \Im(\xi_1))}{2\varphi_1}\sech \Bigl(\Re(\xi_1) - \log (\varphi_1)\Bigr),\\
        u_2 &\simeq \rho_2\frac{\exp(\i \theta_2)}{p_1^* + \i \alpha_2} \biggl(\i \Bigl(\alpha_2 - \Im(p_1)\Bigr) - \Re(p_1)\tanh\Bigl(\Re(\xi) - \log(\varphi_1)\Bigr) \biggr),
    \end{align*}
    where the phase term \(\varphi_1\) is given by
    \begin{align*}
        \varphi_1 = \sqrt{|C_1|^2(p_1+p_1^*)(q_2+q_2^*)\left(\frac{1}{(q_1+q_1^*)(q_2+q_2^*)}-\frac{1}{(q_2+q_1^*)(q_1+q_2^*)}\right)}.
    \end{align*}

    Soliton 2 (\(\xi_2 + \xi_2^* \approx 0,\ \xi_1 + \xi_1^* \rightarrow + \infty\))
   \begin{align*}
        u_1 &\simeq \frac{\exp(\i \theta_1)}{C_2^*} \frac{(q_2 + q_2^*)(p_2 - p_1)}{p_2 + p_1^*} \frac{\exp(\i \Im(\xi_2))}{2 \psi_1} \sech \Bigl(\Re(\xi_2) + \log(\psi_1)\Bigr), \\
        u_2 &\simeq \rho_2\exp(\i \theta_2)\left(-\frac{p_1 - \i \alpha_2}{p_1^* + \i \alpha_2}\right)\biggl(\i \Bigl(\alpha_2 - \Im(p_2)\Bigr) - \Re(p_2)\tanh\Bigl(\Re(\xi_2) + \log(\psi_1)\Bigr) \biggr),
    \end{align*}
    where the phase term \(\psi_1\) is given by
    \begin{align*}
        \psi_1 = \sqrt{\left(\frac{1}{(p_1+p_1^*)(p_2+p_2^*)} - \frac{1}{(p_1+p_2^*)(p_1^*+p_2)}\right)\frac{(p_1+p_1^*)(q_2+q_2^*)}{|C_2|^2}}.
    \end{align*}

    \item After collision, i.e., \(t \rightarrow + \infty\), this gives

    Soliton 1 (\(\xi_1 + \xi_1^* \approx 0,\ \xi_2 + \xi_2^* \rightarrow + \infty\))
    \begin{align*}
        u_1 &\simeq \frac{\exp(\i \theta_1)}{C_1^*} \frac{(q_1 + q_1^*)(p_1 - p_2)}{p_1 + p_2^*} \frac{\exp(\i \Im(\xi_1))}{2 \varphi_2} \sech \Bigl(\Re(\xi_1)  + \log(\varphi_2)\Bigr), \\
        u_2 &\simeq \rho_2\exp(\i \theta_2)\left(-\frac{p_2 - \i \alpha_2}{p_2^* + \i \alpha_2}\right)\biggl(\i \Bigl(\alpha_2 - \Im(p_1)\Bigr) - \Re(p_1)\tanh\Bigl(\Re(\xi_1) + \log(\varphi_2)\Bigr) \biggr),
    \end{align*}
    where the phase term \(\varphi_2\) is given by
    \begin{align*}
        \varphi_2 = \sqrt{\left(\frac{1}{(p_1+p_1^*)(p_2+p_2^*)} - \frac{1}{(p_1+p_2^*)(p_1^*+p_2)}\right)\frac{(p_2+p_2^*)(q_1+q_1^*)}{|C_1|^2}}.
    \end{align*}

    Soliton 2 (\(\xi_2 + \xi_2^* \approx 0,\ \xi_1 + \xi_1^* \rightarrow - \infty\))
    \begin{align*}
        u_1 &\simeq C_2 \exp(\i \theta_1) \frac{(p_2+p_2^*)(q_2^* - q_1^*)}{q_1 + q_2^*} \frac{\exp(\i \Im(\xi_2))}{2\psi_2}\sech \Bigl(\Re(\xi_2) - \log (\psi_2)\Bigr),\\
        u_2 &\simeq \rho_2\frac{\exp(\i \theta_2)}{p_2^* + \i \alpha_2} \biggl(\i \Bigl(\alpha_2 - \Im(p_2)\Bigr) - \Re(p_2)\tanh\Bigl(\Re(\xi) - \log(\psi_2)\Bigr) \biggr),
    \end{align*}
    where the phase term \(\varphi_2\) is given by
    \begin{align*}
        \psi_2 = \sqrt{|C_2|^2(p_2+p_2^*)(q_1+q_1^*)\left(\frac{1}{(q_1+q_1^*)(q_2+q_2^*)}-\frac{1}{(q_2+q_1^*)(q_1+q_2^*)}\right)}.
    \end{align*}
\end{enumerate}
Therefore, the asymptotic analysis shows that there is a phase shift between the solitons before and after the collision. Specifically, for soliton 1, the phase term \(-\log(\varphi_1) \neq \log(\varphi_2)\), and for soliton 2, \(-\log(\psi_1) \neq \log(\psi_2)\). This phenomenon is evident in the illustration of our two-soliton solution, as shown in \cref{figure:bd_N=2}.

\begin{figure}[!ht]
    \centering
    \subfigure[]{\label{figure:bd_N=2(a)}
        \includegraphics[width=60mm]{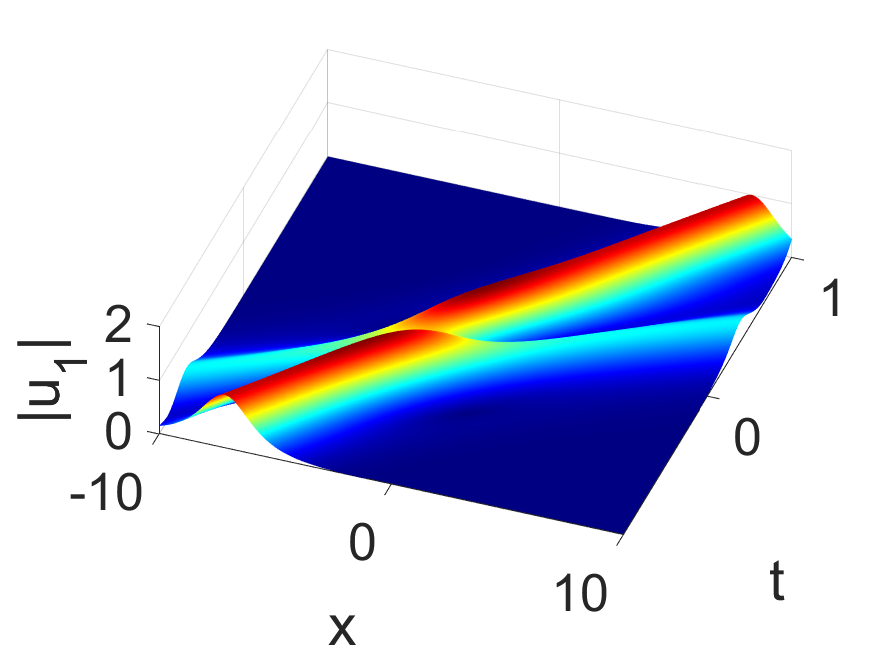}
    }
    \subfigure[]{\label{figure:bd_N=2(b)}
        \includegraphics[width=60mm]{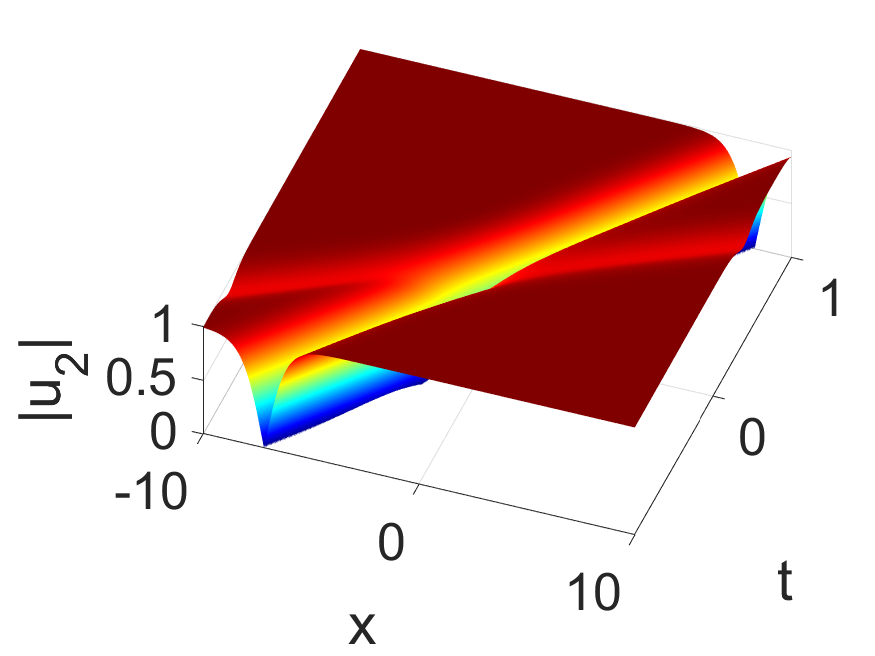}
    }
    \caption{Two-bright-two-dark soliton solution to the cHirota equation with parameters \(p_1= 1 + \i, p_2 = 1/2 + 2\i, C_1 = 1 + \i, C_2 = 2+\i, \rho_2 = 1, \alpha_2 =1, \varepsilon_1 = 1, \varepsilon_2 = 2\).}\label{figure:bd_N=2}
\end{figure}

\section{Bright soliton solution to the Sasa-Satsuma equation}\label{section:b_sasa}
As we mentioned previously, the bright soliton solution to SS equation \eqref{sse} can be obtained by further requiring \(u_1^* = u_2\) and \(\varepsilon = \varepsilon_1 = \varepsilon_2\).  To archive this, we have the following lemma.
\begin{lemma}
    The reduction condition \(u_1^* = u_2\) holds under the following parameter restrictions
    \begin{equation}\label{cond_bright_sse}
        \varepsilon = \varepsilon_1 = \varepsilon_2, \quad p_{N+1-i}^* = p_i, \quad \xi_{N+1-i,0}^* = \xi_{i,0}, \quad D_{N+1-i}^* = C_i, \quad D_i^* = C_{N+1-i}.
    \end{equation}
\end{lemma}
\begin{proof}
    By employing \eqref{cond_bright_sse} on solution determined by \cref{theorem:bb_solution}, we have
    \begin{align*}
        \xi_{i}^* &= p_i^* x + \left(p_i^*\right)^3 t + \xi_{i,0}^* = p_{N+1-i} + p_{N+1-i}^3 t + \xi_{N+1-i},\\
        m_{i,j}^* &= \frac{1}{p_i^*+p_j}\left(e^{\xi_i^*+\xi_j} - \varepsilon C_i C_j^* - \varepsilon D_i D_j^* \right) \\
        &= \frac{1}{p_{N+1-i}+p_{N+1-j}}\left(e^{\xi_i^*+\xi_j} - \varepsilon D_{N+1-i}^* D_{N+1-j} - \varepsilon C_{N+1-i}^* C_{N+1-j} \right) = m_{N+1-i,N+1-j}.
    \end{align*}
    Above complex conjugate relations allow us to rewrite \(g_1^*\) as
    \begin{align*}
        g_1^* = \begin{vmatrix}
            m_{i,j}^* & e^{\xi_i^*} \\ -C_j^* & 0
        \end{vmatrix} 
        = \begin{vmatrix}
            m_{N+1-i,N+1-j} & e^{\xi_{N+1-i}} \\ -D_{N+1-j} & 0
        \end{vmatrix},
    \end{align*}
    now, let us perform the following row operations: switching \(i\)-th row with \((N+1-i)\)-th row, and \(j\)-th column with \((N+1-j)\)-th column in above determinant, where \(i = 1,\ldots, \lfloor N/2 \rfloor\), this results
    \begin{align*}
        g_1^* =  \begin{vmatrix}
            m_{i,j} & e^{\xi_i} \\ -D_j & 0
        \end{vmatrix} = g_2.
    \end{align*}
   Similarly, can show \(f^* = f\), hence we have \(u_1^*  = u_2\). 
\end{proof}

This shows that solution from \cref{theorem:bb_solution} solves the SS equation \eqref{sse} with additional restriction \eqref{cond_bright_sse}. 
In what follows, we will summarize bright soliton solution to the SS equation into the following theorem.
\begin{theorem}\label{theorem:bright_ss}
    Equation \eqref{sse} admits the bright soliton solutions given by \(u=g/f,\) with \(f,\ g\) defined as
    \begin{align}
        f=|M|,\quad
        g=\begin{vmatrix}
            M & \Phi \\ -\left(\bar{\Psi}\right)^T & 0
            \end{vmatrix},
    \end{align}
    where \(M\) is \(N\times N\) matrix, \(\Phi \), \(\bar{\Psi} \), are \(N\)-component vectors whose elements are defined respectively as
    \begin{align}
        &m_{ij}=\frac{1}{p_i+p_j^*}\left(e^{\xi_i+\xi_j^*} - \varepsilon C_i^* C_j - \varepsilon C_{N+1-i} C_{N+1-j}^* \right),\quad \xi_i=p_i x + p_i^3 t+\xi_{i0},\\
        &\Phi = \left(e^{\xi_{1}},e^{\xi_2},\ldots, e^{\xi_N}\right)^T,\quad
        \bar{\Psi}=\left(C_1, C_2,\ldots, C_N\right)^T.
    \end{align}
    Here, \(p_i^* = p_{N+1-i}\), \(\xi_{i0}^* = \xi_{N+1-i,0}\), \(C_i\) are complex parameters, where \(i = 1, 2, \ldots, N\).
\end{theorem}

\subsection{Dynamics of the bright soliton solution to the Sasa-Satsuma equation}
When \(N=1\),  one-soliton solution to the SS equation \eqref{sse} can be simplified as
\begin{equation*}
    u = \frac{2 p_1 C_1 \exp \left(\xi_1\right)}{\exp \left(2\xi_1\right) - 2\varepsilon |C_1|^2 } = \frac{C_1 p_1}{\sqrt{-2\varepsilon|C_1|^2}} \sech\left(\xi_1 - \log\left(\sqrt{-2\varepsilon|C_1|^2}\right)\right), 
\end{equation*}
where \(\xi_1 = p_1 x + p_1^3 t +\xi_{10}\) and \(p_1, \xi_{10} \in \mathbb{R}\). To guarantee  a regular solution, we require \(\varepsilon < 0\). 
We remark here that this solution also satisfies the Hirota equation \eqref{hirota_cmkdv_form}. 
It is noted that \(u^2 u_x^* = |u|^2 u_x\) for the above one-soliton solution. This implies that the solution satisfies the equation
\begin{equation*}
    u_t = u_{xxx} - 12\varepsilon |u|^2 u_x,
\end{equation*}
which corresponds to equation \eqref{hirota_cmkdv_form} with \(\alpha = 4\varepsilon\) and \(\gamma = -1\).

For \(N=2\), we get the second order bright soliton solution to the SS equation. It turns out in this case, one hump solution and one oscillated soliton solution are obtained rather than two-soliton solution. Denoting \(p_1 = a + \i b\), the solution to the SS equation is expressed as \(u = g/f\), where \(\xi_1 = p_1 x + p_1^3 t + \xi_{1,0}\) and
\begin{align*}
    f ={} & \frac{b^2}{4 a^2 (a^2 + b^2)} \exp(2\xi_1+2\xi_1^*) - \varepsilon\frac{|C_1|^2 + |C_2|^2}{2a^2} \exp(\xi_1+\xi_1^*) \\
    &+ \frac{\varepsilon}{2(a^2+b^2)} \left(C_1 C_2^* \exp(2\xi_1) + C_1^* C_2 \exp(2\xi_1^*)\right) \\
    & - \varepsilon^2 \frac{a^2 (|C_1|^2 - |C_2|^2)^2 + b^2 (|C_1|^2 + |C_2|^2)^2}{4 a^2 (a^2 + b^2)},\\
    g ={} & \frac{1}{2a} \left(C_1\exp(\xi_1) + C_2\exp(\xi_1^*)\right) \left(\exp(\xi_1 + \xi_1^*) - \varepsilon|C_1|^2 - \varepsilon|C_2|^2\right) \\
    & - \frac{C_1}{2a+2\i b} \exp(\xi_1 + \xi_1^*)\left(\exp(\xi_1) - 2\varepsilon C_1^* C_2 \exp(-\xi_1)\right)\\
    & - \frac{C_2}{2a-2\i b} \exp(\xi_1 + \xi_1^*)\left(\exp(\xi_1^*) - 2\varepsilon C_1 C_2^* \exp(-\xi_1^*)\right).
\end{align*}

Note that if \(C_1, C_2 \neq 0\), one term of \(g\) can be rewritten as
\begin{align*}
    C_1\exp(\xi_1) + C_2\exp(\xi_1^*) &= \sqrt{C_1 C_2} \exp(\Re (\xi_1)) \cos\left(\Im (\xi_1) - \i \log \left(\sqrt{\frac{C_1}{C_2}}\right)\right),
\end{align*}
which is periodic if \(\Im (\xi_1) \neq 0\), and we get the so called oscillated soliton solution, see \cref{fig:formation_oscillated(c)}. If \(\Im (\xi_1) = 0\), i.e., \(b = 0, \xi_{10}\in \mathbb{R}\), we have
\begin{equation*}
    u = \frac{2 a (C_1-C_2) \left(|C_2| ^2 - |C_1|^2\right)\exp(\xi_1)}{\varepsilon \left(| C_1| ^2 - |C_2|^2\right)^2-2 |C_1-C_2|^2 \exp(2\xi_1)},
\end{equation*}
which is equivalent to the bright soliton when \(N=1\).

Moreover, if either \(C_1 =0\) or \(C_2=0\), we get the a single or double hump solution. Assuming \(C_2 = 0\) without loss of generality, the solution can be rewritten as
\begin{align*}
    u = \frac{2 a C_1 (a - \i b)\exp(\xi_1)  \left(\i b \exp(\xi_1 + \xi_1^*) -\varepsilon(a + \i b) |C_1|^2\right)}{b^2 \exp(2\xi_1+2\xi_1^*) - 2\varepsilon(a^2 + b^2) |C_1|^2 \exp(\xi_1+\xi_1^*) + \varepsilon^2(a^2 + b^2) |C_1|^4},
\end{align*}
which is the bright soliton solution obtained originally by Sasa and Satsuma  \cite{sasa1991new}. 

In what follow, we give a detailed analysis by denoting \(y = \exp(\xi_1+\xi_1^*)\). A simple calculation gives
\begin{align*}
    |u|^2 &= \frac{4 a (a^2 + b^2) |C_1|^2 y \left(b^2 y^2 -2 b^2 \varepsilon |C_1|^2 y + \varepsilon^2(a^2 + b^2) |C_1|^4\right)}{\left(b^2 y^2 - 2\varepsilon(a^2 + b^2) |C_1|^2 y + \varepsilon^2(a^2 + b^2) |C_1|^4\right)^2},\\
    \frac{\partial |u|^2}{\partial x} &= \frac{8 a^3 (a^2 + b^2) |C_1|^4 y \left(\varepsilon^2(a^2 + b^2) |C_1|^4 - b^2 y^2\right)\left(b^2 y^2 - 2\varepsilon (a^2-b^2) |C_1|^2y + \varepsilon^2(a^2+b^2) |C_1|^4\right)}{\left(b^2 y^2 - 2\varepsilon(a^2 + b^2) |C_1|^2 y + \varepsilon^2(a^2 + b^2) |C_1|^4\right)^3}.
\end{align*}
Above \(u\) is singular if \(\varepsilon>0\).
Hence let us consider the case of \(\varepsilon < 0\), the equation \(\partial |u|^2/\partial_x   = 0\) will have five roots 
\begin{align*}
    &y_1 = 0,\quad y_2 = \varepsilon\frac{\sqrt{a^2+b^2}}{b} |C_1|^2,\quad y_3 = -\varepsilon\frac{\sqrt{a^2+b^2}}{b} |C_1|^2, \\
    &y_4 = \frac{\varepsilon(b^2 - a^2) - \sqrt{a^2 \varepsilon^2(a^2 - 3 b^2)}}{b^2} |C_1|^2,\quad y_5 = \frac{\varepsilon(b^2 - a^2) + \sqrt{a^2 \varepsilon^2(a^2 - 3 b^2)}}{b^2} |C_1|^2\,.
\end{align*}
By assuming \(\text{Im}(p_1) = b > 0\) without loss of generality, the hump solution can be divided into two types 
\begin{enumerate}
    \item When \(a^2 > 3 b^2\), a double-hump solution is obtained, see \cref{fig:formation_two_hump(c)}. In this case, three extremes \(y_3, y_4, y_5\) exist. It can be proven that \(y_4 < y_3 < y_5\), i.e., the two humps are located on the lines \(\xi_1 + \xi_1^* = \log y_4\) and \(\xi_1 + \xi_1^* = \log y_5\) and they are symmetric with respect to \(\xi_1 + \xi_1^* = \log y_3\).
    \item When \(a^2 \leq 3 b^2\), a single-hump solution is obtained, see \cref{fig:formation_one_hump(c)}. In this case, only one maximum exists at \(y_3\).
\end{enumerate}
Since \(\xi_1+\xi_1^* = 2a (x - (3b^2 - a^2) t)\), which implies the velocity of the hump solution is \( 3 \left[\Im(p_1)\right]^2 - \left[\Re(p_1)\right]^2 = 3b^2 - a^2\). Therefore, a single-hump solution always moves to the right while a double-hump soliton always moves to the left.
Similar result was obtained by Xu in the coupled SS equation previously \cite{xu2013single}.
\begin{figure}[!ht]
    \centering
    \subfigure[]{\label{fig:formation_oscillated(c)}
        \includegraphics[width=50mm]{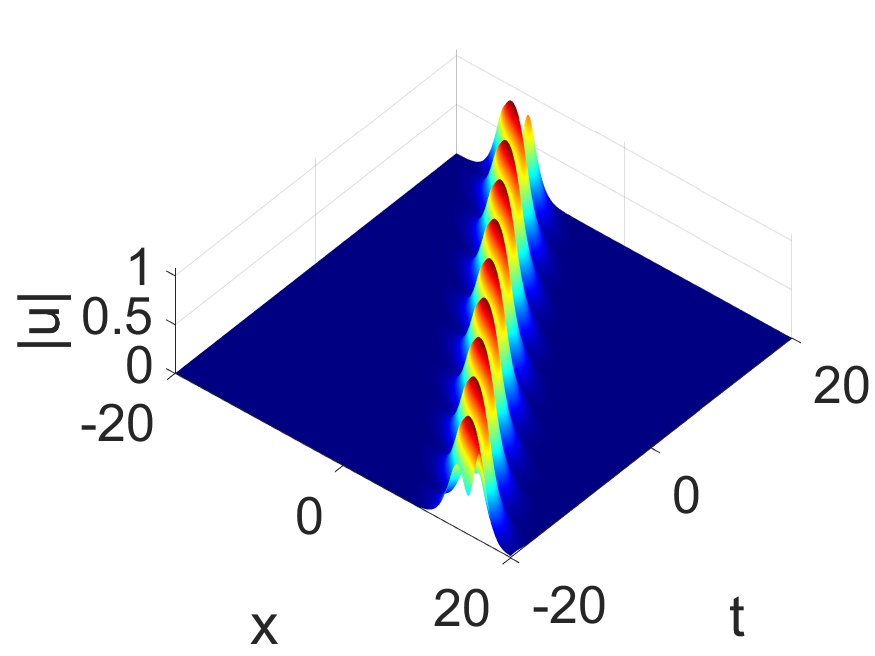}
    }
     \subfigure[]{\label{fig:formation_two_hump(c)}
        \includegraphics[width=50mm]{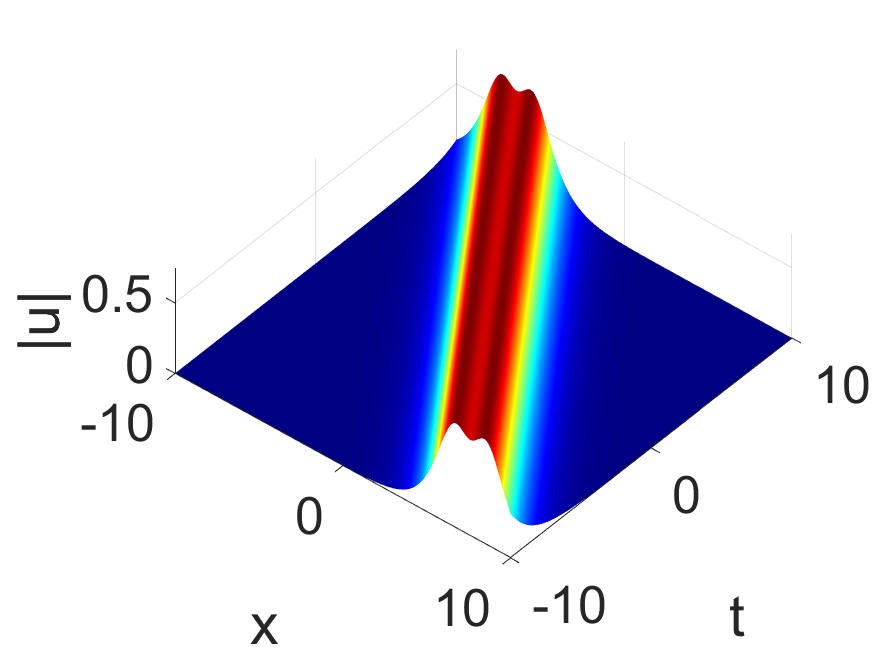}
    }
        \subfigure[]{\label{fig:formation_one_hump(c)}
        \includegraphics[width=50mm]{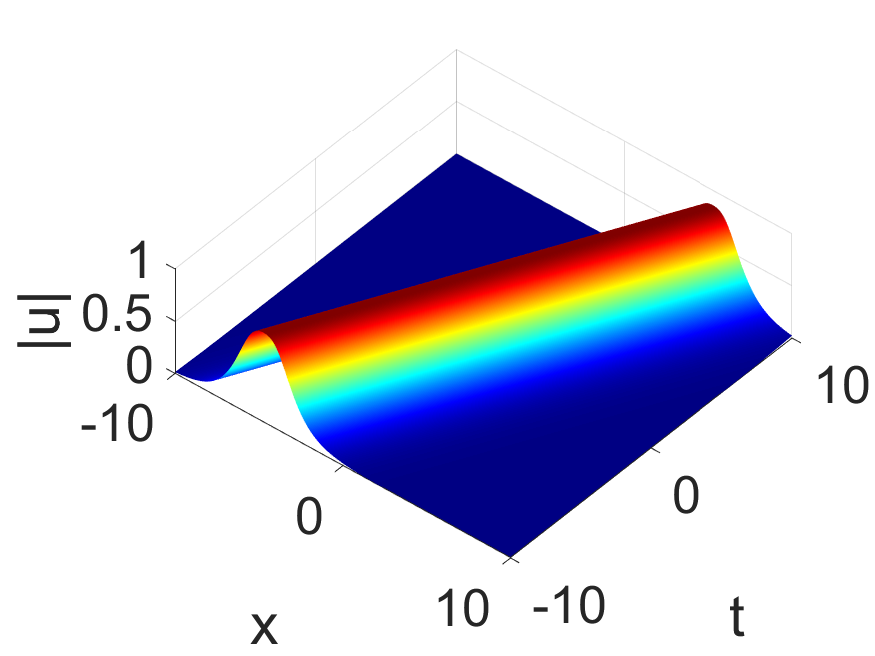}
    }
    \caption{Three types of bright soliton for \(N=2\); (a) oscillating type while \(p_1 = 1-\i/\sqrt{10}, C_1 = 1, C_2 = 1, \xi_{1,0} = 0, \epsilon = 1\); (b) double-hump while \(p_1 = 1-\i/\sqrt{10}, C_1 = 1, C_2 = 0, \xi_{1,0} = 0, \epsilon = 1\); (c) single-hump while \(p_1 = 1-\i/\sqrt{2}, C_1 = 1, C_2 = 0, \xi_{1,0} = 0, \epsilon = 1\).}\label{figure:bright_N=2} 
\end{figure}
In \cref{figure:bright_N=3}, we show two examples of the third-order bright soliton with \(N=3\). It can be seen it involves a collision between one soliton of \(N=1\) and one soliton of \(N=2\) which are discussed previously. It is observed that it is the soliton of \(N=1\) keeps its shape after the collision, but the soliton of \(N=2\) can change it type, for example, from an oscillated one to a double-hump one (see \cref{fig:SSF_N=3(b)}).  

We can further explore the fourth-order bright soliton which involves a collision between two solitons of \(N=4\). In 
\cref{fig:SSF_N=4(a),fig:SSF_N=4(b)}, we can observe the collision between two oscillated solitons, and the one between a single-hump and a double-hump solitons. The collisions are elastic without the change of shape. However, in \cref{fig:SSF_N=4(c),fig:SSF_N=4(d)}, we do observe the change of shape in one of the solitons due to the collision (a change of an oscillated soliton to a single-hump in \cref{fig:SSF_N=4(c)}, while a change of double-hump to an oscillated one in \cref{fig:SSF_N=4(d)}).
\begin{figure}
    \centering
    \subfigure[]{\label{fig:SSF_N=3(a)}
        \includegraphics[width=60mm]{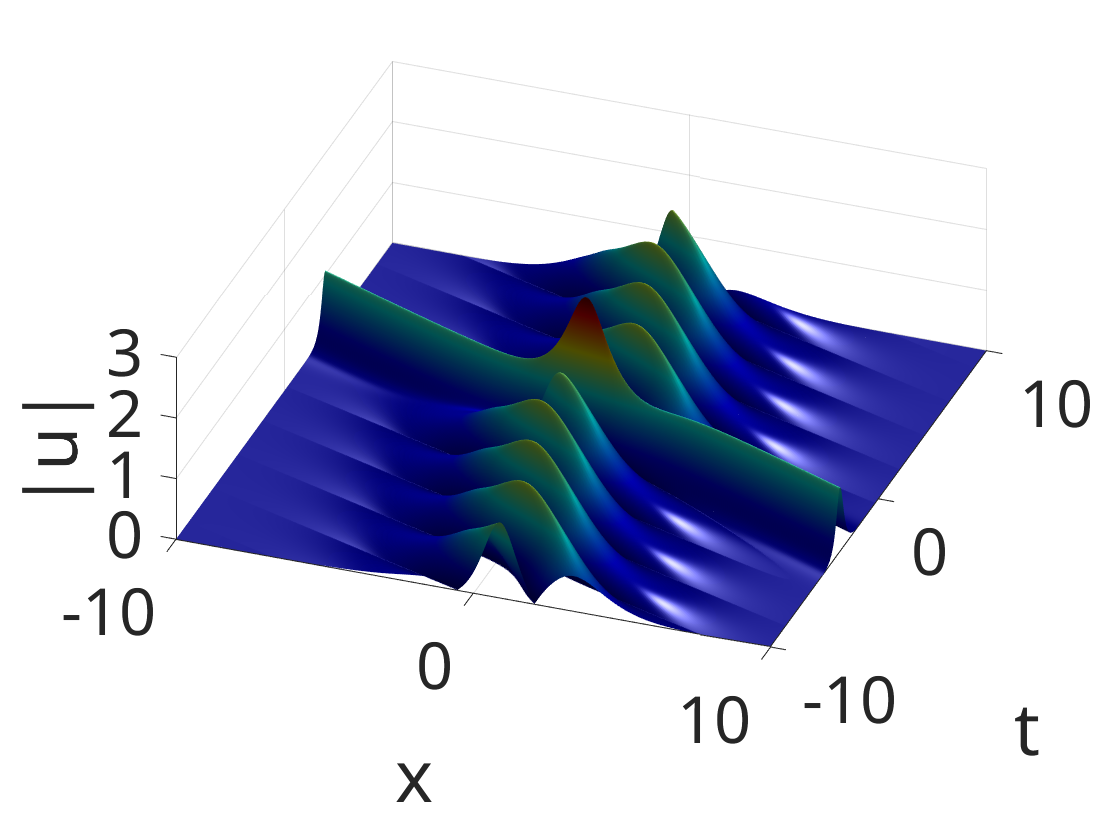}
    }
     \subfigure[]{\label{fig:SSF_N=3(b)}
        \includegraphics[width=60mm]{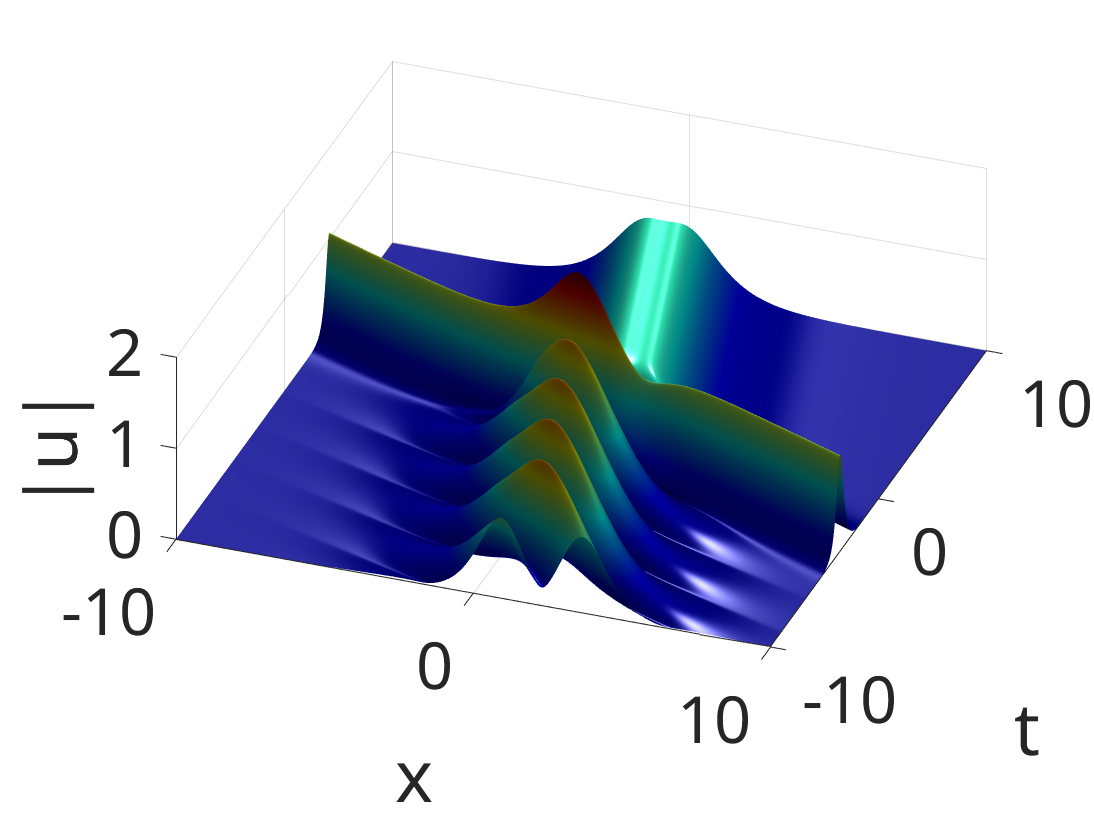}
    }
    \caption{Third-order bright soliton solutions for \(N=3\); (a) collision between oscillated soliton and one bright soliton of \(N=1\) with \(p_1=1+\i, p_2 = 1, C_1 = C_2 = C_3 = 1, \xi_{10} = \xi_{20} = 0, \epsilon = -1\); (b) collision between oscillated soliton and of \(N=1\) with \(p_1=1+\i, p_2 = 1, C_1 = C_2 = 1, C_3 = 0, \xi_{10} = \xi_{20} = 0, \epsilon = -1\).}\label{figure:bright_N=3} 
\end{figure}

\begin{figure}
    \centering
        \subfigure[]{\label{fig:SSF_N=4(a)}
        \includegraphics[width=60mm]{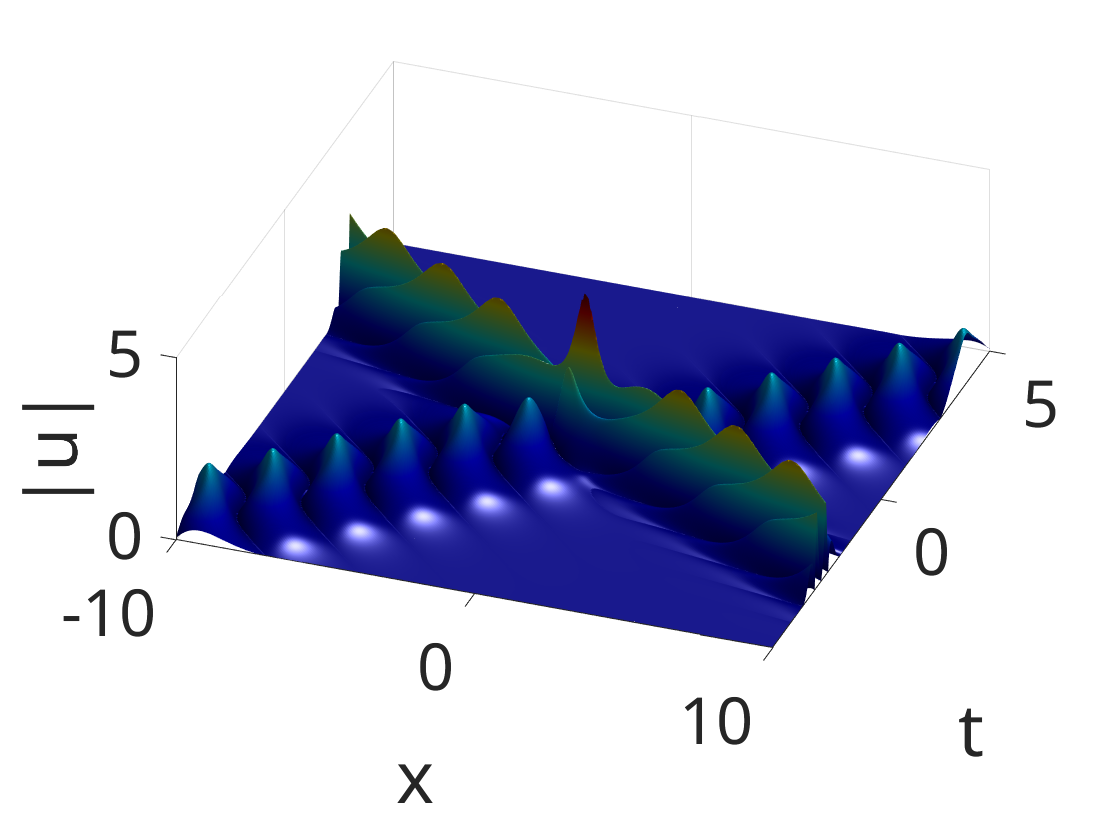}
    }
          \subfigure[]{\label{fig:SSF_N=4(b)}
        \includegraphics[width=60mm]{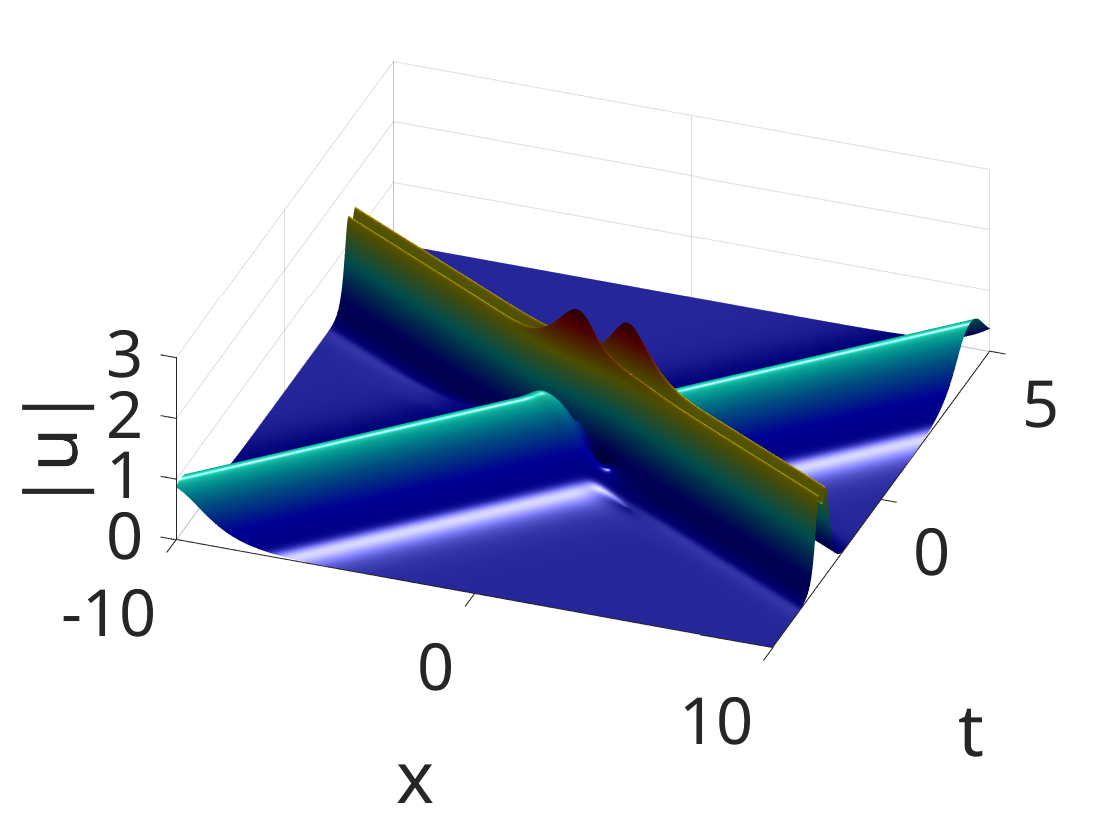}
    }
    \subfigure[]{\label{fig:SSF_N=4(c)}
        \includegraphics[width=60mm]{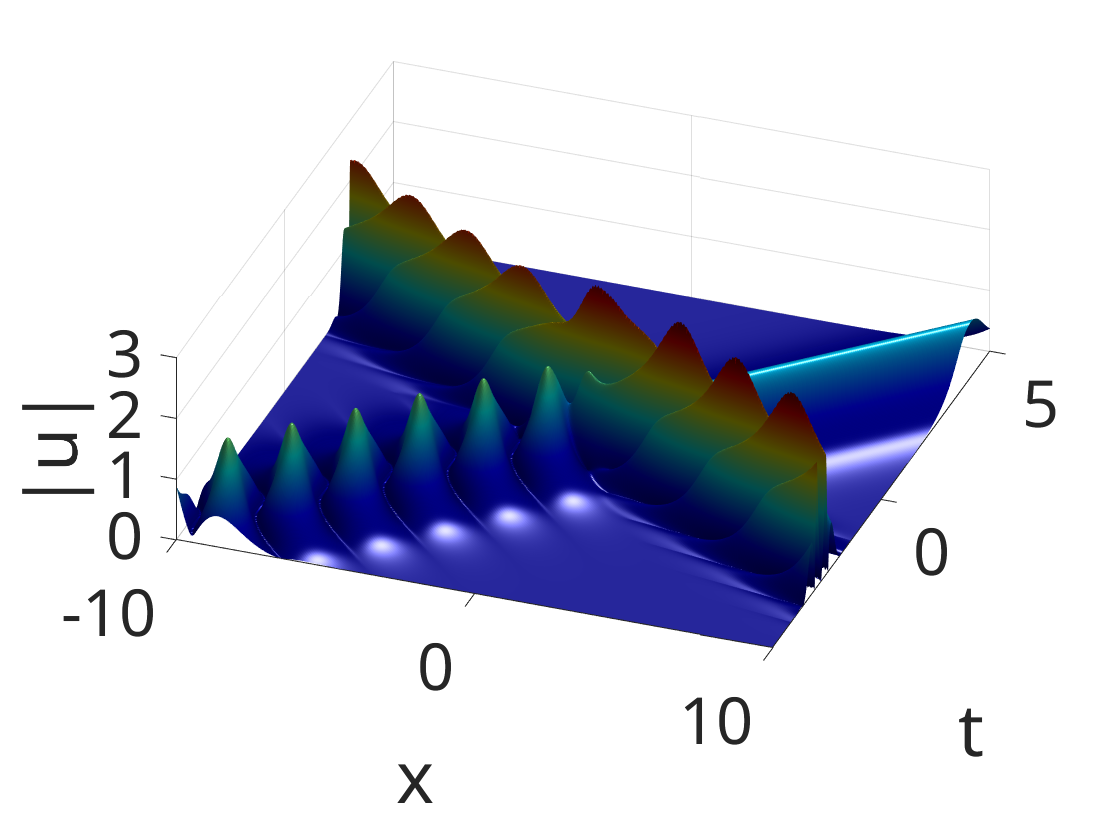}
    }
     \subfigure[]{\label{fig:SSF_N=4(d)}
        \includegraphics[width=60mm]{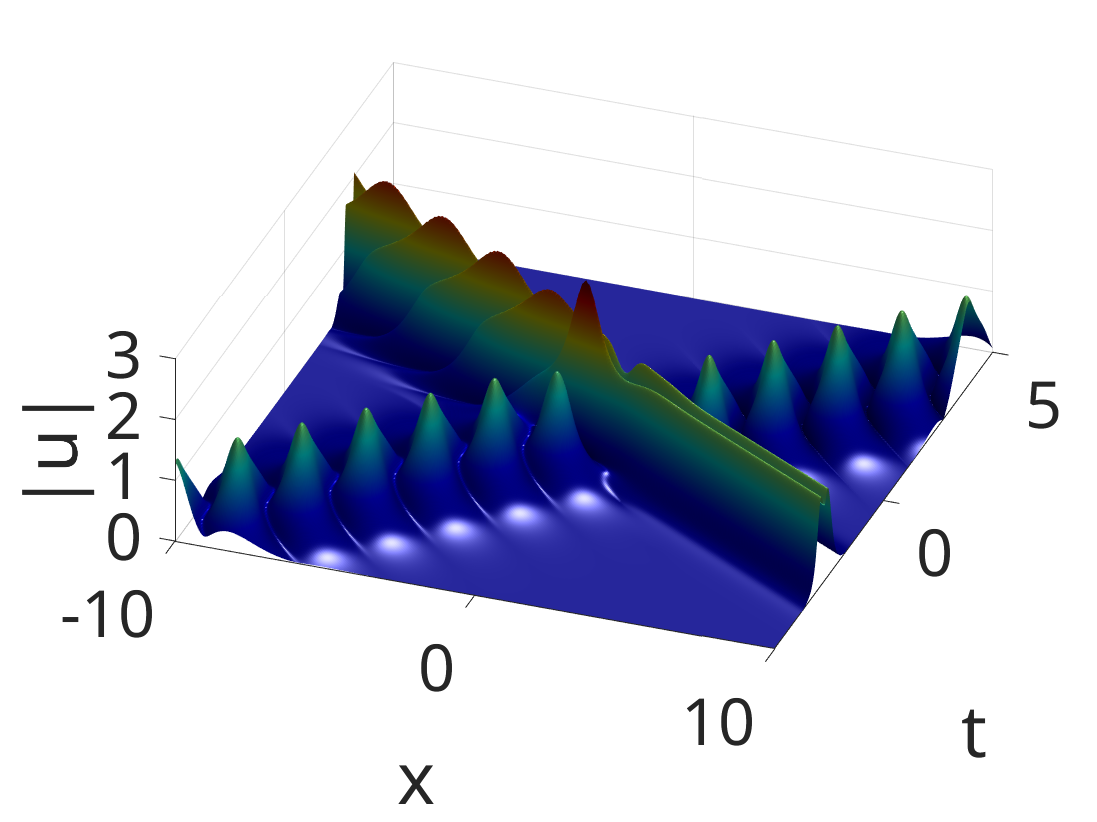}
    }
     \caption{Fourth-order bright soliton solutions for \(N=4\); (a) collision between oscillated solitons with \(p_1=2+\i, p_2 = 1-\i, C_1 = C_2 = C_3 = C_4 = 1, \xi_{10} = \xi_{20} = 0, \epsilon = -1\); (b)  collision between double-hump and single-hump with \(p_1=2+\i, p_2 = 1-\i, C_1 = C_2 = 1, C_3 = C_4 = 0, \xi_{10} = \xi_{20} = 0, \epsilon = -1\);
     (c) collision between oscillated solitons with shape change with \( p_1=2+\i, p_2 = 1-\i, C_1 = C_2 = C_4 = 1 , C_3 = 0, \xi_{10} = \xi_{20} = 0, \epsilon = -1\); (d) collision between oscillated soliton and two-hump with shape change \(p_1=2+\i, p_2 = 1-\i, C_1 = C_2 = C_3 = 1 , C_4 = 0, \xi_{10} = \xi_{20} = 0, \epsilon = -1\); .}\label{figure:bright_N=4} 
\end{figure}

\section{Dark soliton solution to the Sasa-Satsuma equation}\label{section:d_sasa}
To provide dark soliton solution to the SS equation \eqref{sse}, we first give the following lemma.
\begin{lemma}
    For solution determined by \cref{theorem:dd_solution}, reduction condition \(u_1^* = u_2\) is archived under the following parameter restrictions
    \begin{align}
        \varepsilon = \varepsilon_1 = \varepsilon_2, \quad \rho = \rho_1 = \rho_2, \quad \alpha = \alpha_1 = -\alpha_2, \\
        p_i = q_{N+1-i},\quad d_i = d_{N+1-i}, \quad \xi_{i,0} = \xi_{N+1-i,0}, & \quad \text{for} \quad i \in I_1(h),\\
        p_i,q_i \in \mathbb{R},& \quad \text{for} \quad i \in I_2(h) \cap I_3(h).
    \end{align}
    where \(I_1\), \(I_2\), \(I_3\) are index sets defined in \eqref{index_sets}.
\end{lemma}
\begin{proof}
    For \(i,j \in I_1(h)\), one can deduce that \(p_i = p_{N+1-i}^* = q_i^* = q_{N+1-i}\), and 
    \begin{align*}
        \left(\mathfrak{m}_{ij}^{k,l}\right)^* &= \delta_{ij}d_i e^{-\xi_i^*-\eta_j^*} +\frac{1}{p_i^*+q_i^*} \left(-\frac{p_i^*+\i \alpha_1}{q_j^*-\i \alpha_1}\right)^k \left(-\frac{p_i^*+\i \alpha_2}{q_j^*-\i \alpha_2}\right)^l,\\
        &= \delta_{ij}d_i e^{-\xi_{N+1-i}-\eta_{N+1-j}} + \frac{1}{p_{N+1-i}+q_{N+1-j}} \left(-\frac{p_{N+1-i}-\i \alpha_2}{q_{N+1-j}+\i \alpha_2}\right)^k \left(-\frac{p_{N+1-i}-\i \alpha_1}{q_{N+1-j}+\i \alpha_1}\right)^l \\
        &= \mathfrak{m}_{N+1-i,N+1-j}^{l,k}.
    \end{align*}
    For \(i,j \in I_2(h) \cup I_3(h)\), we have \(p_i,q_i \in \mathbb{R}\), thus
    \begin{align*}
        \left(\mathfrak{m}_{ij}^{k,l}\right)^* &= \delta_{ij}d_i e^{-\xi_i-\eta_j} +\frac{1}{p_i+q_i}\left(-\frac{p_i + \i \alpha_1}{q_j - \i \alpha_1}\right)^k\left(-\frac{p_i + \i \alpha_2}{q_j - \i \alpha_2}\right)^l \\
        &= \delta_{ij}d_i e^{-\xi_i-\eta_j} +\frac{1}{p_i+q_i}\left(-\frac{p_i - \i \alpha_2}{q_j + \i \alpha_2}\right)^k\left(-\frac{p_i - \i \alpha_1}{q_j + \i \alpha_1}\right)^l = \mathfrak{m}_{i,j}^{l,k},
    \end{align*}
    Moreover, for \(i \in I_1(h),j \in I_2(h) \cup I_3(h)\), we have \(\left(\mathfrak{m}_{ij}^{k,l}\right)^* = \mathfrak{m}_{N+1-i,j}^{l,k}\), \(\left(\mathfrak{m}_{i,N+1-j}^{k,l}\right)^* = \mathfrak{m}_{N+1-i,N+1-j}^{l,k}\), and for \(i \in I_2(h) \cup I_3(h),j \in I_1(h)\), we have \(\left(\mathfrak{m}_{ij}^{k,l}\right)^* = \mathfrak{m}_{i,N+1-j}^{l,k}\), \(\left(\mathfrak{m}_{N+1-i,j}^{k,l}\right)^* = \mathfrak{m}_{N+1-i,N+1-j}^{l,k}\).
    For above determinant \(\tau_{k,l}^*\), let us consider the following row operations: for \(i \in I_1(h)\), switching \(i\)-th row with \((N+1-i)\)-th row, \(i\)-th column with \((N+1-i)\)-th column. Therefore,
    \begin{align*}
        \tau_{k,l}^* = \det\left(\left(\mathfrak{m}_{ij}^{k,l}\right)^*\right) = \det\left(\mathfrak{m}_{ij}^{l,k}\right) = \tau_{l,k}.
    \end{align*}
    Overall, we have \(\tau_{k,l}^* = \tau_{l,k}\) and subsequently \(g_1^* = g_2\) and \(f^* = f\) or \(f \in \mathbb{R}\), hence \(u_1^* = u_2\).
\end{proof}
Hence, the $N$-dark soliton solution to the SS equation is summarized  by the following theorem.
\begin{theorem}\label{theorem:dark_ss}
    Equation \eqref{sse} admits the dark soliton solutions given by
    \begin{equation}
        \begin{split}
            u&=\rho \frac{g}{f} \exp\left(\i \left(\alpha x - \left(\alpha^3 + 6\varepsilon\alpha \rho^2\right) t\right)\right),\\
        \end{split}
    \end{equation}
    and \(f,\ g\) are defined as
    \begin{equation}
        f=\tau_{0,0}, \quad g=\tau_{1,0},
    \end{equation}
    where \(\tau_{k,l}\) is a \(N\times N\) determinant defined as
    \begin{equation}
        \tau_{k,l}=\det\left(\delta_{ij}d_i e^{-\xi_i-\eta_j} + \frac{1}{p_i+q_j} \left(-\frac{p_i - \i \alpha}{q_j + \i \alpha}\right)^k\left(-\frac{p_i + \i \alpha}{q_j - \i \alpha}\right)^l\right),
    \end{equation}
    with \(\xi_i=p_i (x - 6\varepsilon\rho^2 t) + p_i^3 t + \xi_{i 0}\), \(\eta_i=q_i (x - 6\varepsilon\rho^2 t) + q_i^3 t + \xi_{i 0}\). Where \(\alpha,\ \rho\) are real parameters. And parameters \(d_i,\ \xi_{i,0},\ p_i,\ q_i\) satisfy the following complex conjugate relation for each \(h = 0, 1, \ldots, \lfloor N/2 \rfloor\)
    \begin{align}
        \begin{split}
            d_i = d_{N+1-i},\quad \xi_{i,0} = \xi_{N+1-i,0}, & \quad \text{for } i = 1, 2, \ldots, N,\\
            p_i = p_{N+1-i}^* = q_i^* = q_{N+1-i},& \quad \text{for } i \in \{\mathbb{Z} | 1 \leq i \leq h\},\\
            p_i = q_{N+1-i} \in \mathbb{R},\quad p_{N+1-i} = q_i \in \mathbb{R},& \quad \text{for } i \in  \{\mathbb{Z} | h+1 \leq i \leq \lceil N/2 \rceil\},
        \end{split}
    \end{align}
    Moreover, these parameters also need to satisfy the equation \(G(p_i,q_i) = 0\), for \(i = 1, 2, \ldots, N\), where \(G(p,q)\) defined as
    \begin{align}\label{ss_cond}
        G(p,q)&=\frac{\varepsilon\rho^2}{(p_i-\i \alpha)(q_i+\i \alpha)}+\frac{\varepsilon\rho^2}{(p_i+\i \alpha)(q_i - \i \alpha)}-1.
    \end{align}
\end{theorem}

\subsection{Dynamics of the dark soliton solution to the Sasa-Satsuma equation} 
When \(N=1\), we have the first order dark soliton solution 
\begin{align*}
    u &= \rho \exp(\i \theta)\frac{d_1 \exp(-2\xi_1) - \dfrac{1}{2p_1}\left(\dfrac{p_1 - \i \alpha}{p_1 + \i \alpha}\right)}{d_1 \exp(-2\xi_1) + 1/(2p_1)} \\
    &= \frac{\rho \exp(\i \theta)}{p_1^* + \i \alpha} \biggl(\i \Bigl(\alpha - \Im(p_1)\Bigr) - \Re(p_1)\tanh\Bigl(\Re(\xi_1) - \log\left(4 d_1 \Re(p_1)\right)\Bigr) \biggr)
\end{align*}
where \(\theta = \alpha x - \left(\alpha^3 + 6\varepsilon\alpha \rho^2\right) t\), \(\xi_1 = p_1 (x - 6\varepsilon\rho^2 t) + p_1^3 t + \xi_{1 0}\), \(p_1, \xi_{10}, d_1 \in \mathbb{R}\) and they satisfies the following condition
\begin{equation}
    \frac{2\varepsilon\rho^2}{(p_1-\i \alpha)(p_1+\i \alpha)}=1.
\end{equation}
Above condition implies the dark soliton exists only for $\varepsilon>0$ which gives \(p_1 = \sqrt{2 \varepsilon \rho^2 - \alpha^2}\). In other words, the wave number \(p_1\) is not a free parameter but determined by  \(\varepsilon,\alpha\) and \(\rho\).  It is noted that we have the similar result for the coupled SS equation \cite{zhang2025dark}.

When \(N=2\), the second order solution allows two different choices of parameter \(p_1, p_2\) depending on the value of \(h\)
\begin{align}
    &(h=0 \ \text{case}) &p_1=q_2, p_2=q_1 \in \mathbb{R},\quad \xi_{10}, \xi_{20} \in \mathbb{R}, \label{SS_dark_N=2_complex_1}\\
    &(h=1 \ \text{case}) &p_1 = p_2^* = q_1^* = q_2,\quad \xi_{10} = \xi_{20}^*, \quad d_2 = d_1 \in \mathbb{R}. \label{SS_dark_N=2_complex_2}
\end{align}
Recall, for each set of solutions, we require \(G(p_1,q_1=p_2)=0\) and \(G(p_1,q_1=p_1^*) = 0\) for case \(h=0\) and \(h=1\) respectively.
\begin{itemize}
    \item For case \(h=0\),  we have \(u = \rho_1 \exp(\i \theta) g_2 /f \) with
    \begin{align*}
        \theta ={}& \alpha x - \left(\alpha^3 + 6\varepsilon \alpha \rho_1^2 \right)t,\\
        g ={}& d_1^2 \exp(-2 \xi_1 -2 \xi_2)  - \frac{p_1 - \i \alpha}{p_1 + \i \alpha}\frac{p_2 - \i \alpha}{p_2 + \i \alpha} \frac{\left(p_1-p_2\right)^2}{4 p_1 p_2 \left(p_1+p_2\right)^2} \nonumber\\
        & - \frac{d_1 \exp(- \xi_1 - \xi_2) (p_1^2 + p_2^2 + 2\alpha^2)}{(p_1 + p_2)(p_1 + \i \alpha)(p_2 + \i \alpha)}, \\
        f ={}& d_1^2 \exp(-2 \xi_1 -2 \xi_2) -\frac{\left(p_1-p_2\right)^2}{4 p_1 p_2 \left(p_1+p_2\right)^2} + \frac{2d_1 \exp(- \xi_1 - \xi_2)}{p_1 + p_2}. 
    \end{align*}
    Similar to the bright soliton, there are five roots for \(\partial |u|^2/\partial_x  = 0\) by denoting \(y = \exp(\xi_1+\xi_1^*)\)
    \begin{align*}
        &y_1 = 0,
        \quad 
        y_2 = -\frac{2 \i d_1 \sqrt{p_1} \sqrt{p_2} \left(p_1+p_2\right)}{|p_1-p_2|},
        \quad 
        y_3 = \frac{2 \i d_1 \sqrt{p_1} \sqrt{p_2} \left(p_1+p_2\right)}{|p_1-p_2|}, \\
        &y_4 = -2d_1 \left(p_1+p_2\right) \frac{\left(p_1^2+p_2^2\right) +\left(p_1+p_2\right) \sqrt{ \left(p_1^2-p_2 p_1+p_2^2\right)}}{\left(p_1-p_2\right){}^2},\\
        &y_5 = -2d_1 \left(p_1+p_2\right) \frac{\left(p_1^2+p_2^2\right) -\left(p_1+p_2\right) \sqrt{\left(p_1^2-p_2 p_1+p_2^2\right)}}{\left(p_1-p_2\right){}^2}.
    \end{align*}
We have the following cases regarding to the solution
    \begin{enumerate}
        \item \(p_1 p_2 > 0\):  dark soliton with the extreme value occurs at \(y_5\)  
        \begin{align*}
            |u(y_5)|^2 = \rho_1^2 \frac{\alpha _1^2 \left(\alpha _1^2+p_1^2-p_2 p_1+p_2^2\right)}{\left(\alpha _1^2+p_1^2\right) \left(\alpha _1^2+p_2^2\right)} < \rho_1^2.
        \end{align*}

        \item  \(p_1 p_2 < 0\) and \(d_1 (p_1 + p_2) < 0\):  Mexican hat soliton. In this case, \(0 < y_4 < y_3 < y_5\), the solution exhibits three extreme points.
        The extreme value at point \(y_4\) or \(y_5\) determines the Mexican hat to be upwards or downwards.
        \begin{align*}
            |u(y_3)|^2 &= \rho_1^2 \left(1 + \frac{\left(p_1^2-p_2^2\right){}^2 \left(\alpha _1^2+p_1 p_2\right)}{4 p_1 p_2 \left(\alpha _1^2+p_1^2\right) \left(\alpha _1^2+p_2^2\right)}\right),\\
            |u(y_4)|^2 &= |u(y_5)|^2 = \rho_1^2 \frac{\alpha _1^2 \left(\alpha _1^2+p_1^2-p_2 p_1+p_2^2\right)}{\left(\alpha _1^2+p_1^2\right) \left(\alpha _1^2+p_2^2\right)}.
        \end{align*}
        If \(\alpha^2+p_1 p_2 > 0\), \(|u(y_4)|^2  > \rho_1^2 > |u(y_3)|^2\), one obtains an anti-Mexican hat soliton illustrated in  \cref{fig:SSD_N=2_h=0_amh(a)}. 
        
        If \(\alpha^2+p_1 p_2 < 0\), \(|u(y_4)|^2  < \rho_1^2 < |u(y_3)|^2\), one has Mexican hat soliton shown in \cref{fig:SSD_N=2_h=0_mh(a)}. 

        \item  \(p_1 p_2 < 0\) and \(d_1 (p_1 + p_2) > 0\): dark or anti-dark soliton. Since
        \begin{align*}
            |u(y_2)|^2 = \rho_1^2 \left(1 + \frac{\left(p_1^2-p_2^2\right){}^2 \left(\alpha _1^2+p_1 p_2\right)}{4 p_1 p_2 \left(\alpha _1^2+p_1^2\right) \left(\alpha _1^2+p_2^2\right)}\right),
        \end{align*}
        we have dark soliton (\(|u(y_3)|^2 < \rho_1^2\))  if \(\alpha^2+p_1 p_2 > 0\), see \cref{fig:SSD_N=2_h=0_dark(b)} and anti-dark soliton  (\(|u(y_3)|^2 > \rho_1^2\)) If \(\alpha^2+p_1 p_2 < 0\) as shown in  \cref{fig:SSD_N=2_h=0_antidark(b)}.
    \end{enumerate}

    \begin{figure}[!ht]
        \centering
        \subfigure[]{\label{fig:SSD_N=2_h=0_amh(a)}
            \includegraphics[width=50mm]{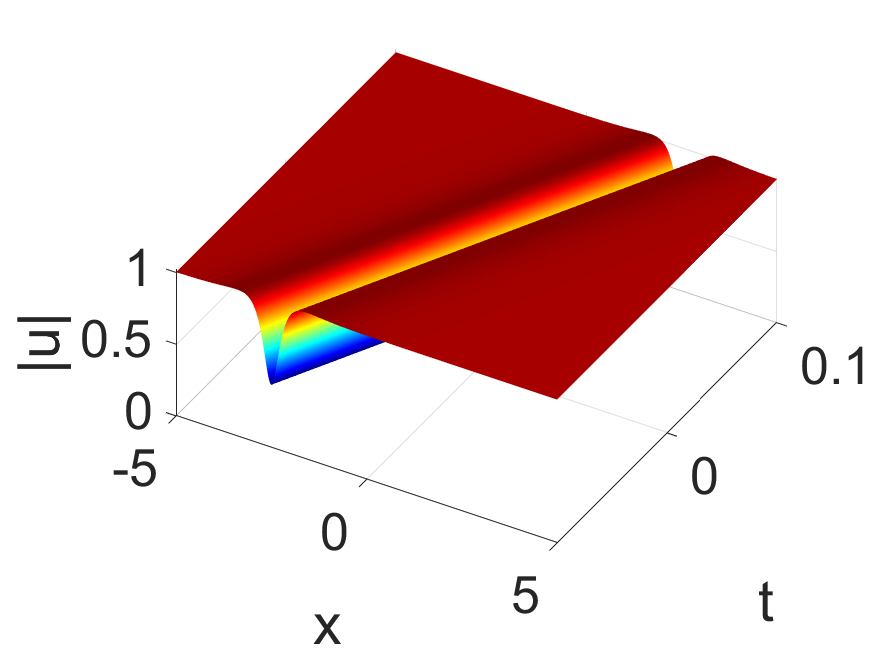}
        }
        \subfigure[]{\label{fig:SSD_N=2_h=0_dark(b)}
            \includegraphics[width=50mm]{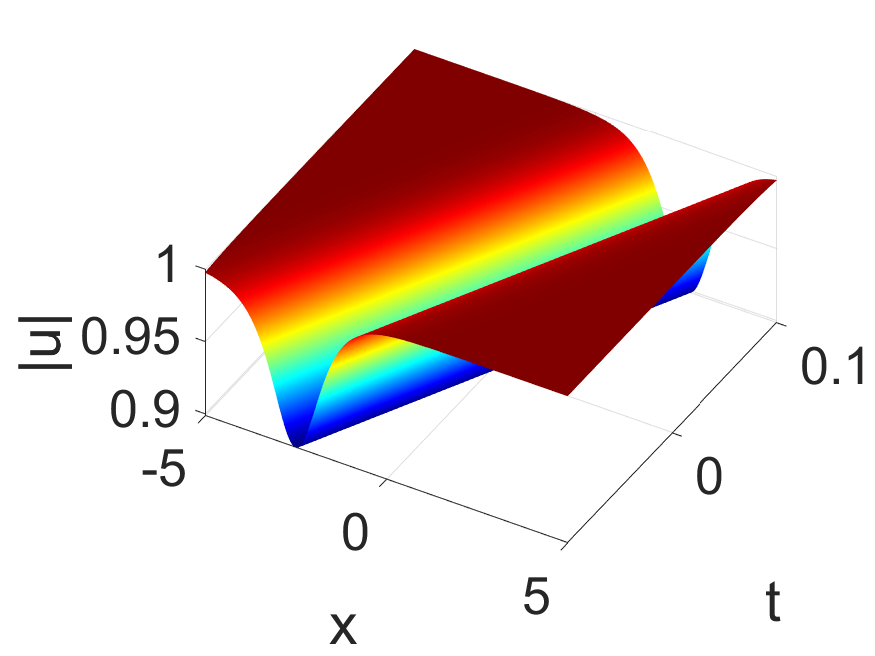}
        }
        \subfigure[]{
            \includegraphics[width=50mm]{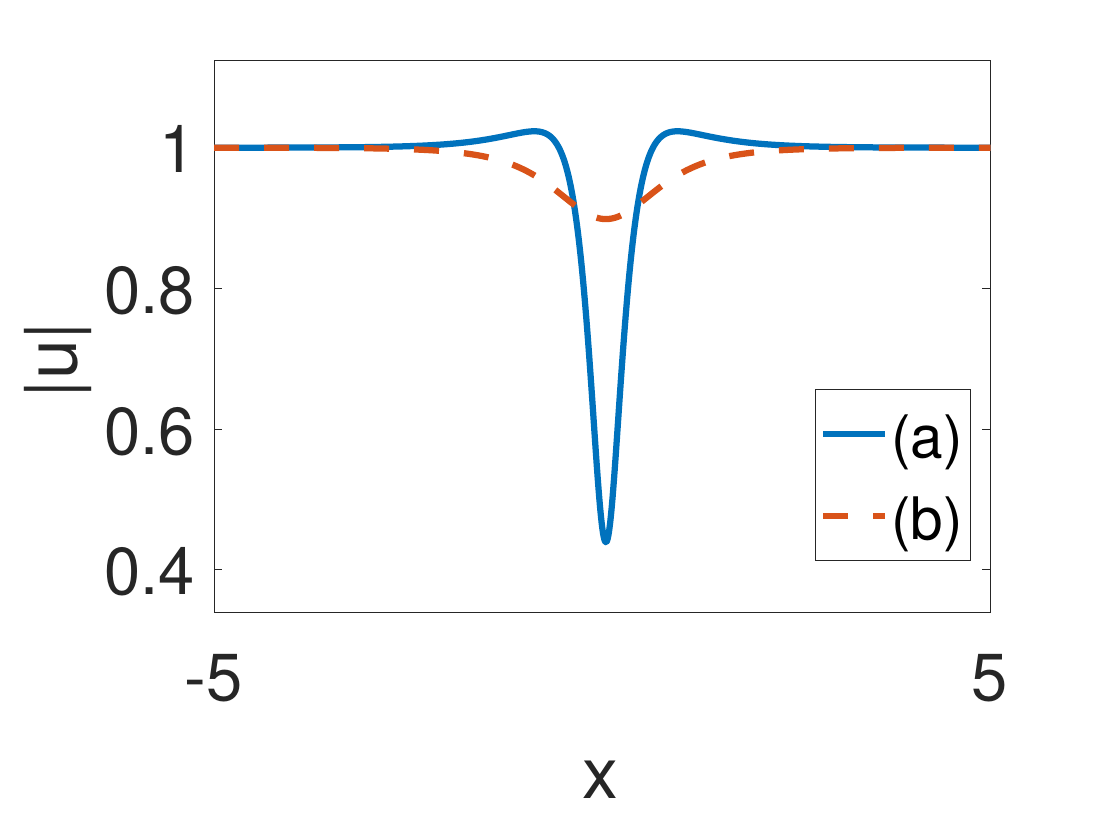}
        }
        \caption{Anti-Mexican hat and dark soliton solutions to the SS equation with parameters (a) \(p_1= 2, p_2 = 2 - \sqrt{5}, d_1 = -1, \alpha = 1, \rho_1 = 1, \varepsilon = 5, \xi_{1,0} = 0\), (b) \(p_1= 2, p_2 = 2 - \sqrt{5}, d_1 = 1, \alpha = 1, \rho_1 = 1, \varepsilon = 5, \xi_{1,0} = 0\), (c) comparison of two cases when \(t = 0\).}\label{fig:SSD_N=2_h=0_mh}
    \end{figure}

    \begin{figure}[!ht]
        \centering
        \subfigure[]{\label{fig:SSD_N=2_h=0_mh(a)}
            \includegraphics[width=50mm]{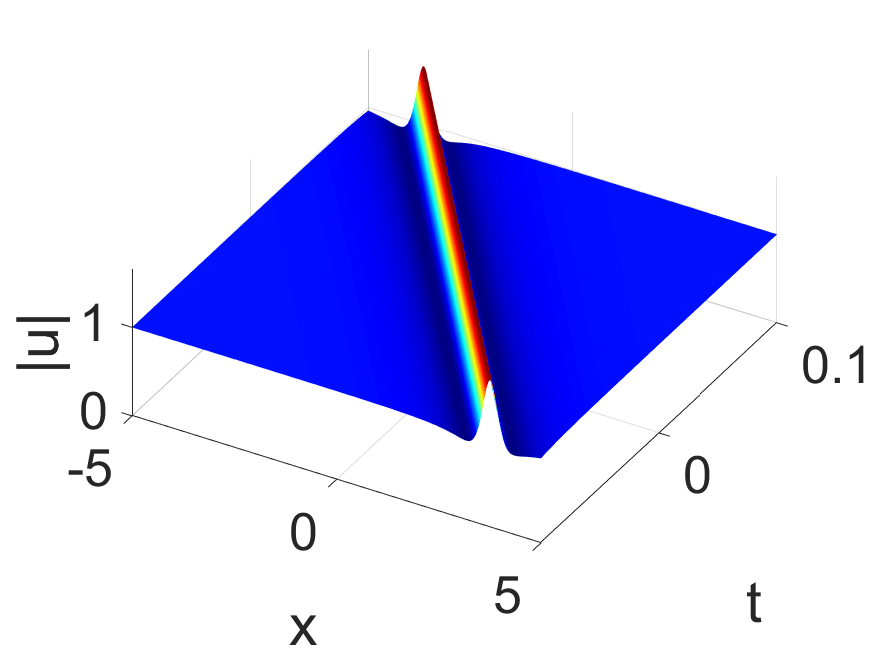}
        }
        \subfigure[]{\label{fig:SSD_N=2_h=0_antidark(b)}
            \includegraphics[width=50mm]{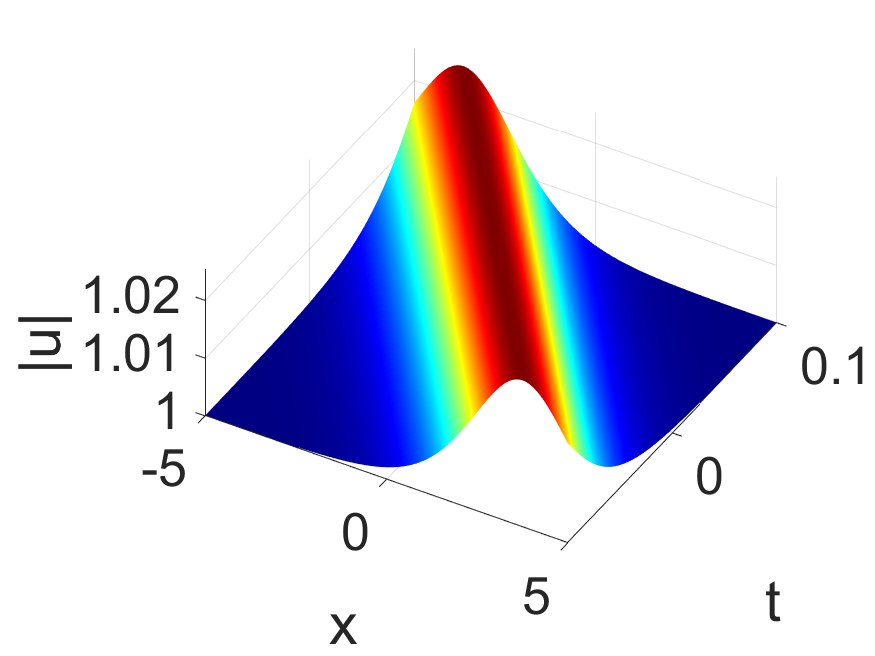}
        }
        \subfigure[]{
            \includegraphics[width=50mm]{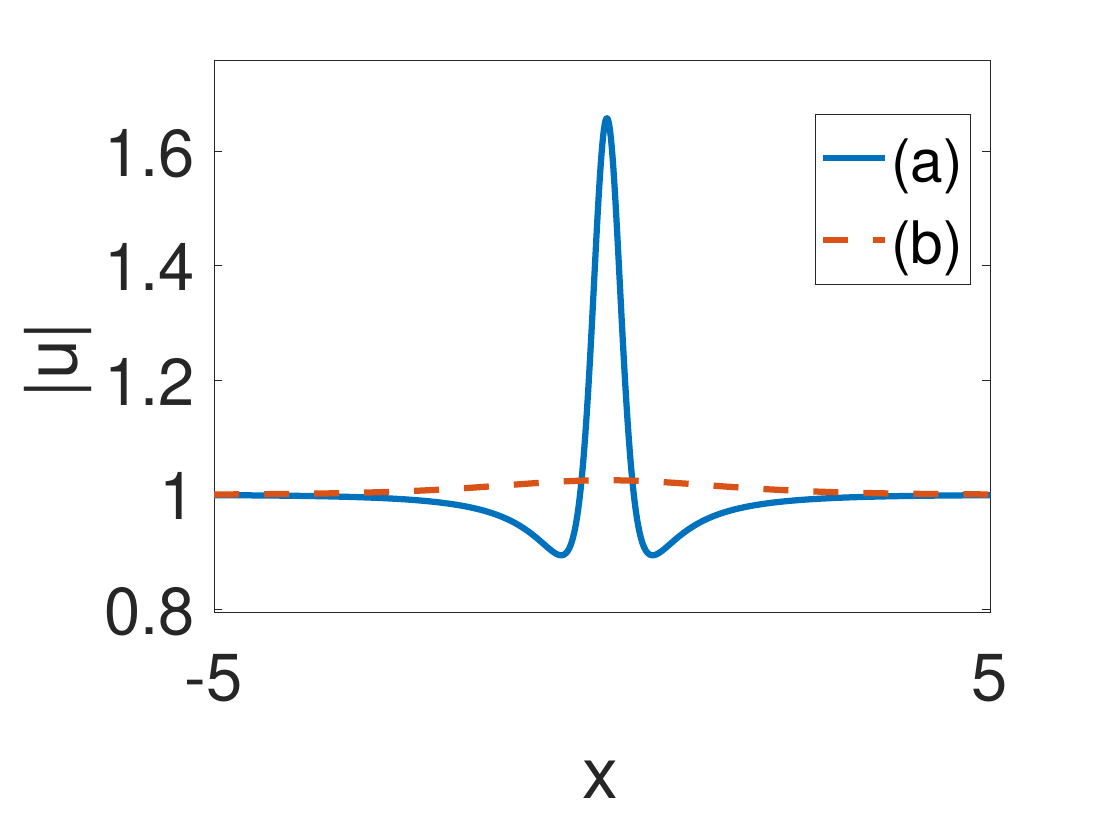}
        }
        \caption{Mexican hat and anti-dark soliton solutions to the SS equation with parameters (a) \(p_1= 1, p_2 = -2, d_1 = 1, \alpha = 1, \rho_1 = 1, \varepsilon = -5\), (b) \(p_1= 1, p_2 = -2, d_1 = -1, \alpha = 1, \rho_1 = 1, \varepsilon = -5\), (c) comparison of two cases when \(t = 0\).}\label{fig:SSD_N=2_h=0_d_ad}
    \end{figure}

    \item For case \(h=1\), we have 
    \begin{align*}
        \theta ={}& \alpha x - \left(\alpha^3 + 12\varepsilon \alpha \rho_1^2 \right)t,\\
        g ={}& d_1^2 \exp(-2 \Re(\xi_1))  + \frac{\left[\Im(p_1)\right]^2 (p_1 - \i \alpha)(p_1^* - \i \alpha)}{4 \left[\Re(p_1)\right]^2 |p_1| (p_1 + \i \alpha)(p_1^* + \i \alpha)}\nonumber\\
        & - \frac{ d_1 \exp(- \Re(\xi_1)) (\left[\Im(p_1)\right]^2 - \left[\Re(p_1)\right]^2 + \alpha^2)}{\Re(p_1)(p_1 + \i \alpha)(p_1^* + \i \alpha)}, \\
        f ={}& d_1^2 \exp(-2 \Re(\xi_1)) + \frac{\left[\Im(p_1)\right]^2}{4 \left[\Re(p_1)\right]^2 |p_1|^2} + \frac{2d_1 \exp(- \Re(\xi_1))}{2 \Re(p_1)}. 
    \end{align*}
In this case, it can be shown that we only have nonsingular soliton when \(\varepsilon > (\Im(p_1)-\alpha)^2(\Im(p_1)+\alpha)^2/(2\rho^2(\left[\Im(p_1)\right]^2+\alpha^2)) > 0\).
By denoting \(y = \exp(\xi_1 + \xi_1^*), p_1 = a + \i b\), there are again five roots of \(\partial |u|^2/\partial_x  = 0\)
    \begin{align*}
        &y_1 = 0,\quad y_2 = -2 a d_1 \frac{\sqrt{a^2+b^2}}{b},\quad y_3 = 2 a d_1\frac{\sqrt{a^2+b^2}}{b}, \\
        &y_4 = 2 a d_1\frac{\left(a^2-b^2\right)- a\sqrt{a^2-3 b^2}}{b^2},\quad
        y_5 = 2 a d_1\frac{\left(a^2-b^2\right)+ a\sqrt{a^2-3 b^2}}{b^2}.
    \end{align*}
    Let us assume \(a d_1 > 0\) and \(b > 0\), without loss of generality. Since \(y = \exp(\xi_1 + \xi_1^*) > 0\), we have the following cases
    \begin{enumerate}
        \item  \(a^2 > 3b^2\): double holes soliton as illustrated in \cref{fig:SSD_N=2_h=1_hole(a)}. In this case, there are three extremes \(y_3\), \(y_4\), \(y_5\) where \(0 < y_4 < y_3 < y_5\). Moreover, since
        \begin{align*}
            &|u(y_3)|^2 = \rho_1^2 \frac{\left(\alpha _1^2 \left(b \left(\sqrt{a^2+b^2}+b\right)+a^2\right)+\left(a^2+b^2\right) \left(a^2-b \left(\sqrt{a^2+b^2}+b\right)\right)\right){}^2}{\left(b \left(\sqrt{a^2+b^2}+b\right)+a^2\right)^2 \left(\left(a^2+b^2\right)^2+2 \alpha _1^2 \left(a^2-b^2\right)+\alpha _1^4\right)},\\
            &|u(y_4)|^2 = |u(y_5)|^2 = \rho_1^2 \frac{\alpha _1^2 \left(a^2+\alpha _1^2-3 b^2\right)}{\left(a^2+b^2\right)^2+2 \alpha _1^2 (a-b) (a+b)+\alpha _1^4},
        \end{align*}
       we have \(|u(y_4)|^2 = |u(y_5)|^2 < |u(y_3)|^2 < \rho_1^2\).
        \item \(a^2 \leq 3b^2\): single hole soliton as shown in \cref{fig:SSD_N=2_h=1_hole(b)}. In this case, only one extreme value exists at \(y_3\) and \(|u(y_3)|^2 < \rho_1^2\).
    \end{enumerate}

    \begin{figure}[!ht]
        \centering
        \subfigure[]{\label{fig:SSD_N=2_h=1_hole(a)}
            \includegraphics[width=50mm]{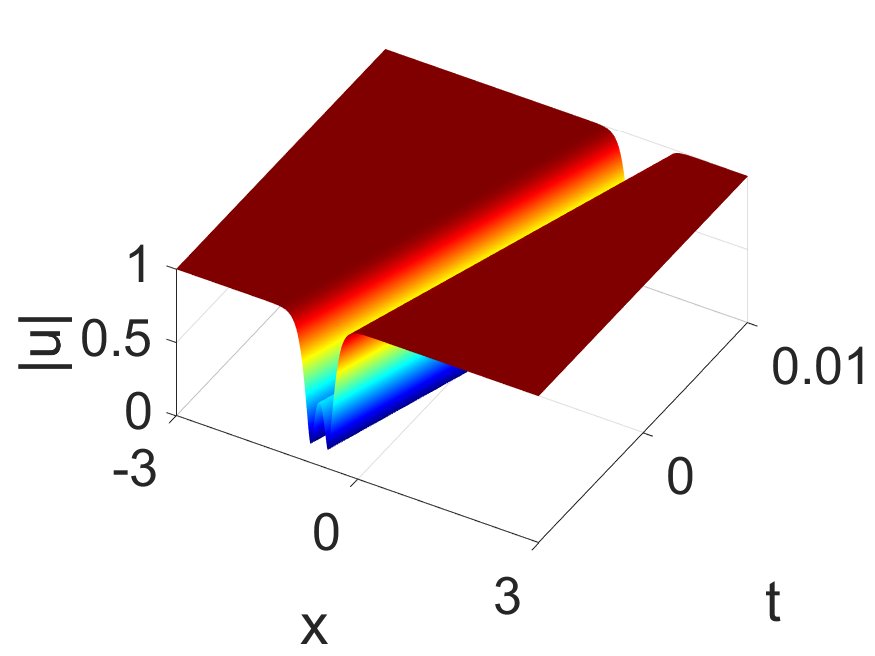}
        }
        \subfigure[]{\label{fig:SSD_N=2_h=1_hole(b)}
            \includegraphics[width=50mm]{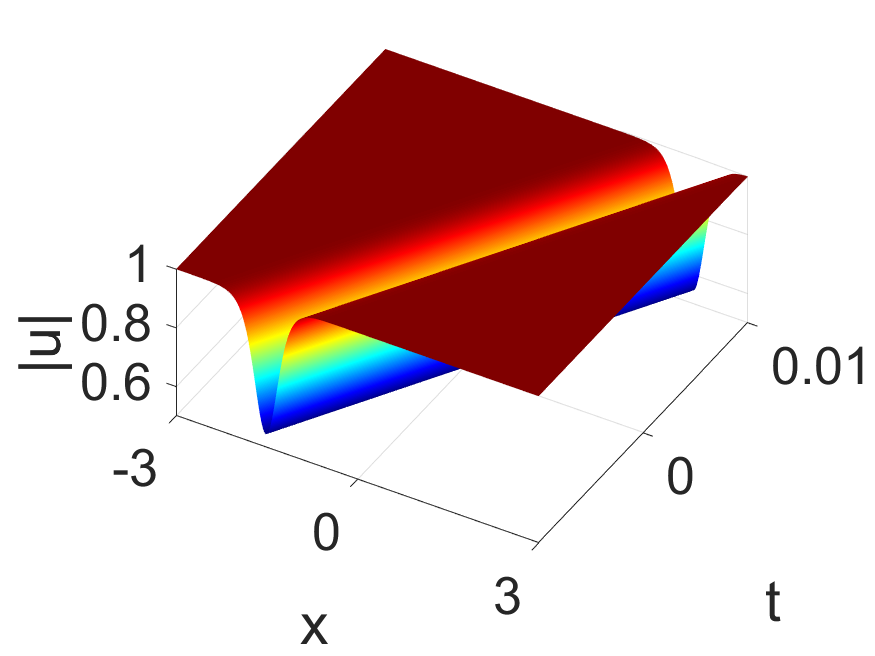}
        }
        \subfigure[]{
            \includegraphics[width=50mm]{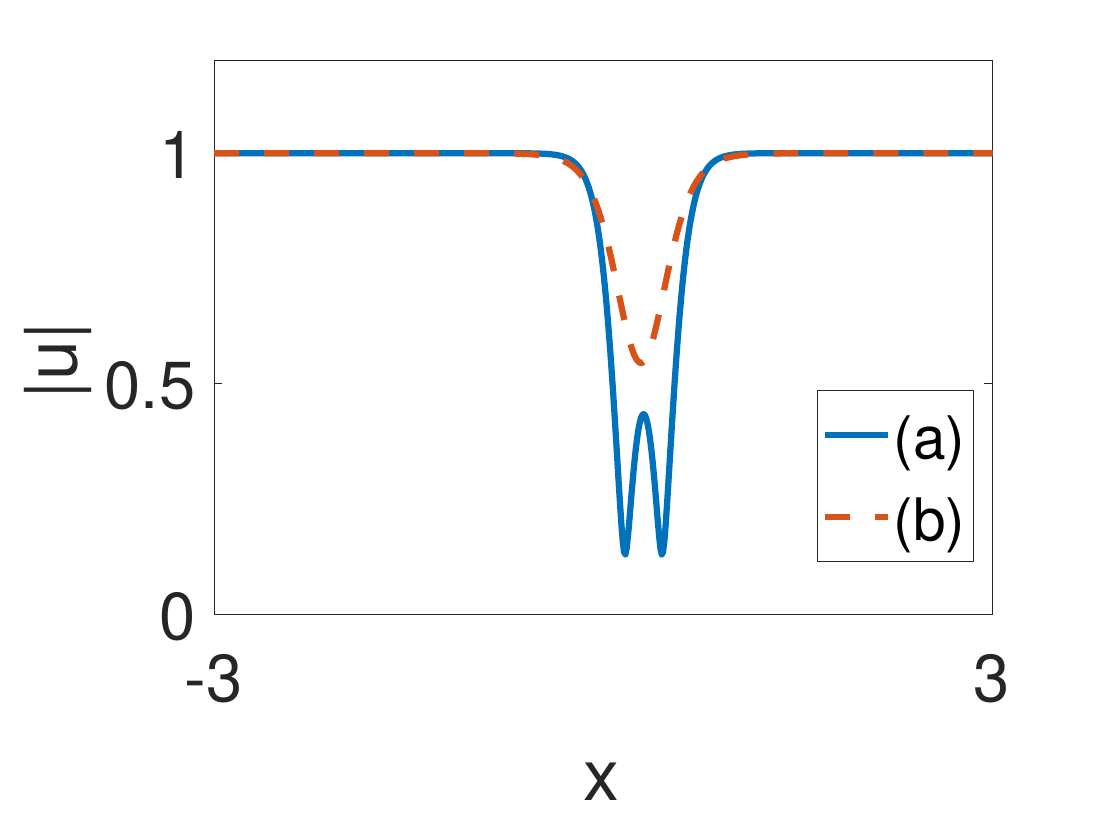}
        }
        \caption{Double-hole soliton and single-hole soliton solution to the SS equation with parameters (a) \(p_1= 6 + \sqrt{-15 + 4 \sqrt{21}} \i, d_1 = 1, \alpha = 1, \rho_1 = 1, \varepsilon = 20, \xi_{1,0} = 0\), (b) \(p_1= 4 + \sqrt{5 + 4 \sqrt{26}} \i, d_1 = 1, \alpha = 1, \rho_1 = 1, \varepsilon = 20, \xi_{1,0} = 0\), (c) comparison of two cases when \(t = 0\).}\label{fig:SSD_N=2_h=1_hole}
    \end{figure}
\end{itemize}

For \(N=3\), one of  the wave numbers \(p_2 = \sqrt{2\epsilon \rho^2 - \alpha^2}\) is fixed and determined by the plane wave background. This third order soliton involves a collision between such a fixed soliton with either ananti-Mexican-hat (\cref{figure:dark_N=3_h=0} or with a double-hole  soliton
(\cref{figure:dark_N=3_h=1}). As far as we know, such kind of soliton hasn't reported elsewhere.

\begin{figure}
    \centering
    \subfigure[]{\label{fig:SSD_N=3(a)}
        \includegraphics[width=50mm]{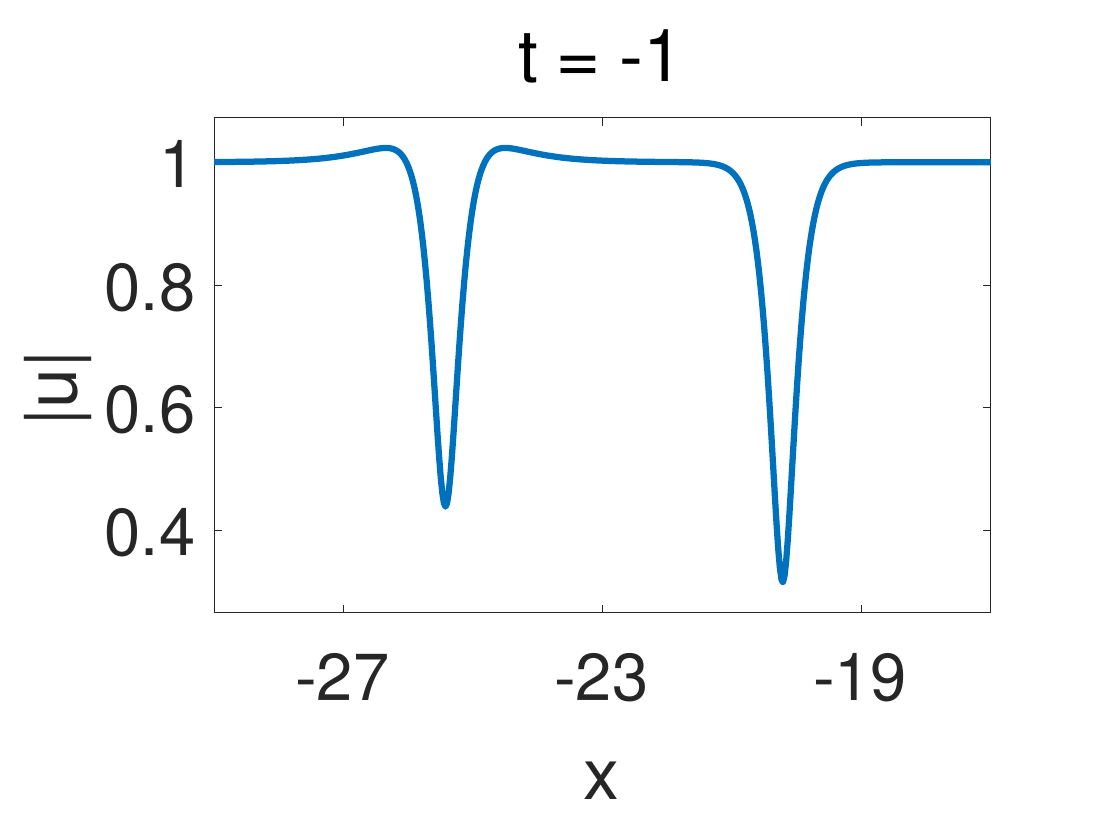}
    }
     \subfigure[]{\label{fig:SSD_N=3(b)}
        \includegraphics[width=50mm]{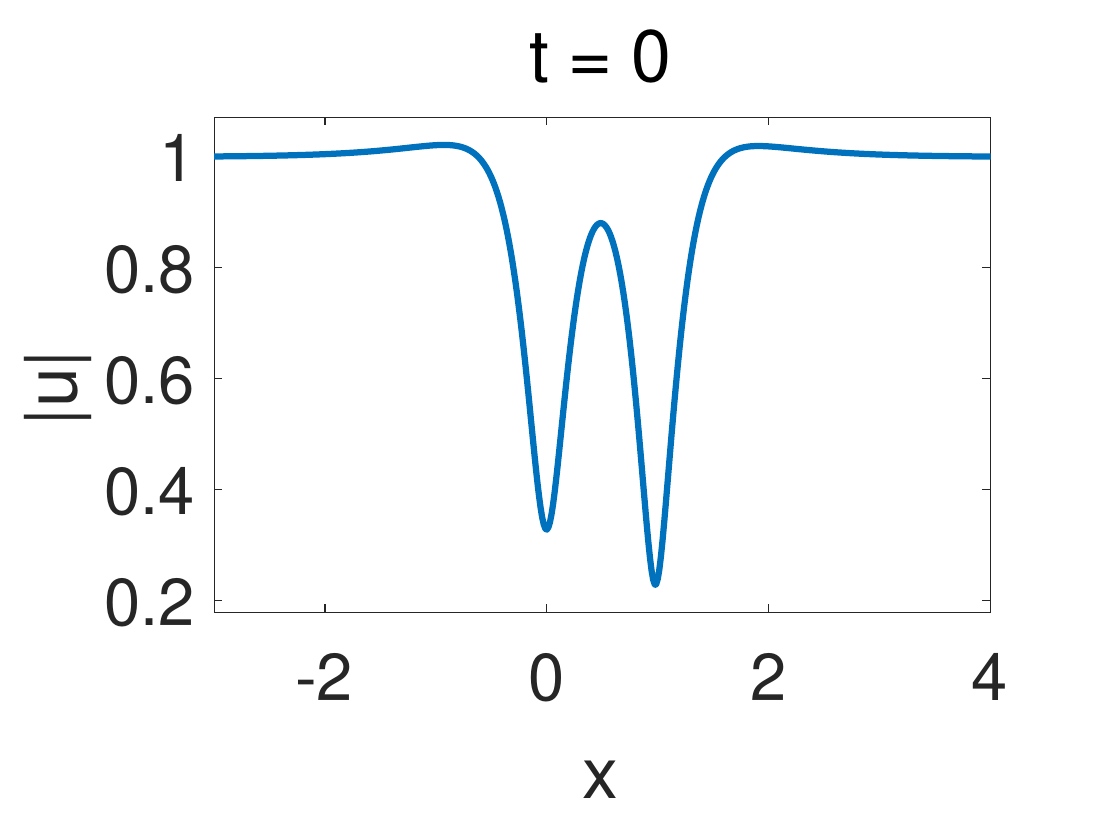}
    }
        \subfigure[]{\label{fig:SSD_N=3(c)}
        \includegraphics[width=50mm]{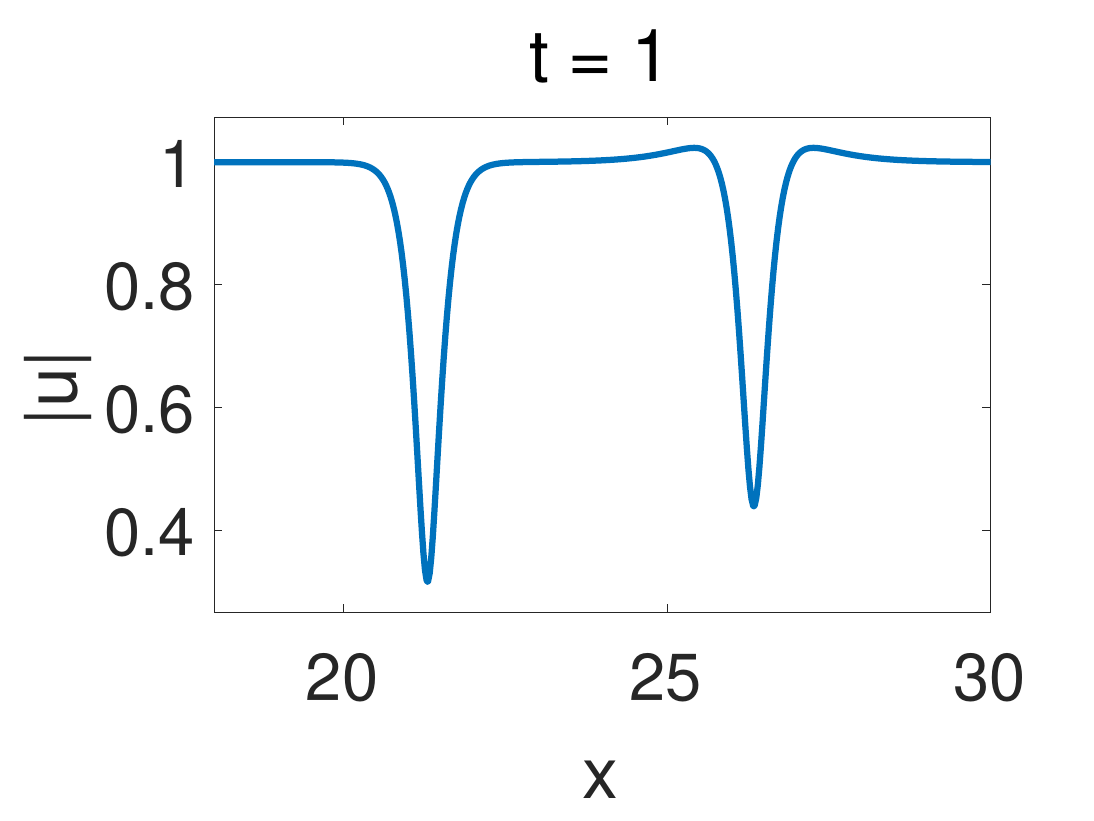}
    }
    \caption{Collision between anti-Mexican-hat and dark soliton while \(N=3, h = 1, p_1 = 2, p_2 = 3, p_3 = 2-\sqrt{5}, d_1 = d_2 = 1, \epsilon = 5\).}\label{figure:dark_N=3_h=0} 
\end{figure}

\begin{figure}
    \subfigure[]{\label{fig:SSD_N=3(d)}
        \includegraphics[width=50mm]{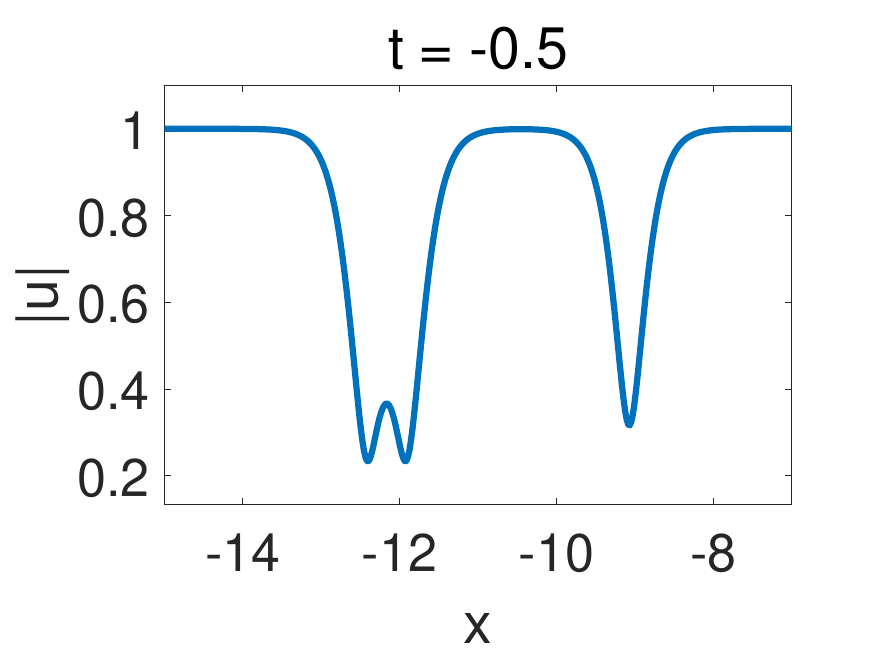}
    }
     \subfigure[]{\label{fig:SSD_N=3(e)}
        \includegraphics[width=50mm]{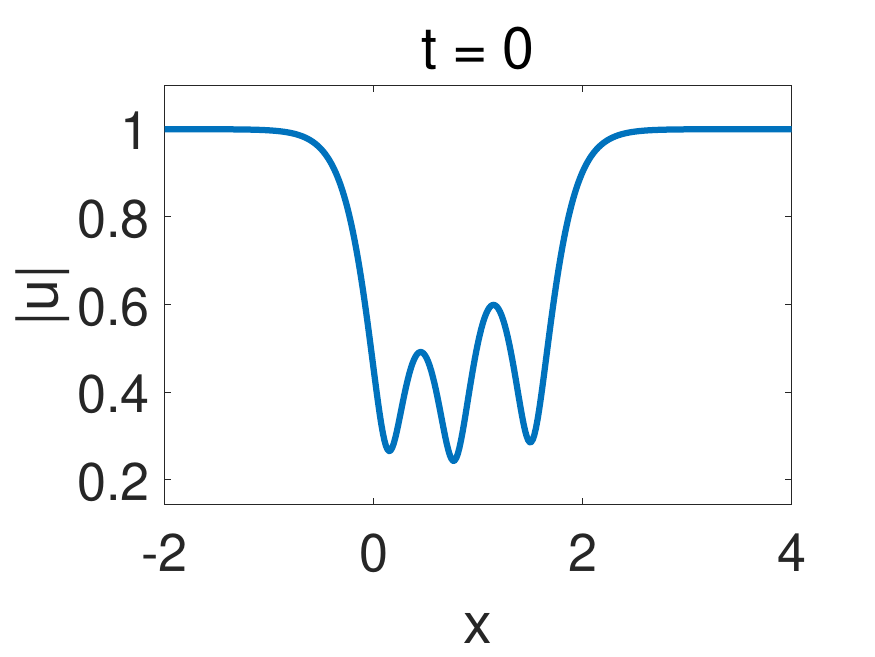}
    }
        \subfigure[]{\label{fig:SSD_N=3(f)}
        \includegraphics[width=50mm]{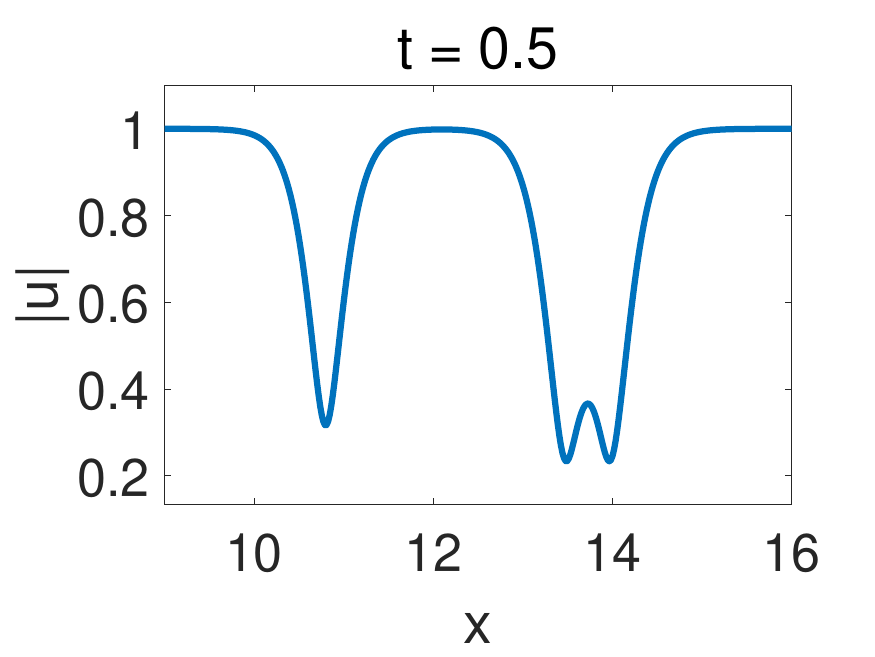}
    }
    \caption{Collision between double-hole and single-hole solitons while \(N =3, h = 1, p_1 = p_3^* = 23/8 + \i/8\sqrt{-145 + 16 \sqrt{191}}, p_2 = 3, d_1 = d_2 = 1, \epsilon = 5\).}\label{figure:dark_N=3_h=1} 
\end{figure}

\section{Conclusion}
In this paper, we derived bright-bright, dark-dark, and bright-dark soliton solutions for the cHirota equation using the KP reduction method. Notably, different bilinear equations and tau functions were introduced to derive the bright-dark soliton solutions. The dynamic behavior of these solutions is analyzed and illustrated for both the first- and second-order cases. For the cHirota equation, regular one- and two-soliton solutions are presented, and a breather solution is obtained as one special case of the second-order dark soliton.

We also derived bright and dark soliton solutions to the SS equation by using a link between the cHirota equation to the SS equation. For the SS equation, we found that the first-order bright soliton solution also satisfies the complex mKdV equation, while the first-order dark soliton solution does not. We observed various soliton solutions for the second-order bright and dark solitons. Specifically, the second-order bright soliton can be oscillating, single-hump or two-hump one. Meanwhile, the second-order dark soliton can be classified into dark (single-hole), anti-dark, Mexican hat, anti-Mexican hat, and double-hole ones. For both the bright or dark soliton, the type can undergo change due to the collision. 
\section*{Acknowledgements}
B.F. Feng's work was partially supported by the U.S. Department of Defense (DoD), Air Force for Scientific Research (AFOSR) under grant No. W911NF2010276.
\bibliography{cite}

\end{document}